\documentclass{arxiv}
\usepackage{thm-restate}
\usepackage{etoolbox}

\newbool{break}
\boolfalse{break}

\newtheoremstyle{std}{}{}{\itshape}{}{\bfseries}{.}{.5em}{\thmname{#1}\thmnumber{ #2}\thmnote{ (#3)}}
\newtheoremstyle{numberless}{}{}{}{}{\bfseries}{.}{.5em}{\thmnote{#3}}

\theoremstyle{std}
\newtheorem{theorem}{Theorem}
\newtheorem{lemma}[theorem]{Lemma}
\newtheorem{proposition}[theorem]{Proposition}

\newtheorem*{claim}{Claim}
\newtheorem{corollary}[theorem]{Corollary}

\theoremstyle{definition}
\newtheorem{remark}[theorem]{Remark}
\newtheorem{definition}[theorem]{Definition}

\theoremstyle{numberless}
\newtheorem*{unnumbered}{}

\AddToHook{env/lemma/begin}{\crefalias{theorem}{lemma}}
\AddToHook{env/proposition/begin}{\crefalias{theorem}{proposition}}
\AddToHook{env/observation/begin}{\crefalias{theorem}{observation}}
\AddToHook{env/corollary/begin}{\crefalias{theorem}{corollary}}
\AddToHook{env/remark/begin}{\crefalias{theorem}{remark}}
\AddToHook{env/definition/begin}{\crefalias{theorem}{definition}}

\crefname{theorem}{Theorem}{Theorems}
\crefname{lemma}{Lemma}{Lemmas}
\crefname{proposition}{Proposition}{Propositions}
\crefname{observation}{Observation}{Observations}
\crefname{claim}{Claim}{Claims}
\crefname{corollary}{Corollary}{Corollaries}
\crefname{remark}{Remark}{Remarks}
\crefname{definition}{Definition}{Definitions}

\title{Efficiently Listing Projected Trees, \\ and Equivalence of Listing and Enumeration}
\author{
    Karl Bringmann\footnote{ETH Zurich, Switzerland, \email{\{karl.bringmann,yanheng.wang\}@inf.ethz.ch}. Part of this work was finished while affiliated to Saarland University and Max Planck Institute for Informatics, Saarbrücken, Germany, where this work was part of the project TIPEA that has received funding from the European Research Council (ERC) under the European Union's Horizon 2020 research and innovation programme (grant agreement No. 850979).} \and
    Nick Fischer\footnote{Max Planck Institute for Informatics, Saarbrücken, Germany, \email{nfischer@mpi-inf.mpg.de}.} \and
    Yanheng Wang\textsuperscript{\textasteriskcentered}
}

\begin{document}
\maketitle

\begin{abstract}
\noindent
The subgraph isomorphism problem and its generalizations, such as conjunctive queries where some nodes are projected, are among the most fundamental problems in graph algorithms and database theory. In this paper, we study the listing and enumeration variants of these problems and present two main results.

The first result is an algorithm for enumerating projected trees with preprocessing time $\widetilde{O}(n^{17.42})$ and delay $\mathrm{polylog}(n)$. Prior to this work, for trees on $k$ nodes all algorithms in the literature required preprocessing time $n^{\Omega(k)}$ or delay $n^{\Omega(1)}$ or assumed $\omega=2$. Our result generalizes to arbitrary projected hypergraphs, achieving enumeration in preprocessing time $\widetilde{O}(m^{17.42 \, \mathrm{subw}(H)})$ and polylogarithmic delay, where $\mathrm{subw}(H)$ is the submodular width of the pattern hypergraph $H$. We heavily rely on fast (rectangular and output-sensitive) matrix multiplication, which we complement by fine-grained lower bounds indicating that any algorithm beating preprocessing time $n^{\Omega(k)}$ with polylogarithmic delay must rely on fast matrix multiplication.

The second result is a generic enumeration-to-listing reduction, establishing that listing and enumeration are equivalent under natural assumptions. For (colored) subgraph isomorphism, our reduction transforms any listing algorithm running in time $O(f(n,m) + t \cdot g(n,m))$ into an enumeration algorithm with preprocessing time $O\left( (f(n,m)+g(n,m)+n+m) \log^2 n \right)$ and delay $O(g(n,m))$. We utilize this reduction to prove our first main result, and we expect that our generic reduction will find many future applications.
\end{abstract}

\thispagestyle{empty}
\setcounter{page}{0}
\clearpage

\tableofcontents
\thispagestyle{empty}
\setcounter{page}{0}
\clearpage

\newcommand{\ALIGN}{\ifbool{break}{&}{}}
\newcommand{\AND}{\ifbool{break}{,\\}{\quad\text{and}\quad}}
\newcommand{\NEWLINE}{\ifbool{break}{\\}{}}
\newcommand{\MATH}[1]{\ifbool{break}{\begin{align*} #1 \end{align*}}{$#1$}}

\newcommand{\tOh}{\widetilde{O}}
\newcommand{\eq}{\operatorname{eq}}
\newcommand{\type}{\operatorname{type}}
\newcommand{\bag}{\operatorname{bag}}
\newcommand{\tw}{\operatorname{tw}}
\newcommand{\fctw}{\operatorname{fc-tw}}
\newcommand{\subw}{\operatorname{subw}}
\newcommand{\fcsubw}{\operatorname{fc-subw}}
\newcommand{\rows}{\operatorname{rows}}
\newcommand{\sol}{\mathcal{S}}
\newcommand{\mm}{\mathrm{MM}}
\newcommand{\alfa}{{\widetilde{\alpha}}}
\newcommand{\indicate}{\ifbool{break}{\mathbf{1}}{\mathbb{1}}}
\newcommand{\tup}[1]{\mathbf{#1}}
\newcommand{\clique}[1]{\mathrm{Clique}(#1)}

\newcommand{\algolist}[1]{\Call{list}{#1}}
\newcommand{\algorecover}[1]{\Call{recover}{#1}}
\newcommand{\algomult}[1]{\Call{multiply}{#1}}
\newcommand{\algolight}[1]{\Call{light}{#1}}
\newcommand{\algoreach}[1]{\Call{reach}{#1}}
\newcommand{\algocount}[1]{\Call{count}{#1}}
\newcommand{\algodcount}[1]{\Call{dcount}{#1}}
\newcommand{\algocontract}[1]{\Call{contractions}{#1}}
\newcommand{\algosplit}[1]{\Call{split}{#1}}
\newcommand{\algosearch}[1]{\Call{search}{#1}}
\newcommand{\algotrivlist}[1]{\Call{trivialList}{#1}}
\newcommand{\algobudgetlist}[1]{\Call{budgetedList}{#1}}
\newcommand{\algoparbudgetlist}[1]{\Call{pBudgetedList}{#1}}
\newcommand{\algocheatenum}[1]{\Call{cheatEnum}{#1}}
\newcommand{\algoparcheatenum}[1]{\Call{pCheatEnum}{#1}}

\newenvironment{pseudocode}{
\par\medskip\begin{algorithmic}
}{\end{algorithmic}\medskip}

\section{Introduction}

The subgraph isomorphism problem and its variants are among the most fundamental graph problems, with a large number of applications in various areas such as databases, network motifs, statistical physics, and probabilistic inference. Within theoretical computer science, it is primarily studied by two communities: algorithm design and database theory. Although closely related, these communities adopt slightly different viewpoints, and our results apply to both.

The algorithms community typically studies the \emph{Uncolored $H$-Subgraph} problem (e.g.~\cite{CN85,Kowalik03,ABKZ22,ABF23,JX23,AKLS23,DMVX24,JVZ24,VW25,BG25}). Here, $H$ is a fixed pattern graph, the input is a large host graph $G$, and any copy of $H$ in $G$ is considered a solution. In contrast, the database theory community usually studies the \emph{Colored $H$-Subgraph} and \emph{Colored $H$-Subhypergraph} problems, also known as (self-join-free) \emph{join queries} (e.g.,~\cite{Yannakakis81,NRR13,Veldhuizen14,NPRR18,JR18,Durand20,Hu25,KHS25}). Here, writing $V(H) =[k] = \set{1,\dots,k}$, the given host \makebox{(hyper)}graph is equipped with a partition $V(G) = V_1 \cup \cdots \cup V_k$, and a solution is an $H$-copy that respects this partition.
See \Cref{sec:preliminaries} for formal problem definitions.
While the color-coding technique~\cite{AYZ95} provides a reduction from the Uncolored to the Colored problem, in general these variants are not equivalent. (An exception is the triangle, the most intensely studied pattern~\cite{IR78,Patrascu10,BPVZ14,KPP16,VX20}, for which Colored and Uncolored are equivalent.)

In this paper we focus on the well-studied generalization of these problems where some nodes in $H$ are \emph{projected}. Formally, let $I \subseteq V(H)$. In the Colored/Uncolored \emph{$(H,I)$-Subgraph} problem, a solution is a tuple of vertices $(v_i)_{i \in I}$ in $G$ that can be extended to a Colored/Uncolored copy of $H$. The accordingly defined Colored \emph{$(H,I)$-Subhypergraph} problem is also known as \emph{conjunctive queries} or \emph{join-project queries} and constitutes one of the most important problems in database theory, because it formalizes the core functionality of SQL queries (e.g.,~\cite{BDG07,KNS17,BGS20,CK20,DHK20,DZFK24,DHK25,KNS25,KNS26}); see, e.g.,~\cite{Marx13,DRW19} for work on this problem in the algorithm theory community. 

For deciding whether a solution exists, there is no difference between join queries and conjunctive queries, as projections are equivalent to existential quantifiers. The problems differ only for more complex problem variants, e.g., asking to count~\cite{Yannakakis81,DRW19}, list, or enumerate solutions. In this paper, we study listing and enumeration. Formally, for any of the problems~$\mathcal{P}$ defined before, in $\mathcal{P}$-Listing the task is to compute the set of all solutions, and we bound the running time in terms of the total number of solutions $t$. In $\mathcal{P}$-Enumeration, the task is also to compute the set of all solutions, but we bound the time until the first solution is printed (called preprocessing time) and the maximum time between printing two consecutive solutions (called delay). The time after printing the last solution can also be assumed to be bounded by the delay.
Both listing (e.g.~\cite{CN85,AKLS23,JX23,ABKZ22,ABF23,DMVX24,JVZ24,BG25,VW25}) and enumeration (e.g.~\cite{KR95,Uno97,Strozecki19,HHMW20,CS21,JX23,MM24,BG25}) are widely studied in the algorithm theory community, and the same holds for the database community.

Instead of $(H,I)$-Subhypergraph Listing we write $(H,I)$-Listing for short; similarly for Enumeration. We study Colored and Uncolored $(H,I)$-Listing and Enumeration. Unless further specified, we consider the colored variant, which is algorithmically harder.
Throughout, the pattern $(H,I)$ is fixed, and we focus on the complexity in terms of the number of vertices $n$ and the number of edges~$m$ of the host graph $G$. That is, $O$-notation suppresses factors of the form $f(k)$, or more generally any factors depending only on~$(H,I)$; this is called ``data complexity'' in the database community. We will also use $\tOh$-notation to hide factors of the form $\polylog(n)$, or equivalently $\polylog(m)$.

We present two main results, which we discuss in more detail below:

\begin{enumerate}[1.]
\item When $H$ is a tree, we show that Colored and Uncolored $(H,I)$-Listing are in time $\tOh(n^{17.42} + t)$. The best previously known running time was $\Omega(n^{\Omega(k)} + t)$ or $t \cdot n^{\Omega(1)}$, unless one assumes $\omega=2$.
For general hypergraphs~$H$, we show that Colored and Uncolored $(H,I)$-Listing are in time $\tOh(m^{17.42 \cdot \subw(H)} + t)$, where $\subw(H)$ is the submodular width of $H$.

\item Any enumeration algorithm with preprocessing time $T_P$ and delay $T_D$ yields a listing algorithm running in time $O(T_P + (t+1) \cdot T_D)$, by simply running the enumeration algorithm until it has printed all $t$ solutions and terminates. We provide the first generic converse to this statement, an \emph{enumeration-to-listing reduction}, for a class of problems that includes Colored $(H,I)$-Subhypergraph. This tool converts the previous item into an enumeration algorithm with $\tOh(m^{17.42 \cdot \subw(H)})$ preprocessing time and $\tOh(1)$ delay.
\end{enumerate}

\subsection{Polynomial Preprocessing for Projected Trees} \label{sec:intro:trees}
Let us start by discussing the case where $H$ is a tree.
Without projections, Yannakakis' seminal work~\cite{Yannakakis81} provides an algorithm for Colored $H$-Listing in linear time $O(m+t)$, where $m$ is the number of edges of the host graph and $t$ is the number of solutions. Variants of this algorithm also solve Colored $H$-Enumeration with preprocessing time $O(m)$ and delay $O(1)$, and count the number of trees in time $O(m)$~\cite{Yannakakis81,BDG07}.
These algorithms can be converted to the uncolored setting by color coding~\cite{AYZ95}, at the cost of replacing $O$ with $\tOh$ (and increasing the hidden dependence on $k$ from polynomial to exponential), see \Cref{sec:uncoloredalgo} for details.

With projections, Brault-Baron~\cite{BraultBaron13} classified all projected trees $(H,I)$ with linear complexity: If $(H,I)$ is \emph{free-connex}, i.e., the induced subgraph~$H[I]$ is connected, then preprocessing time $O(m)$ and delay $O(1)$ is possible; otherwise preprocessing time $\tOh(m)$ and delay $\tOh(1)$ is impossible.\footnote{The `otherwise' direction is a fine-grained lower bound assuming a believable hypothesis on detection of hypercliques~\cite{BraultBaron13}; alternatively it can be proved assuming the Zero-Clique Hypothesis~\cite{BringmannCM25}.}
The same classification also applies to listing, with a time separation of $O(m+t)$ versus non-$\tOh(m+t)$.

Hence, for general projected trees near-linear time $\tOh(m+t)$ is impossible, and instead we relax the goal and ask for the smallest $c = c(H,I)$ such that $(H,I)$-Listing is in time $\tOh(m^c + t)$. As our first main result, we prove that $c \le 17.42$.

\begin{restatable}{theorem}{MainOneTheorem}
\label{thm:main1}
For every tree $H$ and every $I$, Colored/Uncolored $(H,I)$-Listing is in time $\tOh(n^{17.42} + t)$ and Colored/Uncolored $(H,I)$-Enumeration is in preprocessing time $\tOh(n^{17.42})$ and delay $\tOh(1)$. If $n \times n$ matrix multiplication is in time $\tOh(n^2)$, then the $n^{17.42}$ term improves to $n^3$. For the Colored problems, the hidden constant in the $\tOh$-notation is polynomial in $k$.
\end{restatable}

This extends the classic results for listing and enumerating trees~\cite{Yannakakis81,BDG07,AYZ95} to projected trees, at the cost of relaxing the preprocessing time from $O(m)$ to $\tOh(n^{17.42})$ and the delay from $O(1)$ to $\tOh(1)$.
Using $n \le O(m)$, which holds after removing isolated vertices, the $n^{17.42}$ term can also be replaced by $m^{17.42}$.

We further improve the dependence on $m$ for the special case where $H$ is a \emph{star} and the center node is projected away. We call this pattern \emph{projected-$(k-1)$-star}; formally it is defined as $(H,I)$ with
\begin{align*}
    V(H) \ALIGN= \set{0} \cup [k-1],\quad\NEWLINE
    E(H) \ALIGN= \Set{\set{0,i} : i \in [k-1]},\quad\NEWLINE
    I \ALIGN= [k-1].
\end{align*}

\begin{theorem}
    \label{thm:star-listing}
    Colored/Uncolored projected-$(k-1)$-star-Listing is in time $\tOh(m^{11.61} + t)$, and Colored/Uncolored projected-$(k-1)$-star-Enumeration is in preprocessing time $\tOh(m^{11.61})$ and delay $\tOh(1)$. If $n \times n$ matrix multiplication is in time $\tOh(n^2)$, then the $m^{11.61}$ term improves to $m^2$.
\end{theorem}

The proof of \Cref{thm:star-listing} requires considerable technical effort. We leave open whether it can be generalized from projected stars to all projected trees.

\subsubsection{Comparison with Prior Work} \label{sec:intro:priorwork}
Let us discuss why our result is surprising given the state of the art. For every tree $H$ on $k$ nodes, $(H,I)$-Listing is trivially in time $O(n^k)$ and $O(m^{k-1})$, since the search space size is at most $n^k$ and, slightly less trivially, $m^{k-1}$. To the best of our knowledge, no algorithm in the literature unconditionally beats the trivial time bounds. Specifically, no algorithm for projected-$(k-1)$-star-Listing in the literature unconditionally runs in time $\tOh(n^{k-\eps} + t)$ or $\tOh(m^{k-1-\eps} + t)$, for any $\eps > 0$. In \Cref{sec:subw} we verify this claim on some algorithms whose running times are difficult to calculate. Broadly speaking, we will argue later that any algorithm without using fast matrix multiplication cannot beat the trivial time bounds. Fast matrix multiplication is a frequent tool in algorithm theory~\cite{BPVZ14,DMVX24} and has also been explored in the database community~\cite{DHK20,Hu24,KHS25,DHK25}. However, even these works do not beat the trivial time bounds on the projected-$(k-1)$-star, except for an algorithm by Hu~\cite{Hu24} which solves projected-$(k-1)$-star-Listing in time $\tOh(m \cdot t^{2/3 - 4/(9k - 12)} + t)$ assuming $\omega = 2$.\footnote{In fact, $\omega=2$ does not suffice, one needs to assume that $n \times n$ matrix multiplication can be solved in time $\tOh(n^2)$.} 
In this paper we focus on running times of the form $\tOh(m^c + t)$, so to make it comparable we take a weighted mean\footnote{The inequality $x^\beta y^{1-\beta} \le x+y$ holds for all $x,y \ge 0$ and $\beta \in [0,1]$. We apply it with $\beta = 2/3 - 4/(9k - 12)$, $x = t$ and $y = m^{1/(1-\beta)}$.} to bound Hu's running time by $\tOh(m^{3 - 4/k} + t)$.
Hu's algorithm was the first to beat the trivial time bound---however, it is conditioned on $\omega = 2$, which is conjectured by some researchers but far from proven. Using currently known matrix multiplication algorithms, the algorithm's running time deteriorates to $t \cdot m^{\Omega(1)}$, which is not output-linear.
Summarizing, under current knowledge on matrix multiplication no prior work solves projected-$(k-1)$-star-Listing in time $\tOh(n^{k-\eps} + t)$ or $\tOh(m^{k-1-\eps} + t)$. 

In contrast, our algorithm unconditionally works with the current knowledge on matrix multiplication and its running time is a fixed polynomial in $n$. Moreover, even in an ideal world where matrix multiplication is in time $\tOh(n^2)$, we improve Hu's running time from $\tOh(m^{3 - 4/k} + t)$ to $\tOh(m^2 + t)$ for stars, and from $\tOh(m^6+t)$ to $\tOh(n^3+t)$ for trees, in terms of output-linear algorithms.

\smallskip
Finally, let us compare with the Colored $(H,I)$-Counting problem, where the task is to count the number of solutions. Without projections, a variant of Yannakakis' algorithm~\cite{Yannakakis81} solves this problem in time $O(m)$. With projections, a parameterized complexity classification was shown in~\cite{DRW19}; this includes a conditional lower bound ruling out $O(n^{k-1-\eps})$-time algorithms for Colored projected-$(k-1)$-star-Counting for any $\eps > 0$. That is, projections make counting trees intractable. In contrast, we show that listing and enumerating trees remain tractable with projections.

\subsubsection{Discussion of Drawbacks}

In comparison with the classic results that trees without projections can be listed in time $O(m+t)$ and enumerated in preprocessing time $O(m)$ and delay $O(1)$, our result has three drawbacks, most of which are necessary:

\paragraph{Superlinear Preprocessing Time} 
Unlike the linear preprocessing time $O(m)$ for trees without projections, we need preprocessing time $\tOh(n^{17.42}) \le \tOh(m^{17.42})$, and conditioned on very fast matrix multiplication we still need preprocessing time $\tOh(n^3) \le \tOh(m^3)$. This is a necessary relaxation, because our algorithm works for all projected trees $(H,I)$, in particular also for non-free-connex~$(H,I)$, for which $\tOh(m)$ preprocessing time and $\tOh(1)$ delay is impossible by the known fine-grained classification mentioned above~\cite{BraultBaron13}.

\paragraph{Polylogarithmic Delay} 
While constant delay is possible for trees without projections, we need polylogarithmic delay $\tOh(1)$. At the core, this comes from polylogarithmic factors in the running time of a fast rectangular matrix multiplication algorithm that we are using. We leave it open whether the delay can be improved to $O(1)$ without significantly worsening the preprocessing time.

\paragraph{Fast Matrix Multiplication}
Our algorithms are impractical due to the use of fast (rectangular) matrix multiplication. However, we show that any algorithm that breaks the trivial time bounds $n^k$ or $m^{k-1}$ has to use fast matrix multiplication. Specifically, we show that if projected-$(k-1)$-star-Listing has a fast combinatorial algorithm (i.e., one that does not use fast matrix multiplication) then so does $k$-Clique, contradicting the Combinatorial $k$-Clique Hypothesis. See \Cref{sec:lower-bounds} for more discussion.

\begin{restatable}{theorem}{LowerBoundCombinatorialTheorem}
\label{thm:lower-bound-combinatorial}
    Let $k \ge 3$. Assuming the Combinatorial $k$-Clique Hypothesis, no combinatorial algorithm solves Colored/Uncolored projected-$(k-1)$-star-Listing in time $\tOh(n^{k-\eps} + t)$ or $\tOh(m^{k-1-\eps} + t)$, for any $\eps > 0$. 
\end{restatable}

\noindent
This lower bound is similar to \cite[Theorem 11]{FanKZ23} and \cite[Theorem 8.5]{DZFK24}.

\subsubsection{More Lower Bounds for Projected Trees}
At this point one might ask: Can the preprocessing time $\tOh(n^{17.42})$ be further improved? Can one prove fine-grained lower bounds that limit how much further it can be improved?

We cannot hope to prove a lower bound ruling out preprocessing time $\tOh(n^3)$ or higher, because the preprocessing time of \Cref{thm:main1} would improve to $\tOh(n^3)$ if matrix multiplication were in time $\tOh(n^2)$, which currently cannot be ruled out.
Preprocessing time $O(n^{2-\eps})$ is easily ruled out, because the input size can be up to $\Theta(n^2)$, so in time $O(n^{2-\eps})$ one cannot read the whole input.
In the following theorem, we rule out some superquadratic time if $\omega > 2$, assuming the $k$-Clique Hypothesis. See \Cref{sec:lower-bounds} for background on this and other hypotheses used in this paper.

\begin{restatable}{theorem}{LowerBoundCliqueNTheorem}
\label{thm:lower-bound-clique-n}
    For every $\eps > 0$ there exist sufficiently large $k,k' \in \N$ such that assuming the $k$-Clique Hypothesis no algorithm solves Colored projected-$(k'-1)$-star-Listing in time $\tOh\left( n^{\frac{\omega}{3-\omega} - \eps} + t \right)$.
\end{restatable}

This is based on a standard fine-grained hypothesis, but does not match our upper bound (no matter what $\omega$ turns out to be). Next, we show a barrier that would match the $\tOh(n^3)$ preprocessing time of \Cref{thm:main1} if matrix multiplication were in time $\tOh(n^2)$.

\begin{restatable}{theorem}{LowerBoundLopsidedTheorem}
\label{thm:lower-bound-lopsided}
    If Colored projected-3-star-Listing is in time $\tOh(n^{3-\eps} + t)$ for some $\eps > 0$, then 4-Clique on a 4-partite graph whose parts have sizes $n,n,n,n^{1+\eps/3}$ can be solved in time $\tOh(n^3)$.
\end{restatable}

Note that this does not violate any standard fine-grained hypothesis. However, the conclusion of the above theorem would be surprising, as known 4-Clique algorithms construct an $n^2 \times n^{1+\eps/3}$ matrix and thus require time $\Omega(n^{3+\eps/3})$ for such instances. Thus, the theorem shows a barrier to improving the preprocessing time of \Cref{thm:main1} to subcubic, which makes it plausible that our preprocessing time in terms of $n$ is near-optimal.

\medskip
So far we have discussed running time in terms of $n$, now consider $m$. For projected stars, \Cref{thm:star-listing} yields preprocessing time $\tOh(m^{11.61})$, which would improve to $\tOh(m^2)$ if matrix multiplication were in time $\tOh(n^2)$. We prove a matching lower bound assuming $\omega=2$.

\begin{restatable}{theorem}{LowerBoundCliqueTheorem}
\label{thm:lower-bound-clique}
    For every $\eps > 0$ there exist sufficiently large $k,k' \in \N$ such that assuming the $k$-Clique Hypothesis no algorithm solves Colored projected-$(k'-1)$-star-Listing in time $\tOh\left( m^{\frac{\omega}{3-\omega} - \eps} + t \right)$.
\end{restatable}

For general projected trees, \Cref{thm:main1} yields preprocessing time $\tOh(m^{17.42})$, which would improve to $\tOh(m^3)$ if matrix multiplication were in time $\tOh(n^2)$. We leave open whether this can be further improved to $\tOh(m^2)$.

\subsection{Generalizations to Graphs and Hypergraphs}

Now let us broaden our perspective from trees to general (hyper)graphs $H$.
For $H$-Detection, where the task is to decide whether there exists a solution, Marx in his seminal work~\cite{Marx13} defined the submodular width $\subw(H)$ and showed that the problem can be solved in time $n^{O(\subw(H))}$, and requires time at least $n^{\subw(H)^{\Omega(1)}}$ assuming the Exponential Time Hypothesis. 
The PANDA algorithm~\cite{KNS17,KNS25} improved the running time to $\tOh(m^{\subw(H)})$, and generalized the result to $H$-Listing in time $\tOh(m^{\subw(H)} + t)$ and $H$-Enumeration in preprocessing time $\tOh(m^{\subw(H)})$ and delay $O(1)$ (see also~\cite{KNS26,KC26}).

While projections are irrelevant for the detection problem, it remained open whether the known listing and enumeration algorithms can be generalized to incorporate projections. As discussed in \Cref{sec:intro:priorwork}, already on projected trees prior work had a worst-case preprocessing time of $\Omega(n^k)$ or $\Omega(m^{k-1})$, or a delay of $m^{\Omega(1)}$, or assumed $\omega = 2$. (For example, the PANDA algorithm~\cite{KNS17,KNS25} solves $(H,I)$-Listing in time $\tOh(m^{\fcsubw(H,I)} + t)$, where $\fcsubw(H,I)$ is the free-connex submodular width of $(H,I)$, which is $k-1$ for the projected-$(k-1)$-star; see \Cref{sec:subw}.)

We show that $(H,I)$-Enumeration is in preprocessing time $m^{O(\subw(H))}$ and delay $\tOh(1)$, thus fully generalizing Marx' result from detection to listing and enumeration of projected hypergraphs.

\begin{restatable}{theorem}{AlgoHypergraphsTheorem}
\label{thm:algohypergraphs}
  For every $H$ and $I$, Colored/Uncolored $(H,I)$-Listing is in time $\tOh(m^{17.42 \subw(H)} + t)$, and Colored/Uncolored $(H,I)$-Enumeration is in preprocessing time $\tOh(m^{17.42 \subw(H)})$ and delay $\tOh(1)$. If $n \times n$ matrix multiplication is in time $\tOh(n^2)$, then the $m^{17.42 \subw(H)}$ term improves to $m^{3 \subw(H)}$.
\end{restatable}

Prior to our work, Hu~\cite{Hu25} showed that Colored $(H,I)$-Listing is in time $\tOh(m^{\subw(H)} t^{5/6} + t) \le \tOh(m^{6 \subw(H)}+t)$ but assumed $\omega = 2$ (or more precisely that matrix multiplication is in time $\tOh(n^2)$). We improve this prior work in several directions: (1) We remove the unproven assumption on matrix multiplication. (2) If matrix multiplication can be solved in time $\tOh(n^2)$, we improve the running time from $\tOh(m^{6 \subw(H)}+t)$ to $\tOh(m^{3 \subw(H)}+t)$. (3) We extend the result from listing to enumeration.

Our result has the same drawbacks as in the case of projected trees, most of which are again necessary: 
(1) In contrast to the $\tOh(m^{\subw(H)})$ preprocessing time for $H$-Enumeration, for $(H,I)$-Enumeration we only obtain $m^{O(\subw(H))}$. However, preprocessing time $\tOh(m^{\subw(H)})$ is impossible by the known fine-grained classification of trees, where $\subw(H) = 1$~\cite{BraultBaron13}.
(2) In contrast to PANDA's delay of $O(1)$ we only achieve delay $\tOh(1)$; we do not know whether this can be improved.
(3) Our algorithm is impractical due to the use of fast (rectangular) matrix multiplication, but matrix multiplication is necessary already for projected stars.

\subsection{Equivalence of Listing and Enumeration}

Listing and enumeration are widely studied in the database community and in the algorithm theory community.
So far they have been studied separately. Here we show that under natural assumptions listing and enumeration are \emph{equivalent}.

Recall that reducing listing to enumeration is trivial: If we have an enumeration algorithm with preprocessing time $T_P$ and delay $T_D$ then we obtain a listing algorithm in time $O(T_P + (t+1) \cdot T_D)$ by running the enumeration algorithm until it has printed all $t$ solutions and terminates. 
Reducing enumeration to listing is the difficult direction. So far, all attempts at this direction were restricted to specific classes of listing algorithms (e.g.~\cite{Uno03} and \cite[Section 3.3]{CS24} presented a conversion of so-called flashlight listing algorithms).
As our second main result, we present the first \emph{enumeration-to-listing reduction} without assumptions on the structure of the listing algorithm. It does make some natural assumptions on the listing problem, which are typically satisfied for colored problems. Specifically, it works for any problem over a discrete search space $\set{0, 1}^d$ that admits an efficient self-reduction into two pieces. (For this general formulation see \Cref{thm:enumtolist} in \Cref{sec:enumtolist}.) When specialized to Colored $(H,I)$-Subhypergraph, the reduction reads:

\begin{restatable}{theorem}{MainSimplifiedTheorem}
\label{thm:main2-simplified}
  Let $f,g$ be monotone functions such that $f(n,m)$ and $g(n,m)$ can be computed in time $O(f(n,m) + g(n,m))$.
  If Colored $(H,I)$-Listing is in time $O(f(n,m) + t \cdot g(n,m))$, then Colored $(H,I)$-Enumeration is in preprocessing time $O((f(n,m) + g(n,m) + n + m) \log^2 n)$ and delay $O(g(n,m))$.
\end{restatable}

We make use of this theorem to obtain enumeration algorithms in the proofs of \Cref{thm:main1,thm:algohypergraphs}. At the core, we design a fast algorithm for Colored $(H,I)$-Listing, running in time $O(f(n,m) + t \cdot g(n,m))$ for $f(n,m) = \tOh(m^{17.42 \cdot \subw(H)})$ and $g(n,m) = \tOh(1)$. Since $H$ is fixed, $\subw(H)$ is a constant, and thus $f(n,m)$ and $g(n,m)$ can be computed in constant time. Hence the above theorem is applicable and yields an algorithm for Colored $(H,I)$-Enumeration with preprocessing time $\tOh(m^{17.42 \cdot \subw(H)})$ and delay $\tOh(1)$. Then we use color coding to convert both algorithms to the uncolored setting. This exemplifies the usefulness of our enumeration-to-listing reduction, and we expect more applications to come.

\section{Technical Overview}
In this section we describe the high-level ideas behind our results. We start by describing our algorithms for Colored $(H,I)$-Listing.

\paragraph{Matrix Multiplication Tools} 
At the core of our work, we rely on two tools for multiplying matrices. 
Denote by $\mm(x,y,z)$ the time complexity of multiplying an $x \times y$ matrix by a $y \times z$ matrix, both with nonnegative integer entries.
The first tool is \emph{rectangular} matrix multiplication: $\mm(n,n^{\alfa},n) \le \tOh(n^2)$ for $\alfa = 0.17227$~\cite{Coppersmith82,Williams11}.\footnote{Note that $\alfa$ differs from the usual dual matrix multiplication exponent $\alpha$. The latter is defined such that \smash{$\mm(n, n^\alpha, n) \le O(n^{2+\epsilon})$} for any constant $\epsilon > 0$, and the state-of-the-art bound is $\alpha \geq 0.321334$~\cite{VXXZ24}. In our work, we want to avoid the $n^\epsilon$ overhead and insist on a polylogarithmic overhead, so we settle for $\alfa \geq 0.17227$ defined as above.} 
An implication of this result, which follows from a simple blocking trick, is that $\mm(x,y,z) \le \tOh(xz + y \cdot (xz)^{1-\alfa} (x^{\alfa} + z^{\alfa}))$.

The second tool is \emph{output-sensitive} matrix multiplication. Write $\mm(x,y,z; t)$ for the time complexity of matrix multiplication where the product has $t$ non-zero entries. It was shown in \cite{ABFK24} that
\[\mm(x,y,z; t) \le \tOh\bigg(\max_{\substack{x' \le x, z' \le z\\x'z' \le 8t}} \mm(x',y,z')\bigg). \]
Combining this with rectangular matrix multiplication yields
\begin{equation} \label{eq:rect-mm}
    \mm(x,y,z; t) \le \tOh\left(t + y \cdot t^{1-\alfa} (x^{\alfa} + z^{\alfa})\right).
\end{equation}
See \Cref{sec:mm} for details and a further improvement on the time bound.

\paragraph{A Simple Algorithm for $d$-Stars}
To illustrate the main ideas behind our algorithm, we first describe a particularly simple and clean algorithm for the special case of listing projected-$d$-stars. Recall that the projected-$d$-star is the pattern $(H,I)$ where $H$ is the star graph with center node $0$ and leaf nodes $1, \dots, d$, and we are projecting to the leaves $I = \set{1,\ldots,d}$. Since for now we focus on the colored problem, the host graph $G$ has a correspondingly partitioned vertex set $V = V_0 \cup V_1 \cup \ldots \cup V_d$.

Assume that each $v_0 \in V_0$ has at least $D_j$ and at most $2D_j$ neighbors in $V_j$, for some $D_1,\ldots,D_d \in \N$. This assumption can be ensured by splitting $V_0$ into $\log^d n$ classes. (This is sometimes called \emph{degree uniformization}, and it can be avoided, but in this overview we use it for simplicity.)

We partition the leaves $\set{1,\ldots,d}$ into $J \cup \overline{J}$ such that the products of degrees are roughly balanced. By a greedy assignment, one can ensure
\begin{equation*}
    \max \bigg\{ \prod_{j \in J} D_j,\; \prod_{j \in \overline{J}} D_j \bigg\} \;\le\; \sqrt{ n \cdot \vphantom{\prod} \smash{\prod_{1 \leq j \leq d}} D_j } \;\;=:\; \kappa.
\end{equation*}
We write down two incidence matrices $A \in \set{0,1}^{(\prod_{j \in J} V_j) \times V_0}$ and $B \in \set{0,1}^{V_0 \times (\prod_{j \in \overline{J}} V_j)}$. Specifically, $A[(v_j)_{j \in J}, v_0] = 1$ if and only if $v_0$ is adjacent to each $v_j, j \in J$. Similarly, $B[v_0, (v_j)_{j \in \overline{J}}] = 1$ if and only if $v_0$ is adjacent to each $v_j, j \in \overline{J}$.
Observe that (the support of) the product $AB$ encodes all projected-$d$-star solutions.
In particular, the product has $t$ non-zero entries, where $t$ is the number of projected-$d$-stars. We thus compute the matrix product $AB$ by output-sensitive matrix multiplication, and read off the list of all projected-$d$-stars.

It remains to analyze the running time. Since each of the at most $n$ vertices in $V_0$ is connected to at most $\prod_{j \in J} (2D_j) \leq 2^d \kappa$ tuples in $\prod_{j \in J} V_j$, the matrix $A$ has at most $n \cdot 2^d \kappa = O(n \kappa)$ non-zeros. We remove the all-zero rows from $A$ to ensure that $A$ has at most $O(n \kappa)$ rows. We proceed similarly with~$B$ to reduce its number of columns to $O(n \kappa)$. Now observe that each single vertex in $V_0$ generates at least~$\prod_{j=1}^d D_j$ solutions. In other words, we have $t \ge \prod_{j=1}^d D_j = \kappa^2/n$, and thus $\kappa \le (n t)^{1/2}$. Hence, the number of rows in $A$ and the number of columns in $B$ are at most $O(n^{3/2} t^{1/2})$, and therefore the matrix multiplication runs in time $\mm(n^{3/2} t^{1/2}, n, n^{3/2} t^{1/2}; t)$. Using~\eqref{eq:rect-mm}, we bound this by
\begin{align*}
    \tOh\left(t + n \cdot t^{1-\alfa} (n^{3/2} t^{1/2})^{\alfa} \right)
    \ALIGN= \tOh\left(t + n^{1 + 3 \alfa / 2} t^{1 - \alfa / 2}\right) \NEWLINE
    \ALIGN\leq \tOh\left(t + n^{2/\alfa + 3} \right),
\end{align*}
where in the last step we applied the weighted AM-GM inequality.\footnote{Specifically, we used that $x^{\beta} y^{1-\beta} \leq \beta x + (1 - \beta) y \leq x+y$ for the values $\beta = \alfa/2$, $x = n^{(1 + 3 \alfa / 2)/\beta} = n^{2/\alfa + 3}$, and $y = t$.}
The resulting algorithm runs in time $\tOh(t + n^{14.61})$ since $\alfa \geq 0.17227$. If $n \times n$ matrix multiplication were in time $\tOh(n^2)$, we could use $\alfa = 1$ and the running time would improve to $\tOh(t + n^5)$.

\paragraph{Listing a Restricted Class of Projected Trees}
The core of our work is to generalize the above simple algorithm from stars to a restricted class of projected trees $(H,I)$, where $H$ is a rooted full binary tree (i.e., a rooted tree where each node has either zero or two children) and $I$ is exactly the set of leaves of $H$. 

In our generalized algorithm, we split $H$ into two subtrees $H_0, H_1$ intersecting in exactly one non-leaf node $i \notin I$. For $b \in \set{0,1}$, we let $I_b := I \cap V(H_b)$ and list \emph{some} $(H_b,I_b)$-Subgraphs. We then write down the incidence matrix $A_b$ between the $(H_b,I_b)$-Subgraphs and the vertices $V_i$. By computing the matrix product $A_0 A_1^{\mathrm{T}}$, we effectively glue these subgraphs and project away $i$, thereby listing \emph{some} $(H,I)$-Subgraphs.
We cannot afford to list all $(H_b,I_b)$-Subgraphs in one go, because this might lead to an incidence matrix $A_b$ of dimension $\Theta(t) \times n$. Hence, we only list the few ones that have ``many'' extensions to global solutions, which ensures that the incidence matrices have manageable sizes. We then eliminate them from future considerations, and proceed with another split of $H$. By processing all splits in an appropriate order, we can guarantee that all $(H,I)$-Subgraphs are covered. See \Cref{sec:tree} for details. To make this work, we also need to extend the matrix multiplication tools to compute only the sparse rows of the product matrix, at a faster running time; see \Cref{sec:mm}.

\paragraph{From Restricted Trees to Arbitrary Trees}
From this point on, the remainder of our algorithm is a sequence of reductions, reducing more and more general problems to the algorithm for restricted trees from the previous paragraph.

Let us consider an arbitrary tree $H$ and an arbitrary $I \subseteq V(H)$. 
First, we observe that leaves not in $I$ can be removed by a simple, linear-time reduction. Second, for each $i \in I$ that is not a leaf, we can split the pattern at $i$ and recurse on the connected components of $H-i$. These two reductions ensure that $I$ is exactly the set of leaves. Finally, we pick a suitable root, replace every node having more than two children with a path of nodes having two children each, then attach a dummy child to every node having only one child. (We also introduce corresponding edges in $G$ so that solutions in the new instance are in bijection to solutions in the old instance.) This turns the tree into a rooted full binary tree.
Altogether, the three reductions turn the pattern $(H,I)$ into a restricted projected tree, on which the algorithm from the last paragraph applies. This extends the applicability of the algorithm and proves \Cref{thm:main1} for colored listing. See \Cref{sec:reductions1} for details.

\paragraph{From Trees to Hypergraphs}
Now consider an arbitrary hypergraph $H$. The seminal PANDA algorithm~\cite{KNS17,KNS25} for conjunctive queries shows its full power when phrased as a reduction~\cite{Hu24}. In this form, it allows us to reduce general hypergraphs to acyclic hypergraphs (in particular, in the graph setting it reduces to trees). This reduction lets us completely bypass submodular width in our work, and instead we only need to refer to the PANDA algorithm.

From acyclic hypergraphs it is easy to further reduce to trees, by replacing each hyperedge with a node and the intersection of two hyperedges with an edge; see \Cref{sec:reductions2} for details. Altogether, we have turned an arbitrary hypergraph $H$ to a tree, where the algorithm from the last paragraph applies. This extends the applicability of the algorithm and proves \Cref{thm:algohypergraphs} for colored listing.

\paragraph{From Listing to Enumeration} 
So far we discussed Colored $(H,I)$-Listing.
We can extend the algorithm to Colored $(H,I)$-Enumeration by applying our enumeration-to-listing reduction.

This reduction roughly works as follows. Suppose there is a listing algorithm that solves an instance $I$ with $t$ solutions in time $f(I) + t \cdot g(I)$. We let this algorithm run with a time budget of $2 f(I)$. If it uses up its time budget and aborts, then we branch on the first bit of the search space, recursively solving subinstances $I_0,I_1$. In the analysis, we use that each inner node of the recursion tree has more than $b := f(I)/g(I)$ solutions, which bounds the size of the recursion tree. Specifically, we show that at the time of printing the $j$-th solution, the size of the recursion tree is $O(d + d j / b)$, where $d = O(\log n)$ bounds the recursion depth. Since each recursive call has a direct time budget of $O(f(I))$, the total time before printing the $j$-th solution is $O(f(I) \cdot d + j \cdot g(I) \cdot d)$.
Thus, we arrived at an ``amortized'' enumeration algorithm, and a well-known construction converts this into a proper enumeration algorithm with preprocessing time $O(f(I) \cdot d)$ and delay $O(g(I) \cdot d)$.
Here we described a variant of our \Cref{thm:main2-simplified}; the general tradeoff (in particular, the algorithm with preprocessing time $O(f(I) \cdot d^2)$ and delay $O(g(I))$ as claimed in \Cref{thm:main2-simplified}) requires additional ideas and a careful charging argument, see \Cref{thm:enumtolist} in \Cref{sec:enumtolist}.

\paragraph{From Colored to Uncolored Problems} 
So far we discussed the colored problems.
The classic color coding technique~\cite{AYZ95} converts algorithms from the colored setting to the uncolored setting. Specifically, it replaces a given uncolored problem instance by several colored problem instances, such that the union of the solution sets of the constructed colored instances is the solution set of the original uncolored instance. This readily reduces uncolored listing to colored listing: Run the colored listing algorithm on all colored problem instances constructed by color coding, then compute the union of the listed solutions, e.g., by sorting and deduplicating. For enumeration, by following the same recipe we obtain a ``cheating'' enumeration algorithm which may print the same solution multiple times. The Cheater's Lemma~\cite{CK21} then converts such a cheating enumeration algorithm into a proper enumeration algorithm, see \Cref{sec:uncoloredalgo}.

This finishes the outline of the proof of \Cref{thm:algohypergraphs}.

\paragraph{Faster Algorithm for Projected-$d$-Stars}
For projected-$d$-star-Listing we obtain an improved algorithm in terms of~$m$ (\Cref{thm:star-listing}). This is the most elaborate result of this paper. It starts from a more delicate balancing procedure than we used in the simple projected-$d$-star algorithm described earlier, which solves the ``dense'' case efficiently. In order to apply this procedure, we utilize Shearer's Lemma to identify the dense case. Finally, we use hashing to densify a general instance, and iterative recovery techniques to read out solutions from the densified instance. For more details see \Cref{sec:star-m2}.

\bigskip
This finishes the overview of our algorithmic results. The fine-grained lower bounds presented in this paper are simple reductions from $k$-Clique, varying only in the cardinalities of the central part versus the leaf parts, see \Cref{sec:lower-bounds}. 

\section{Preliminaries}
\label{sec:preliminaries}

We write $[k] := \set{1, \dots, k}$. By $\tOh$-notation we hide factors of the form $\polylog(n)$, or equivalently $\polylog(m)$.

\begin{unnumbered}[Projections and Extensions]
    Let $I \subseteq J$ be sets of indices. For a tuple $\tup{v} = (v_j)_{j \in J}$, we denote $\tup{v}_I := (v_i)_{i \in I}$. We say that $\tup{v}_I$ is $\tup{v}$ \EMPH{projected to} $I$. Interchangeably, we say that $\tup{v}$ \EMPH{extends} $\tup{v}_I$.
    For a tuple $\tup{u} = (u_i)_{i \in I}$ and a set of tuples $\mathcal{S}$, we write $\mathcal{S}[\tup{u}] := \Set{\tup{v} \in \mathcal{S} : \tup{v}_I = \tup{u}}$.
\end{unnumbered}

\begin{unnumbered}[$H$-Copies]
    Let $H = ([k], \mathcal{E})$ be a pattern hypergraph. An \EMPH{uncolored $H$-copy} in a host hypergraph $G = (V,E)$ is a tuple $\tup{v} = (v_1, \dots, v_k) \in V^k$ of distinct vertices such that $\tup{v}_e \in E$ for all $e \in \mathcal{E}$. A \EMPH{colored $H$-copy} in a host hypergraph $G = (V,E)$ with a partition $V = V_1 \cup \cdots \cup V_k$ is a tuple $\tup{v} = (v_1, \dots, v_k) \in \prod_{i \in [k]} V_i$ such that $\tup{v}_e \in E$ for all $e \in \mathcal{E}$.
\end{unnumbered}

To avoid confusion, in the pattern we speak of \emph{nodes}, whereas in the host we speak of \emph{vertices}. 

\begin{unnumbered}[Problem Definitions]
    Let $H = ([k], \mathcal{E})$ be a pattern hypergraph, and let $I \subseteq [k]$ be a set of nodes.
    In the \EMPH{Uncolored $(H,I)$-Subhypergraph} problem, we are given a host hypergraph $G = (V,E)$, and a \EMPH{solution} is any uncolored $H$-copy projected to $I$.
    In the \EMPH{Colored $(H,I)$-Subhypergraph} problem, we are given a host hypergraph $G = (V,E)$ with partition $V = V_1 \cup \ldots \cup V_k$, and a \EMPH{solution} is any colored $H$-copy projected to $I$.

    When $I = [k]$ is the full set of nodes, we abbreviate the ``\,$(H,I)$'' in the problem names to ``$H$''. When $H$ and $G$ are graphs, we further abbreviate ``Subhypergraph'' to ``Subgraph''.
\end{unnumbered}

Our algorithms and reductions will use the following subroutine that filters out all vertices of the host graph that do not participate in a solution. This is achieved by a variant of Yannakakis' classic algorithm~\cite{Yannakakis81}, for details see the \texttt{DeDangle} procedure in~\cite{Hu24}.

\begin{lemma}
    \label{lem:remove-ghost}
    Let $H = ([k], \mathcal{E})$ be a tree. Given a host graph $G = (V_1 \cup \ldots \cup V_k, E)$, in time $O(n + m)$ we can list all vertices $v \in V(G)$ that do not extend to any colored $H$-copy in $G$.
\end{lemma}

\section{Fast Matrix Multiplication}
\label{sec:mm}
Our algorithms exploit fast (rectangular and output-sensitive) matrix multiplication. Every matrix considered in this paper is non-negative, integral, and each row is sparsely represented as a list of non-zero entries. For $A \in \N^{x \times y}$, we write $\supp(A)$ for the set of non-zero entries in $A$, and let $|A| := x + |\supp(A)|$ be the size of the representation. We denote by $\mm(x,y,z)$ the time complexity of (deterministically) multiplying an $x \times y$ matrix with a $y \times z$ matrix.
In general form, the exponent of matrix multiplication is defined as $\omega(a,b,c) := \inf \set{ d : \mm(n^a, n^b, n^c) \le O(n^d)}$. The exponent of square matrix multiplication is $\omega := \omega(1,1,1)$, and the currently best bound is $2 \leq \omega < 2.3714$ \cite{ADVXXZ25}.

\emph{Rectangular} matrix multiplication normally refers to the dual exponent of matrix multiplication $\alpha := \sup \set{ d : \omega(1,d,1) = 2}$, and the currently best bound is $0.321334 \le \alpha \le 1$~\cite{VXXZ24}. 
Note that $\alpha$ is defined in such a way that \smash{$\mm(n, n^\alpha, n) \le O(n^{2+\eps})$} for any constant $\eps > 0$. In this paper, we need to avoid the $n^\eps$ overhead as it would translate into $m^{\eps}$ delay for our enumeration algorithms.
Therefore, instead of using $\alpha$, we start from a value $\alfa$ such that $\mm(n,n^\alfa,n) \le \tOh(n^2)$. 
Williams~\cite[Appendix~B]{Williams11} verified that Coppersmith's rectangular matrix multiplication algorithm~\cite{Coppersmith82} runs in time $O(n^2 \log^2 n) = \tOh(n^2)$, so we can pick $\alfa := 0.17227$ as the constant of Coppersmith's algorithm.\footnote{Williams~\cite{Williams11} only states a value of $0.1$, but the argument also works for Coppersmith's original constant $0.17227$.} The state-of-the-art improvements of $\alpha$ (e.g.~\cite{VXXZ24}) do not seem to lead to algorithms running in time $\tOh(n^2)$.
In later calculations, we will use that $3/\alfa \le 17.42$ for $\alfa = 0.17227$. Furthermore, if $\mm(n,n,n) = \tOh(n^2)$, then we could pick $\alfa = 1$ and would obtain $3/\alfa \le 3$.

Rectangular matrix multiplication implies the following, via a simple blocking trick.
\begin{proposition}
    \label{prp:mm}
    We have
    \MATH{\mm(x,y,z) \le \tOh\left( xz + y (xz)^{1-\alfa} (x^{\alfa} + z^{\alfa}) \right).}
\end{proposition}

\begin{proof}
    Suppose first that $x \leq z$. Given input matrices $A, B$, we break $A$ into blocks of shape $x \times \floor{x^\alfa}$, and $B$ into blocks of shape $\floor{x^\alfa} \times x$, padding with zeros when necessary. This way, computing $AB$ reduces to computing $\ceil{y / \floor{x^\alfa}} \cdot \ceil{z/x}$ products of an $x \times \floor{x^\alfa}$ matrix with an $\floor{x^\alfa} \times x$ matrix, each of which can be computed in time $\mm(x, \floor{x^\alfa}, x) \le \tOh(x^2)$. Therefore,
    \[
        \mm(x,y,z) \leq O\left( \frac{z}{x} + \frac{yz}{x^{1+\alfa}} \right) \cdot \tOh(x^2)
        = \tOh(xz + y (xz)^{1-\alfa} z^\alfa).
    \]
    In the symmetric case $x \geq z$, the same argument applies by swapping $x$ and $z$. 
\end{proof}

The second tool that we will use is \emph{output-sensitive} matrix multiplication. In \cite[Lemma 3.10]{ABFK24}, it was shown that multiplying $A \in \N^{x \times y}$ with $B \in \N^{y \times z}$ is in time
\[ \tOh\Bigg( |A| + |B| + \max_{\substack{x' \le x, z' \le z\\x'z' \le 8s}} \mm(x',y,z')\Bigg) \]
where $s := |AB|$. Combining with \Cref{prp:mm} gives a bound $\tOh\left(|A| + |B| + s + y s^{1-\alfa} (x^{\alfa} + z^{\alfa})\right)$. Unfortunately, this is insufficient for our purpose. We will need a generalization that only computes sparse rows of the product at an improved running time. To this end, we need to repeat parts of the algorithm of \cite{ABFK24}. The starting point is a variant of \cite[Lemma 3.2]{ABFK24} where we impose a sparsity condition $\kappa$ on each row of the product matrix, which allows us to bound the compressed dimension $z'$ by $4\kappa$.

\begin{restatable}{lemma}{RecoverProduct}
    \label{lem:recover}
    There is a deterministic algorithm \algorecover{$A,B,\allowbreak S,\kappa$} that, given $A \in \N^{x \times y}$, $B \in \N^{y \times z}$, $S \supseteq \supp(AB)$ and $\kappa \geq \max_{i \in [x]} \Card{\set{j \in [z] : (i,j) \in S}}$, computes $AB$ in time
    \[
        \tOh\Bigg( |A| + |B| + |S| +
        \max_{\substack{x' \leq x, z' \leq 4\kappa,\\ x'z' \leq 4|S|}}
        \mm(x', y, z') \Bigg).
    \]
\end{restatable}

The proof is essentially the same as in \cite{ABFK24}; see \Cref{sec:omitted-proofs} for details. Building on this lemma, we extend \cite[Lemma 3.10]{ABFK24} to the following:
\begin{lemma}[Sparse Output-Sensitive Matrix Multiplication]
    \label{lem:sparse-mm}
    There is a deterministic algorithm \algomult{$A,B,\kappa$} that, given $A \in \N^{x \times y}$, $B \in \N^{y \times z}$ and $\kappa \geq 1$, outputs the product matrix $AB$ restricted to \EMPH{light} rows, i.e., rows with at most $\kappa$ non-zeros. Denoting $s = \min\set{|AB|, x\kappa}$, the algorithm runs in time 
    \[
        \tOh\Bigg( |A| + |B| + s + \max_{\substack{x' \leq x, z' \leq 8\kappa,\\ x'z' \leq 8s}}
         \mm(x', y, z') \Bigg).
    \]
    In particular, its running time can be bounded by $\tOh\big(|A| + |B| + s + y s^{1-\alfa} (x^\alfa + \kappa^\alfa)\big)$.
\end{lemma}

\begin{proof}
    By padding the matrix $B$ with all-zero columns, we may assume that $z$ is a power of two. We write $M_L$ for the restriction of a matrix $M$ to rows $L$. The algorithm works recursively:

    \begin{pseudocode}
    \Function{multiply}{$A,B,\kappa$}
        \If{$z = 1$}
            \State compute $C := AB$ by the naive algorithm
        \Else
            \State construct $B' \in \N^{y \times (z/2)}$ where $B'_{k,j} = B_{k,2j-1} + B_{k,2j}$
            \State $C' := \algomult{A,B',\kappa}$
            \State $L' := \rows(C')$
            \State $S := \bigcup_{(i,j) \in \supp(C')} \Set{(i,2j-1), (i,2j)}$
            \State $C := \algorecover{A_{L'}, B, S, 2\kappa}$
        \EndIf
        \State let $L$ be the set of light rows of $C$
        \State \Return $C_L$
    \EndFunction
    \end{pseudocode}

    \smallskip
    We claim the following:
    \begin{enumerate}[(i),topsep=5pt,parsep=2pt]
        \item the computed set $L$ is the set of light rows of $AB$; and
        \item the returned matrix is $A_L B$.
    \end{enumerate}
    
    If $z = 1$ then the claim holds trivially. Now consider $z \geq 2$. We inductively assume that the claim holds for the recursive call; that is,
    \begin{enumerate}[(i'),topsep=5pt,parsep=2pt]
        \item $L'$ is the set of light rows of $AB'$; and
        \item $C' = A_{L'} B'$.
    \end{enumerate}
    Let us gather some simple observations:
    \begin{enumerate}[(i),resume]
        \item We have $(AB')_{i,j} = (AB)_{i,2j-1} + (AB)_{i,2j}$ by the definition of $B'$. Hence, $(i,j) \in \supp(AB')$ if and only if $(i,2j-1) \in \supp(AB)$ or $(i,2j) \in \supp(AB)$.
        
        \item $S = \bigcup_{(i,j) \in \supp(A_{L'} B')} \Set{(i,2j-1), (i,2j)}$ by (ii'). Hence $S \supseteq \supp(A_{L'} B)$ by (iii).

        \item $2\kappa \geq \max_{i \in L'} \Card{\set{j \in [z] : (i,j) \in S}}$. Indeed, each row $i \in L'$ is light in $AB'$ by (i'). In particular, it contains at most $\kappa$ non-zeros. Each non-zero in row $i$ contributes two elements in $\set{j \in [z] : (i,j) \in S}$, while each zero does not contribute any. So the set has cardinality at most $2\kappa$.
    \end{enumerate}
    Observations (iv) and (v) verify the preconditions of \algorecover{$A_{L'}, B, S, 2\kappa$}, so it computes $C = A_{L'} B$ by \Cref{lem:recover}. Note that $L \subseteq L'$ by definition of \algomult{}, so the returned matrix is $C_L = A_L B$, which establishes (ii).
    To establish (i), note that if row $i$ is light in $AB$, then it is also light in $AB'$ by~(iii), so $i \in L'$ by (i'). Consequently, every light row of $AB$ appears in $C$, so $L$ (the set of light rows in $C = A_{L'} B$) is exactly the set of light rows of $AB$. The induction is complete.

    Let us bound the running time. As we descend in the recursion, the parameters $x, \kappa$ stay the same, and we have $|AB'| \leq |AB|$ by (iii). Hence the parameter $s = \min\set{|AB|, x\kappa}$ does not increase. Besides, the parameter $z$ halves on each level, so the depth of the recursion is $\log z = \tOh(1)$.
    
    Now we focus on a fixed recursion level. In the case $z = 1$, computing $C$ by the trivial algorithm takes time $O(|A|)$ because $B$ is just a column vector. In the case $z \geq 2$, we have $|C'| \leq |AB|$ by definition, and $|C'| \leq x + |S| \leq x + 2x\kappa \leq 3x\kappa$ by (v). Hence $|C'| \leq 3s$. Constructing the matrix $B'$ takes time $O(|B|)$. Computing $S$ takes time $O(|C'|) \leq O(s)$. Computing $C$ takes time
    \[
        \tOh \Bigg( |A| + |B| + s + \max_{\substack{x' \leq x, z' \leq 8\kappa,\\ x'z' \leq 8s}}
        \mm(x', y, z') \Bigg)
    \]
    by \Cref{lem:recover}. Computing $L$ takes time $O(|C'|) \leq O(s)$. Summing all these terms over all $\tOh(1)$ levels, the total running time is as claimed.
    
    The ``in particular'' part of the lemma follows by plugging in the bound from Proposition~\ref{prp:mm}.
\end{proof}

\section{Listing Restricted Projected Trees}
\label{sec:tree}

In this section, we present an algorithm that solves $(H,I)$-Listing for any rooted full binary tree $H$ and $I$ being its set of leaves. 
A \emph{rooted tree} is a tree with a designated root node. We think of the edges being directed away from the root, although in reality all graphs in this paper are undirected. Accordingly the \emph{children} of a node are its neighbors that have larger distance to the root than the node itself. The \emph{leaves} of a rooted tree are the nodes without children. Note that the leaves of a rooted tree are almost the same as the leaves of the corresponding unrooted tree, except that if the root has exactly one child then it becomes a leaf in the unrooted tree, while the root is no leaf of the rooted tree (unless the tree is a single node).
A \emph{rooted binary tree} is a rooted tree where every node has at most two children.
A \emph{rooted full binary tree} is a rooted tree where every node has either no child or exactly two children.

Our correctness proofs will use a simple fact about non-negative integer matrix products.
\begin{lemma}
    \label{lem:mm-restriction}
    Let $A \in \N^{x \times y}$, $B \in \N^{y \times z}$ and $\kappa \geq 1$. We say that a row is light if it contains at most $\kappa$ non-zeros. Let $L \subseteq [y]$ be the set of light rows of $B$. Let $B'$ be the matrix $B$ restricted to rows in $L$, and let $A'$ be the matrix $A$ restricted to columns in $L$. Then the set of light rows of $AB$ is exactly \smash{$\Lambda = \Set{i \in [x] : \text{$i$ is a light row of $A'B'$ and $A_{ik} = 0$ for all $k \notin L$}}$.} Moreover, $AB$ and $A'B'$ agree on rows~$\Lambda$.
\end{lemma}

\begin{proof}
    We make the following claim: If $A_{ik} = 0$ for all $k \notin L$, then $(AB)_{ij} = (A'B')_{ij}$ for all $j \in [z]$.
    Indeed, we have
    \[ (A'B')_{ij} = \sum_{k \in L} A'_{ik} B'_{kj} = \sum_{k \in L} A_{ik} B_{kj} = \sum_{k \in [y]} A_{ik} B_{kj} \]
    where the last step used that $A_{ik} = 0$ for all $k \notin L$.
    This shows that $AB$ and $A'B'$ agree on $\Lambda$. It remains to prove that the set of light rows in $AB$ is exactly $\Lambda$.
    
    For one direction, consider any $i \in \Lambda$. Since $A_{ik}=0$ for all $k \notin L$, the claim implies that row $i$ has the same number of non-zeros in $AB$ and $A'B'$. Since $i$ is a light row in $A'B'$, it is also light in~$AB$.
    
    For the other direction, consider any light row $i$ of $AB$. First, suppose that there exists $k \notin L$ with $A_{ik} > 0$. Then for all $j \in [z]$ we have $(AB)_{ij} \geq A_{ik} B_{kj} \geq B_{kj}$ because the matrices are non-negative and integral. Since $k \notin L$, row $k$ in $B$ has more than $\kappa$ non-zeros. As a result, row $i$ in $AB$ also has more than $\kappa$ non-zeros. Hence $i$ is not a light row in $AB$, a contradiction.
    
    Therefore, we can now assume that $A_{ik} = 0$ for all $k \notin L$. Then the claim implies that row $i$ has the same number of non-zeros in $AB$ and $A'B'$. Since $i$ is light in $AB$, it is also light in $A'B'$.
\end{proof}

Before we can present our algorithm, we discuss two subroutines. In the remainder, for a tuple $\tup{v} = (v_i)_{i \in I}$ and a set of tuples $\mathcal{S}$, we use the notation $\mathcal{S}[\tup{v}] := \Set{\tup{u} \in \mathcal{S} : \tup{u}_I = \tup{v}}$. In particular, for a singleton set $I = \set{\rho}$ and a vertex $v_\rho$ we write $\mathcal{S}[v_\rho] := \Set{\tup{u} \in \mathcal{S} : \tup{u}_{\set{\rho}} = v_\rho}$.

\begin{lemma}
    \label{lem:partial-list}
    There is an algorithm \algolight{$H,I,\rho \mid G, \kappa$} with the following properties.
    \begin{itemize}
        \item It is given a rooted binary tree $H$ with root $\rho$ and leaves $I$; an $n$-vertex graph $G$ with one part $V_l$ for each $l \in V(H)$ (it might have more parts but they are irrelevant); and $\kappa \geq 1$.
        \item It outputs $L_\rho = \Set{v_\rho \in V_\rho : |\sol[v_\rho]| \leq \kappa}$ along with $\sol[v_\rho]$ for each $v_\rho \in L_\rho$, where $\sol$ is the set of solutions to $(H,I\cup\set{\rho})$-Listing on $G[V_l : l \in V(H)]$.
        \item It runs in time $\tOh\left(n^{2-\alfa} \kappa + n^2 \kappa^{1-\alfa}\right)$, where the hidden constant is polynomial in $k = |V(H)|$.
    \end{itemize}
\end{lemma}

\begin{proof}
    We implement \algolight{$H,I,\rho \mid G, \kappa$} as a recursive algorithm.
    
    \textbf{Case 0:} $\rho$ is a leaf (thus it is the only node in $H$). The algorithm simply outputs $V_\rho$, and for each $v_\rho \in V_\rho$, outputs the singleton set $\set{v_\rho}$. This is correct because all vertices $v_\rho \in V_\rho$ satisfy $\sol[v_\rho] = \set{v_\rho}$, and in particular, $1 = |\sol[v_\rho]| \leq \kappa$. The running time is $O(n)$.
        
    \textbf{Case 1:} $\rho$ has only one child $\theta$.
    \begin{itemize}[--]
        \item Let $H' = H - \rho$. Recursively run \algolight{$H',I,\theta \mid G,\kappa$}. Writing $\sol'$ for the set of solutions to $(H',I \cup \set{\theta})$-Listing on $G[V_l : l \in V(H')]$, the call returns the subset $L_\theta = \Set{v_\theta \in V_\theta : |\sol'[v_\theta]| \leq \kappa}$ along with $\sol'[v_\theta]$ for each $v_\theta \in L_\theta$.
        \item Let $A \in \set{0,1}^{V_\rho \times L_\theta}$ be the adjacency matrix of $G$ between $V_\rho$ and $L_\theta$.
        \item Let $B \in \set{0,1}^{L_\theta \times (\prod_{i \in I} V_i)}$ where $B_{v_\theta, \tup{w}} = 1$ if and only if $(v_\theta,\tup{w}) \in \sol'[v_\theta]$.
        \item Compute $C := \text{\algomult{$A,B,\kappa$}}$ (the algorithm from \Cref{lem:sparse-mm}).
        \item Output $L_\rho := \{v_\rho \in \rows(C) : \text{$v_\rho$ does not have} \allowbreak \text{  neighbors in $V_\theta \setminus L_\theta$} \}$. Then for each $v_\rho \in L_\rho$, output $\Set{(v_\rho,\tup{w}) : C_{v_\rho, \tup{w}} > 0}$.
    \end{itemize}
        
    For the correctness analysis, we extend the matrix $A$ to $A^+ \in \set{0,1}^{V_\rho \times V_\theta}$ and the matrix $B$ to $B^+ \in \set{0,1}^{V_\theta \times (\prod_{i \in I} V_i)}$, with the same entry definition as written in the algorithm. It is clear that $\supp(A^+ B^+) = \sol$.
        
    On the other hand, $A,B$ are the restrictions of $A^+,B^+$ to $L_\theta$. Since $L_\theta$ is the set of light rows in $B^+$, we can apply \Cref{lem:mm-restriction} to deduce:
    \begin{enumerate}[(i)]
        \item the set of light rows in $A^+B^+$ is $\{v_\rho : \text{$v_\rho$ is a light row} \allowbreak \text{ in $AB$ and $A^+_{v_\rho, v_\theta} = 0$ for all $v_\theta \in V_\theta \setminus L_\theta$}\}$;
        \item $AB$ and $A^+B^+$ agree on these rows.
    \end{enumerate}
    The algorithm computes a matrix $C$ that is equal to $AB$ restricted to its light rows by \Cref{lem:sparse-mm}. It follows from (i) that the output set $L_\rho$ is exactly the set of light rows in $A^+B^+$, or in other words, $\Set{v_\rho \in V_\rho : |\sol[v_\rho]| \leq \kappa}$. It also follows from (ii) that the algorithm outputs all non-zero entries for each light row, or in other words, $\sol[v_\rho]$ for each $v_\rho \in L_\rho$. The correctness is established.
        
    Let us analyze the running time. The recursive call runs in the stated time by induction. The matrix $A$ has shape $n \times n$ and can be constructed in time $O(n^2)$. The matrix $B$ has $n$ rows and each row contains at most $\kappa$ non-zeros by definition of $L_\theta$, so it can be constructed in time $O(n\kappa)$. With these bounds, \algomult{} runs in time $\tOh\left(n^2 + n\kappa + n (n\kappa)^{1-\alfa} (n^\alfa + \kappa^\alfa)\right) = \tOh\big(n^{2-\alfa} \kappa + n^2 \kappa^{1-\alfa}\big)$ by \Cref{lem:sparse-mm}. Altogether the time bound is as claimed.
    
    \textbf{Case 2:} $\rho$ has two children $\theta_1$ and $\theta_2$.
    \begin{itemize}[--]
        \item Apply \Cref{lem:remove-ghost} to compute the set $U_\rho \subseteq V_\rho$ of vertices that do not extend to any $H$-copy in $G[V_l : l \in V(H)]$.
        \item For each $\theta \in \set{\theta_1, \theta_2}$, let $J_\theta$ be the set of nodes in the subtree rooted at $\theta$, $H_\theta := H[J_\theta \cup \set{\rho}]$ and $I_\theta := I \cap J_\theta$. Recursively run \algolight{$H_\theta, I_\theta, \rho \mid G, \kappa$}. Writing $\sol^\theta$ for the set of solutions to $(H_\theta, I_\theta \cup \set{\rho})$-listing on $G[V_l : l \in J_\theta \cup \set{\rho}]$, the recursive call returns $L_\rho^\theta = \Set{v_\rho \in V_\rho : |\sol^\theta[v_\rho]| \leq \kappa}$ together with $\sol^\theta[v_\rho]$ for each $v_\rho \in L_\rho^\theta$.
        \item Let $U'_\rho := \Set{v_\rho \in L_\rho^{\theta_1} \cap L_\rho^{\theta_2} : 1 \leq \Card{\sol^{\theta_1}[v_\rho]} \cdot \Card{\sol^{\theta_2}[v_\rho]} \leq \kappa}$.
        \item Output $L_\rho := U_\rho \cup U'_\rho$. Then for each $v_\rho \in U_\rho$, output $\emptyset$. Finally, for each $v_\rho \in U'_\rho$, output $\sol^{\theta_1}[v_\rho] \times \sol^{\theta_2}[v_\rho]$.
    \end{itemize}
        
    By definition $U_\rho = \Set{v_\rho \in V_\rho : |\sol[v_\rho]| = 0}$. Next we analyze $U_\rho'$. Note that $\sol[v_\rho] = \sol^{\theta_1}[v_\rho] \times \sol^{\theta_2}[v_\rho]$ for all $v_\rho \in V_\rho$. In particular, $|\sol[v_\rho]| = \Card{\sol^{\theta_1}[v_\rho]} \cdot \Card{\sol^{\theta_2}[v_\rho]}$. Hence $1 \leq |\sol[v_\rho]| \leq \kappa$ if and only if both terms and the product itself are in the range $[1,\kappa]$. In other words, $U'_\rho = \Set{v_\rho \in V_\rho : 1 \leq |\sol[v_\rho]| \leq \kappa}$. As a result, $L_\rho = U_\rho \cup U'_\rho = \Set{v_\rho \in V_\rho : |\sol[v_\rho]| \leq \kappa}$ as desired.
        
    Let us analyze the running time. Computing $U_\rho$ takes time $O(n^2)$ by \Cref{lem:remove-ghost}. The two recursive calls run in the stated time by induction. Computing $U'_\rho$ and $L_\rho$ takes time $O(n)$. Outputting for each $v_\rho \in L_\rho$ takes time $O(|\sol[v_\rho]|) \leq O(\kappa)$, so overall we need time $O(n\kappa)$. The total time is thus as claimed.
    
    One can check that the constant hidden in the $\tOh$ notation is polynomial in $k$.
\end{proof}

\begin{lemma}
    \label{lem:reachability}
    Let $H$ be a rooted binary tree with root $\rho$ and leaves $I$. Given an $m$-edge host graph~$G$ and a tuple $\tup{w} \in \prod_{i \in I} V_i$, there is an algorithm \algoreach{$H,I,\rho \mid G,\tup{w}$} that computes in time $O(m)$ the set of all $v_\rho \in V_\rho$ such that $(v_\rho,\tup{w})$ is a solution to $(H,I \cup \set{\rho})$-Listing on $G$.
    The hidden constant in the $O$-notation does not depend on $k$.
\end{lemma}

\begin{proof}
    We implement \algoreach{$H,I,\rho \mid G,\tup{w}$} as a recursive algorithm.

    If $\rho$ is a leaf, then we have $\tup{w} = (v_\rho)$ for some $v_\rho \in V_\rho$. Return the singleton $\set{v_\rho}$.
    
    Otherwise, for every child $\theta$ of $\rho$, let $J_\theta$ be the nodes in the subtree rooted at $\theta$. Let $H_\theta := H[J_\theta]$ and $I_\theta := I \cap J_\theta$. Compute $U_\theta := \text{\algoreach{$H_\theta, I_\theta, \theta \mid G, \tup{w}_{I_\theta}$}}$. Return the set of all $v_\rho \in V_\rho$ such that for each $\theta$ the vertex $v_\rho$~has a neighbor in $U_\theta$.
    
    The correctness of the algorithm is clear. The running time is dominated by the last step, which iterates over all edges between $V_\rho$ and $V_\theta$, for every neighbor $\theta$ of $\rho$. Over all levels of the recursion, the total time is linear in $m$.
\end{proof}

Now we are ready to solve the main problem $(H,I)$-Listing. The key idea is to grow an initially empty subset of leaves $J \subseteq I$, while maintaining a small \emph{residual set} $R \subseteq \prod_{j \in J} V_j$ such that all solutions not extending $R$ have been listed.

\begin{theorem}[\Cref{thm:main1} on Colored Restricted Tree Listing]
    \label{thm:tree}
    For every rooted full binary tree~$H$, and $I$ its set of leaves, Colored $(H,I)$-Listing can be solved in time $\tOh(n^{3/\alfa} + t)$. 
    The constant hidden by the $\tOh$-notation is polynomial in $k$.
\end{theorem}

\begin{proof}
    Write $H = ([k], \mathcal{E})$ and let $G = (V,E)$ with partition $V = V_1 \cup \ldots \cup V_k$ be the given host graph.
    Denote by $\sol$ the set of solutions to $(H,I)$-Listing on $G$.
    %
    We let $\kappa \geq 1$ be a parameter to be fixed later.
    For each node $i \in [k]$, we introduce the following notation:
    \begin{itemize}
        \item $J_i$ is the set of nodes in the subtree of $H$ rooted at $i$;
        \item $H_i := H[J_i]$ is the subtree rooted at $i$;
        \item $I_i := I \cap J_i$ is the set of leaves in $H_i$;
        \item $R_i := \Set{ \tup{u} \in \prod_{l \in I_i} V_l \;:\; |\sol[\tup{u}]| > \kappa}$ is the set of tuples restricted to the leaves $I_i$ that extend to more than $\kappa$ global solutions.
    \end{itemize}

    We devise a recursive algorithm \algolist{$i, \kappa$} with guarantee ($\star$): The algorithm lists all solutions in $\sol[\tup{u}]$ for all $\tup{u} \in \left(\prod_{l \in I_i} V_l\right) \setminus R_i$, and returns the set $R_i$. The algorithm works as follows:
    \begin{description}[leftmargin=5pt]
        \item[Case 1:] $i$ is a leaf.
        \begin{itemize}[--]
            \item Let $H'$ be the reorientation of $H$ such that $i$ becomes the root. Call \algolight{$H',I \setminus \set{i},i \mid G, \kappa$} and obtain a set $L \subseteq V_i$ as well as $\sol[v_i]$ for all $v_i \in L$.
            \item For each $v_i \in L$, print $\sol[v_i]$. Then return $V_i \setminus L$.
        \end{itemize}
        \item[Case 2:] $i$ has two children $a, b$.
        \begin{itemize}[--]
            \item Recursively call $R_a := \text{\algolist{$a, \kappa$}}$ and $R_b := \text{\algolist{$b, \kappa$}}$. Let $S := R_a \times R_b$.
            \item Compute $X_{\tup{u}} := \text{\algoreach{$H[J_a\cup\{i\}], I_a, i \mid G, \tup{u}$}}$ for each $\tup{u} \in R_a$.
            \item Compute $Y_{\tup{v}} := \text{\algoreach{$H[J_b\cup\{i\}], I_b, i \mid G, \tup{v}$}}$ for each $\tup{v} \in R_b$.
            \item Let $H'$ be the reorientation of $H - H_a - H_b$ such that $i$ is the root. Call \algolight{$H',I \setminus (I_a \cup I_b),i \mid G, \kappa$} and obtain a set $L \subseteq V_i$ as well as $\sol'[v_i]$ for all $v_i \in L$. Here, $\sol'$ denotes the set of solutions to $(H',I \setminus (I_a \cup I_b) \cup \set{i})$-Listing on $G$.
            \item Build a matrix $A \in \set{0,1}^{S \times L}$ where $A_{(\tup{u},\tup{v}), v_i} = 1$ if and only if $v_i \in X_{\tup{u}} \cap Y_{\tup{v}}$.
            \item Build a matrix $B \in \set{0,1}^{L \times \prod_{l \in I \setminus (I_a \cup I_b)} V_l}$ where $B_{v_i, \tup{w}} = 1$ if and only if $(v_i,\tup{w}) \in \sol'[v_i]$.
            \item Compute $C := \text{\algomult{$A,B,\kappa$}}$.
            \item Compute $L_i := \Set{(\tup{u},\tup{v}) \in \rows(C) : X_{\tup{u}} \cap Y_{\tup{v}} \subseteq L }$.
            \item For each $(\tup{u},\tup{v}) \in L_i$, print all $(\tup{u},\tup{v},\tup{w})$ such that $C_{(\tup{u},\tup{v}),\tup{w}} > 0$.
            \item Return $S \setminus L_i$.
        \end{itemize}
    \end{description}
    
    Let us argue that \algolist{$i,\kappa$} satisfies the claimed guarantee ($\star$). In case~1 this follows from \Cref{lem:partial-list}. In case~2, for the sake of analysis, we extend the matrices $A,B$ to $A^+ \in \set{0,1}^{S \times V_i}$ and $B^+ \in \set{0,1}^{V_i \times (\prod_{l \in I \setminus (I_a \cup I_b)} V_l)}$, with the same entry definition as written in the algorithm. By definition and \Cref{lem:reachability}, we have $(A^+B^+)_{(\tup{u},\tup{v}),\tup{w}} > 0$ if and only if $(\tup{u},\tup{v}, \tup{w}) \in \sol$.
    
    Recall from \Cref{lem:partial-list} that $L$ is the set of vertices $v_i$ that extend to at most $\kappa$ solutions in $\sol'$. In other words, $L$ is the set of light rows of $B^+$. Since $A,B$ are the restrictions of $A^+,B^+$ to $L$, \Cref{lem:mm-restriction} implies that
    \begin{enumerate}[(i)]
        \item the set of light rows in $A^+B^+$ is $\{(\tup{u},\tup{v}) : (\tup{u},\tup{v}) \allowbreak \text{ is a light row in $AB$ and $A^+_{(\tup{u},\tup{v}),v_i} = 0$ for all $v_i \notin L$}\}$;
        \item $A^+B^+$ and $AB$ agree on these rows.
    \end{enumerate}
    
    The algorithm computes a matrix $C$ that is equal to $AB$ restricted to its light rows by \Cref{lem:sparse-mm}. Hence $L_i = \{ (\tup{u},\tup{v}): (\tup{u},\tup{v}) \text{ is a light row in } AB \text{ and } X_{\tup{u}} \cap Y_{\tup{v}} \subseteq L \}$. By the definitions of $X_{\tup{u}}$ and $Y_{\tup{v}}$, the condition $X_{\tup{u}} \cap Y_{\tup{v}} \subseteq L$ is equivalent to $A^+_{(\tup{u}, \tup{v}),v_i} = 0$ for all $v_i \notin L$. Therefore, $L_i$ is exactly the set of light rows in $A^+B^+$ by (i). That is, $L_i = \Set{(\tup{u},\tup{v}) \in S : |\sol[(\tup{u},\tup{v})]| \leq \kappa}$. Meanwhile, the algorithm lists all non-zero entries in $C$, which by (ii) are the non-zero entries in $A^+B^+$ restricted to the light rows, which in turn correspond to $\sol[(\tup{u},\tup{v})]$ for each $(\tup{u},\tup{v}) \in L_i$.
    
    It remains to argue that $S \setminus L_i = R_i := \{(\tup{u}, \tup{v}) \in  \prod_{l \in I_i} V_l : |\sol[(\tup{u},\tup{v})]| > \kappa\}$. It is clear that $S \setminus L_i \subseteq R_i$. For the other direction, all tuples in $R_i$ must satisfy $\tup{u} \in R_a$, $\tup{v} \in R_b$, and $|\sol[(\tup{u},\tup{v})]| > \kappa$ simultaneously.
    
    We have thus shown that \algolist{$i,\kappa$} satisfies the claimed guarantee ($\star$). 

    \medskip
    Calling \algolist{$i,\kappa$} on the root of $H$, the algorithm returns $R := \Set{\tup{u} \in \prod_{l \in I} V_l : |\sol[\tup{u}]| > \kappa}$. However, $|\sol[\tup{u}]| \leq 1$ for all tuples $\tup{u} \in \prod_{l \in I} V_l$. This implies that $R = \emptyset$. By the guarantee ($\star$), the algorithm must have listed all solutions.

    \medskip
    It remains to analyze the time complexity. In what follows we will bound the running time assuming $t^{2/3} \leq  \kappa \leq 2t^{2/3}$. As the actual value of $t$ is unknown, we have to run \algolist{} on logarithmically many scales $\kappa = 1, 2, 4, \dots, n^k$ in parallel, and output the listed solutions of the run that terminates first. This worsens the analysis below by only an $O(\log n)$ factor.
    
    We focus on case~2, as the cost for case~1 is dominated by case~2. 
    Since each tuple in $R_a$ and $R_b$ extends to more than $\kappa$ global solutions, and these global solutions are distinct, we have $|R_a|, |R_b| < t/\kappa \le t^{1/3} \le \sqrt{\kappa}$. As a result, $|S| = |R_a| \cdot |R_b| \leq \kappa$.
    
    For each $\tup{v} \in R_a$, computing $X_{\tup{v}}$ takes time $O(m) \leq O(n^2)$ by \Cref{lem:reachability}. Similarly, for each $\tup{w} \in R_b$, computing $Y_{\tup{w}}$ takes time $O(n^2)$. Hence, computing all these sets takes time $O(n^2\sqrt{\kappa})$ in total. The subroutine \algolight{} takes time $\tOh\left(n^{2-\alfa} \kappa + n^2 \kappa^{1-\alfa}\right)$ by \Cref{lem:partial-list}. The matrix $A$ has $|S| \leq \kappa$ rows and $n$ columns, so it can be built in $O(n\kappa)$ time. The matrix $B$ has $n$ rows, and each row contains at most $\kappa$ non-zeros by the definition of $L$, so it can be built in $O(n \kappa)$ time as well. The subroutine \algomult{} takes time $\tOh(n\kappa + t + n t^{1-\alfa} \kappa^{\alfa})$ by \Cref{lem:sparse-mm} since $|AB| \le t$ and $A$ has at most $\kappa$ rows. Computing the set $L_i$ takes time $O(n\kappa)$, and printing solutions takes time $O(t)$.
    
    Gathering all terms and recalling $\kappa = \Theta(t^{2/3})$, $0 < \alfa \leq 1$, the running time is
    \begin{align*}
        &\tOh\left( n^2 \kappa^{1/2} + n^{2-\alfa} \kappa + n^2 \kappa^{1-\alfa} + nt^{1-\alfa}\kappa^\alfa + t \right)\\
        &= \tOh\left( n^2 t^{1/3} + n^{2-\alfa} t^{2/3} + n^2 t^{2(1-\alfa)/3} + nt^{1-\alfa/3} + t \right).
    \end{align*}
    We bound each term by the weighted AM-GM inequality:
    \begin{align*}
        n^2 t^{1/3} &\leq \frac{2}{3} n^3 + \frac{1}{3} t,\\
        n^{2-\alfa} t^{2/3} &\leq \frac{1}{3} n^{3(2-\alfa)} + \frac{2}{3} t,\\
        n^2 t^{2(1-\alfa)/3} &\leq \frac{1+2\alfa}{3} n^{6/(1+2\alfa)} + \frac{2(1-\alfa)}{3} t,\\
        nt^{1-\alfa/3} &\leq \frac{\alfa}{3} n^{3/\alfa} + \left(1-\frac{\alfa}{3}\right) t.
    \end{align*}
    It is easy to verify that $\frac{3}{x} \geq \max\Set{3, 3(2-x), \frac{6}{1+2x}}$ for all $0 < x \leq 1$. Hence the running time is bounded by $\tOh(n^{3/\alfa} + t)$.
    The algorithm might list some solutions more than once, but we can remove all duplicates by sorting, which blows up the running time by another $\tOh(1)$ factor.

    It can be checked that the dependence on $k$ hidden by $\tOh$ is polynomial.
\end{proof}

\begin{remark}
    \label{rem:tree}
    The proof also bounds the running time by $\tOh( n^2 t^{1/3} + n^{2-\alfa} t^{2/3} + n^2 t^{2(1-\alfa)/3} + nt^{1-\alfa/3} + t )$. This form is closer to the running time bound of Hu's algorithm~\cite{Hu25} and therefore allows a finer comparison. For example, if $\alfa = 1$ then our algorithm runs in time $\tOh(n^2 t^{1/3} + t)$, while Hu's algorithm runs in time $\tOh(mt^{5/6} + t)$.
\end{remark}

\section{From Restricted Trees to General Trees}
\label{sec:reductions1}

The previous section forms the core of our algorithm, showing that if $H$ is a rooted full binary tree and $I$ is its set of leaves then $(H,I)$-Listing is in time $\tOh(t + n^{3/\alfa})$. In this and the next sections we generalize it further and further by a sequence of reductions. By the end of this section, we will be able to handle general projected trees, i.e., $H$ is any tree and $I$ is any subset of its nodes. Parts of this reduction are similar to arguments by Hu~\cite{Hu24}.

As our first step, we generalize from rooted full binary trees to general trees $H$ while insisting that $I$ is still the set of leaves. In the reduction, we replace nodes with more than two children by a path of nodes each of which has two children (for each new edge in~$H$, we introduce a perfect matching in the host graph), and we attach a new dummy child to each node with one child (corresponding to a fully connected singleton vertex in the host graph).

\begin{lemma}
    \label{lem:reduce-arity}
    Assume that for every rooted full binary tree $H$ and $I$ being its set of leaves we can solve $(H,I)$-Listing in time $f(n,m,t)$. Then for every tree $H$ and $I$ being its set of leaves we can solve $(H,I)$-Listing in time $O(n + m) + f(kn,km,t)$.
\end{lemma}

\begin{proof}
    Let $H = ([k], \mathcal{E})$ be a tree and $G = (\bigcup_{i \in [k]} V_i, E)$ be a host graph.
    
    If $k = 1$ then $H$ is a single node, so the problem is about listing the vertices of $G$, which can be solved in time $O(n)$. If $k = 2$ then $H$ is a single edge, so the problem is about listing the edges of $G$, which can be solved in time $O(m)$.
    
    From now on we assume $k \ge 3$. Then there exists a node $\rho \in [k]$ that is not a leaf of $H$. We pick $\rho$ as the root. Note that $\rho$ has at least two children, so the rooted and unrooted version of~$H$ have the same set of leaves $I$. 

    In what follows, we turn the rooted tree $H$ with leaves $I$ step by step into a rooted full binary tree $H'$ with leaves $I'$. We also turn the host graph $G$ accordingly to $G'$. We will ensure that (i) $(H,I)$-Subgraphs of $G$ are in one-to-one correspondence with $(H',I')$-Subgraphs of $G'$; (ii) $H'$ consists of $|H'| \le 2|H|$ nodes; and (iii) $G'$ can be computed from $G$ in time~$O(m)$.

    Write $c(i)$ for the number of children of node $i \in [k]$. We consider two cases.
    \begin{description}[leftmargin=0pt]
        \item[Case 1:] there is a node $i$ with $c(i) = 1$. Say the child of $i$ is $j$. We construct the new pattern $H'$ by modifying $H$ as follows: We add a new node $i'$ and attach it as a child to $i$. Thus, the set of leaves becomes $I' = I \cup \{i'\}$. We construct the new host graph $G'$ by modifying $G$ as follows: Create a new part $V_{i'} = \{ 0 \}$ containing a single vertex $0$, and put an edge between $0$ and every $v_i \in V_i$. 
        Observe that for every $(H,I)$-Subgraph $\tup{v}$ of $G$, $(v_{i'},\tup{v})$ for $v_{i'} := 0 \in V_{i'}$ is an $(H',I')$-Subgraph of $G'$, and vice versa. In particular, $(H,I)$-Subgraphs of $G$ are in one-to-one correspondence with $(H',I')$-Subgraphs of~$G'$. 
    
        \item[Case 2:] there is a node $i$ with $c(i) \geq 3$. Say the children of $i$ are $j_1, \dots, j_{c(i)}$. We construct the new pattern $H'$ by modifying $H$ as follows: Create a new node $i'$, designate it as the parent of $j_1$ and $j_2$, and designate $i$ as the parent of $i',j_3,\ldots,j_{c(i)}$. This does not change the set of leaves, $I' = I$. We construct the new host graph $G'$ by modifying $G$ as follows: Create a new part $V_{i'}$ as a copy of $V_i$; for each $j \in \set{j_1,j_2}$ rewire the edges between $V_i, V_j$ so that they lie between $V_{i'}, V_j$ instead; and connect every pair of corresponding vertices in $V_i, V_{i'}$ by an edge (note that this adds $O(n)$ edges). Observe that $(H,I)$-Subgraphs of $G$ are in one-to-one correspondence with $(H',I')$-Subgraphs of~$G'$.
    \end{description}
    
    In both cases, the modification reduces $|c(i) - 2|$ by one, does not change $c(x)$ for any other node~$x$, and adds a new node $i'$ with $c(i') \in \{0, 2\}$ children. By repeatedly applying this modification, we eventually obtain a rooted full binary tree. The number of repetitions is exactly $\sum_{i \in [k] : c(i) \geq 1} |c(i)-2| \leq \sum_{i \in [k]} c(i) = k-1$. In each repetition, we add one node to the pattern and one part to the host graph, spending time $O(m)$. The bounds on pattern size and running time thus follow.
\end{proof}

As our next step, we generalize from $I$ being exactly the set of leaves of $H$ to an arbitrary subset of leaves. This reduction removes leaves that are not part of $I$. It is a rephrasing of the \texttt{cleanse} procedure of \cite{Hu24}.
\begin{lemma}
    \label{lem:cleanse}
    Assume that for every tree $H$ and $I$ being its set of leaves we can solve $(H,I)$-Listing in time $f(n,m,t)$. Then for every tree $H$ and every subset of its leaves $I$ we can solve $(H,I)$-Listing in time $O(n+m) + f(n,m,t)$.
\end{lemma}

\begin{proof}
    We use induction on $k=|H|$. For $k = 1$, $(H,I)$-Listing is trivially in time $O(n)$. For $k \geq 2$, if $I$ is exactly the set of leaves of $H$, then the claim follows from the assumption. Otherwise, there exists a leaf $i \notin I$. Let $j$ be the unique neighbor of $i$. We remove all vertices $v_j \in V_j$ that do not have a neighbor in $V_i$, and then solve $(H-i,I)$-Listing on $G' := G[V_l : l \neq i]$. Note that every $(H-i)$-copy in $G'$ extends to an $H$-copy in $G$. Hence, every solution to the new instance is a solution to the old instance, and vice versa. The running time is clearly $O(n+m) + f(n,m,t)$.
\end{proof}

Finally, we generalize to $I$ being any subset of the nodes of $H$. To this end, for any non-leaf $i \in I$ we observe that the problem splits at $i$ into subproblems on the connected components of $H-i$. This lemma is a rephrasing of the \texttt{decompose} procedure of \cite{Hu24}.
\begin{lemma}
    \label{lem:decompose}
    Assume that for every tree $H$ and every subset of its leaves $I$ we can solve $(H,I)$-Listing in time $f(n,m,t)$. Then for every tree $H$ and every subset of its nodes $I$ we can solve $(H,I)$-Listing in time $O(n+m+t) + (k-1) \cdot f(n,m,t)$.
\end{lemma}

\begin{proof}
    We use induction on $k=|H|$. For $k = 1$, $(H,I)$-Listing is trivially in time $O(n)$. For $k \geq 2$, if $I$ is a subset of the leaves, then the claim follows from the assumption. Otherwise, there exists a non-leaf $i \in I$. The node $i$ separates the tree into $r \geq 2$ connected components, whose node sets we denote by $C_1, \dots, C_r$. We do the following:
    \begin{enumerate}[1.]
        \item Apply \Cref{lem:remove-ghost} to remove all vertices in $G$ that do not extend to any solution.
        \item For each $a \in [r]$, let $C_a^+ := C_a \cup \set{i}$ and note that $|C_a^+| \leq k-1$. Inductively solve $(H[C_a^+], I \cap C_a^+)$-Listing in $G[V_j : j \in C_a^+]$ and obtain a set of solutions $\sol_a$. Partition $\sol_a = \bigcup_{v_i \in V_i} \sol_a[v_i]$ where $\sol_a[v_i]$ collects all solutions in $\sol_a$ that use the vertex $v_i$. Importantly, none of $\sol_a[v_i]$ is empty, as otherwise we would have removed $v_i$ in step~1.
        \item Compute $\sol = \bigcup_{v_i \in V_i} \prod_{a \in [r]} \sol_a[v_i]$.
    \end{enumerate}
    
    Clearly, $\sol$ is the set of solutions to $(H,I)$-Listing in $G$, so the algorithm is correct. Let us analyze the time complexity. Step 1 takes time $\tau_1 = O(n+m)$ by \Cref{lem:remove-ghost}. In step 2, each iteration $a$ takes time $O(n+m) + |C_a| \cdot f(n,m,|\sol_a|)$ by induction. Since step 1 has removed irrelevant vertices, each solution to the subinstance extends to at least one solution to the original instance, so $|\sol_a| \leq t$. Summing over all $a$, step 2 takes time
    \begin{align*}
        \tau_2
        \ALIGN\;\leq\; O(n+m) + \sum_{a \in [r]} |C_a| \cdot f(n,m,t) \NEWLINE
        \ALIGN\;=\; O(n+m) + (k-1) \cdot f(n,m,t).
    \end{align*}
    Step 3 takes time $\tau_3 = \sum_{v_i \in V_i} \prod_{a \in [r]} |\sol_a[v_i]|$. Observe that the terms in $\tau_3$ correspond to distinct solutions in $\sol$, thus $\tau_3 \leq t$. Putting everything together, the total running time is bounded by $\tau_1 + \tau_2 + \tau_3 = O(n+m+t) + (k-1) \cdot f(n,m,t)$.
\end{proof}

By chaining the above three lemmas with the algorithm for restricted trees from the previous section, we are able to handle general projected trees:

\begin{theorem}[\Cref{thm:main1} on Colored Listing] \label{thm:main1-colored-listing}
    For every tree $H$ and every $I$, Colored $(H,I)$-Listing can be solved in time $\tOh(n^{3/\alfa} + t)$.
    The constant hidden by $\tOh$ is polynomial in $k$.
\end{theorem}
\begin{proof}
    Combine \Cref{thm:tree} with \Cref{lem:reduce-arity,lem:cleanse,lem:decompose}. It can be checked that these lemmas only blow up the constant in $\tOh$ by a $\poly(k)$ factor.
\end{proof}

\begin{remark}
    \label{rem:main1}
    The running time bound strengthens to $\tOh( n^2 t^{1/3} + n^{2-\alfa} t^{2/3} + n^2 t^{2(1-\alfa)/3} + nt^{1-\alfa/3} + t )$ by plugging in the stronger bound from \Cref{rem:tree}.
\end{remark}

\section{From Trees to Hypergraphs}
\label{sec:reductions2}

We generalize our algorithm further, from trees to arbitrary hypergraphs.

First, we move from trees to acyclic hypergraphs. A hypergraph $H = ([k], \mathcal{E})$ is \EMPH{acyclic} if there exists a tree $T$ with $V(T) = \mathcal{E}$ such that for every $i \in [k]$, the hyperedges in $\mathcal{E}$ containing $i$ form a connected subtree of $T$.

\begin{lemma}
    \label{lem:acyclic-to-tree}
    Assume that for every tree $H$ and every $I$ we can solve $(H,I)$-Listing in time $f(n,m,t)$. Then for every acyclic hypergraph $H$ and every $I$ we can solve $(H,I)$-Listing in time $f(O(m), O(m^2), t) + O(m^2)$. The hidden constant in the $O$-notation is linear in the size of $H$.
\end{lemma}

\begin{proof}
    Let $H = ([k], \mathcal{E})$ be an acyclic hypergraph and $G = \left( \bigcup_{i \in [k]} V_i, E \right)$ a host hypergraph. If a node $i \in [k]$ does not appear in any hyperedge in $\mathcal{E}$, then the vertex $v_i$ can be chosen freely from $V_i$, so the problem reduces to $(H-i, I\setminus\set{i})$-Listing. Therefore, we may assume that every node $i \in [k]$ appears in some hyperedge $\pi(i) \in \mathcal{E}$. We will construct a pattern tree $H'$ and a host graph $G'$ as follows.
    
    First, fix a tree $T$ that witnesses the acyclicity of $H$. Recall that $V(T) = \mathcal{E}$. We add nodes $1, \dots, k$ to $V(T)$, and then connect each $i \in [k]$ with $\pi(i)$. The resulting tree is our $H'$, with $V(H') = [k] \cup \mathcal{E}$.
    
    Second, we construct a corresponding graph $G' = (V',E')$. The vertex set is partitioned into $V' := \bigcup_{i \in [k]} V_i \cup \bigcup_{e \in \mathcal{E}} E_e$ where $E_e := E \cap \prod_{i \in e} V_i$. The edge set $E'$ is defined below:
    \begin{itemize}
        \item For $e, f \in \mathcal{E}$ adjacent in $H'$, put an edge between $(u_i)_{i \in e} \in E_e$ and $(v_i)_{i \in f} \in E_f$ if and only if $u_i = v_i$ for all $i \in e \cap f$.
        \item For $i \in [k]$, put an edge between $u_i \in V_i$ and $(v_j)_{j \in \pi(i)} \in E_{\pi(i)}$ if and only if $u_i = v_i$.
    \end{itemize}
    The graph $G'$ has $O(m)$ vertices, $O(m^2)$ edges, and can be computed in time $O(m^2)$. We solve $(H',I)$-Listing on $G'$ and output all solutions.
    
    Next we show a bijection between $H$-copies in $G$ and $H'$-copies in $G'$. To this end, we map each $H$-copy $(v_1, \dots, v_k)$ to the tuple $(v_1, \dots, v_k, ((v_i)_{i \in e})_{e \in \mathcal{E}})$. The mapping is injective by definition. Also note that the image is an $H'$-copy in $G'$. Indeed, $(v_i)_{i \in e} \in E_e$ for all $e \in \mathcal{E}$ because $(v_1, \dots, v_k)$ is an $H$-copy. Moreover, by definition $((v_i)_{i \in e}, (v_i)_{i \in f}) \in E'$ for all adjacent $e,f \in \mathcal{E}$, and $(v_i, (v_j)_{j \in \pi(i)}) \in E'$ for all $i \in [k]$.
    
    To argue that the mapping is surjective, we consider an $H'$-copy $(u_1, \dots, u_k, ((v^e_i)_{i \in e} )_{e \in \mathcal{E}})$. By definition, if $e,f \in \mathcal{E}$ are adjacent then $v^e_i = v^f_i$ for all $i \in e \cap f$. Now fix any $i \in [k]$ and recall that the hyperedges in $\mathcal{E}$ containing $i$ form a connected subtree. Hence the assignments $v^?_i$ are consistent throughout; that is, $v^e_i = v^f_i$ for all (possibly non-adjacent) $e,f \in \mathcal{E}$ containing $i$. On the other hand, we have $u_i = v_i^{\pi(i)}$ for $i \in [k]$ by definition of $H'$-copies. These facts imply that $(u_i)_{i \in e} = (v^{\pi(i)}_i)_{i \in e} = (v^e_i)_{i \in e} \in E_e \subseteq E$ for all $e \in \mathcal{E}$. In other words, $(u_1, \dots, u_k)$ is an $H$-copy in $G$.
    
    We have shown that the mapping $\phi$ is bijective. We further observe that $\phi$ preserves the projection on $I$, that is, $\tup{v}_I = \phi(\tup{v})_I$. Therefore, solutions to $(H,I)$-Listing in $G$ are exactly solutions to $(H',I)$-Listing in $G'$, so the algorithm is correct and runs in the claimed time.
\end{proof}

Note that the reduction above changes the running time bound, as we replace $n$ by $m$.

For the generalization from acyclic to arbitrary hypergraphs we rely on the seminal PANDA algorithm~\cite{KNS17,KNS25}, see also~\cite{KNS26,KC26}. We use the following phrasing of the PANDA algorithm as a reduction, cf.~\cite[Section 5]{Hu24}. Here, $\subw(H)$ is the submodular width, see \Cref{sec:subw} for a definition.

\begin{theorem}[Rephrasing \cite{KNS17,KNS25}]
    \label{thm:PANDA}
    For every pattern hypergraph $H$, there exist acyclic hypergraphs $H_1,\ldots,H_r$ where $r = r(|V(H)|)$ is a constant, $V(H_i) = V(H)$ for each $i \in [r]$, with the following property. Given a host hypergraph $G$, in time $\tOh(m^{\subw(H)})$ we can compute host hypergraphs $G_1,\ldots,G_r$ such that $\sol = \bigcup_{i \in [r]} \sol_i$, where $\sol$ is the set of colored $H$-copies in $G$ and $\sol_i$ is the set of colored $H_i$-copies in $G_i$.
\end{theorem}

\begin{lemma}
    \label{lem:reduce-PANDA}
    Assume that for every acyclic hypergraph $H$ and every $I$ we can solve $(H,I)$-Listing in time $f(m,t)$. Then for every hypergraph $H$ and every $I$ we can solve $(H,I)$-Listing in time $O\left(f(\tOh(m^{\subw(H)}), t)\right)$.
\end{lemma}

\begin{proof}
    Given a host hypergraph $G$ for $H$, we construct the patterns $H_1,\ldots,H_r$ and hosts $G_1,\ldots,G_r$ from \Cref{thm:PANDA}. 
    For a set of tuples $\sol$ we denote its projection to $I$ by $\sol|_I := \set{ \tup{v}_I : \tup{v} \in \sol}$. 
    Note that the goal of $(H,I)$-Listing is to compute the projection $\sol|_I$, where $\sol$ is the set of $H$-copies in $G$. By \Cref{thm:PANDA}, we can compute this set as $\sol|_I = \bigcup_{i \in [r]} \sol_i |_I$, where $\sol_i$ is the set of $H_i$-copies in $G_i$. 
    That is, we solve $(H_i,I)$-Listing on each $G_i$ and return the union of all listed solutions; this solves $(H,I)$-Listing on $G$. 
    Since $G_1,\ldots,G_r$ can be computed in time $\tOh(m^{\subw(H)})$, their number of edges is bounded by the same term. The running time bound thus follows. 
\end{proof}

Using the above two lemmas, the algorithm for trees from the previous section extends to arbitrary hypergraphs.

\begin{theorem}[\Cref{thm:algohypergraphs} on Colored Listing] \label{thm:algohypergraphs-colored-listing}
    For every hypergraph $H$ and every $I$, Colored $(H,I)$-Listing is in time $\tOh( m^{3\subw(H) / \alfa} + t)$.
\end{theorem}
\begin{proof}
    Combine \Cref{thm:main1-colored-listing} with \Cref{lem:acyclic-to-tree,lem:reduce-PANDA}. The running time changes from $\tOh(n^{3/\alfa} + t)$ via $\tOh(m^{3/\alfa} + t)$ to $\tOh(m^{3\subw(H)/\alfa} + t)$.
\end{proof}

\begin{remark}
    Denoting $M := m^{\subw(H)}$, the running time bound strengthens to $\tOh( M^2 t^{1/3} + M^{2-\alfa} t^{2/3} + M^2 t^{2(1-\alfa)/3} + Mt^{1-\alfa/3} + t )$ by plugging in the stronger bound from \Cref{rem:main1}.
\end{remark}

\section{From Listing to Enumeration}
\label{sec:enumtolist}

For decision and search problems studied in complexity theory, a simple yet fundamental result is a search-to-decision reduction, which we will review in \Cref{sec:searchtodecision}. For listing and enumeration problems, we prove an analogous result, namely an enumeration-to-listing reduction. It works in almost the same generality as the search-to-decision reduction, except that the running time bounds need to be constructible.

\subsection{Technical Setup}
\label{sec:technicalsetup}

Throughout \Cref{sec:enumtolist}, we consider a generic computational problem, where every instance~$I$ induces a search space $\set{0,1}^d$ and a set of solutions $S \subseteq \set{0,1}^d$. We call $d = d(I)$ the \emph{depth}, and write $t = t(I)$ for the number of solutions.

The following definition formalizes that we can split a given instance $I$ by fixing the first bit of the search space to $0$ or $1$, generating corresponding subinstances $I_0$ and $I_1$.

\begin{definition} \label{def:splitting}
A \EMPH{splitting algorithm} is an algorithm that, given an instance $I$ of depth $d > 0$, computes instances $I_0$ and $I_1$ of depth $d-1$ such that for every $b, x_1, \dots, x_{d-1} \in \set{0,1}$ the tuple $(x_1, \dots, x_{d-1})$ is a solution to $I_b$ if and only if $(b, x_1, \dots, x_{d-1})$ is a solution to $I$.

Let \algosplit{} be a splitting algorithm, and let $f$ be a function that maps instances to $\N$. We say that $f$ is \EMPH{\algosplit{}-monotone} if $f(I_0) \le f(I)$ and $f(I_1) \le f(I)$ for all $I$ and $(I_0, I_1) = \text{\algosplit{$I$}}$.
\end{definition}

\subsection{Classic Search-to-Decision Reduction} \label{sec:searchtodecision}

Let us start by reviewing the classic search-to-decision reduction. A \emph{decision algorithm} decides whether the given instance $I$ has a solution, i.e., it returns true if $t > 0$ and false otherwise. A \emph{search algorithm} returns an arbitrary solution if there exists one, and $\bot$ otherwise. Clearly, a search algorithm running in time $f(I)$ yields a decision algorithm running in time $O(f(I))$. The following folklore result shows a converse statement.

\begin{theorem}[Search-to-Decision Reduction]
\label{thm:searchtodec}
Let \algosplit{} be a splitting algorithm and $f,s$ be \algosplit{}-monotone functions. If there is a decision algorithm in time $O(f(I))$ and \algosplit{} runs in time $O(s(I))$, then there is a search algorithm in time $O\left((f(I) + s(I)) \cdot d\right)$.
\end{theorem}

\begin{proof}
For a given instance $I$ of depth $d$, we first run the decision algorithm to check if $I$ has a solution; if not then we return $\bot$. Otherwise $I$ has at least one solution, and we describe a recursive function that returns a solution:
\begin{pseudocode}
\Function{search}{$I$}
    \LComment{finds a solution to $I$, assuming one exists}
    \If{$d(I) = 0$}
        \State \Return the empty string
    \EndIf
    \State construct subinstances $(I_0, I_1) := \algosplit{I}$
    \State use the decision algorithm to decide whether $I_0$ has a solution
    \If{$I_0$ has a solution}
        \State \Return \algosearch{$I_0$} prepended by $0$
    \Else
        \State \Return \algosearch{$I_1$} prepended by $1$
    \EndIf
\EndFunction
\end{pseudocode}

Correctness is immediate. We show by induction on the depth $d = d(I)$ that the algorithm runs in time $O((f(I) + s(I)) \cdot d)$. For $d = 0$ the claim is trivial. For $d \geq 1$, the subinstances are constructed in time $O(s(I))$, and the decision algorithm runs in time $O(f(I))$. The recursive call takes time $O\left(\max_{b \in \set{0,1}}(f(I_b) + s(I_b)) \cdot (d-1)\right)$ by induction. Since both $f$ and $s$ are \algosplit{}-monotone, we have $f(I_b) \leq f(I)$ and $s(I_b) \leq s(I)$, so the time is $O\left((f(I) + s(I)) \cdot (d-1)\right)$. Summing up all terms, the algorithm runs in time $O\left( (f(I) + s(I)) \cdot d \right)$ as claimed.
\end{proof}

The proof demonstrates that the notion of splitting algorithm is indeed natural, as it allows the search procedure to fix a bit in the search space. The requirements that $f$ and $s$ are \algosplit{}-monotone are necessary to avoid blowing up the running time of the decision algorithm and the splitting algorithm in the recursion.

\subsection{Enumeration-to-Listing Reduction}

Analogously to the search-to-decision reduction, we prove an enumeration-to-listing reduction. We start by introducing different meanings of listing and enumeration.

We work with the same technical setup defined in \Cref{sec:technicalsetup}.
A \EMPH{listing algorithm} is given an instance~$I$ and computes the set of solutions. A \EMPH{strong listing algorithm} is given an instance $I$ and a number~$t'$ and computes $\min\set{t, t'}$ solutions to $I$. For listing algorithms and strong listing algorithms we are interested in bounding the total running time.

An \EMPH{enumeration algorithm} has the same correctness guarantee as a listing algorithm, but instead of studying the total running time, we focus on bounding the \EMPH{preprocessing time} (the time until the first solution is printed) and the \EMPH{delay} (the maximum time between printing two consecutive solutions). The time after printing the last solution can also be assumed to be bounded by the delay.

It is clear that any enumeration algorithm with preprocessing time $f(I)$ and delay $g(I)$ yields a strong listing algorithm running in time $O(f(I) + \min\set{t', t+1} \cdot g(I)) \le O(f(I) + g(I) + \min\set{t', t} \cdot g(I))$, by simply running the enumeration algorithm until it has printed $t'$ solutions (or until it has printed all $t$ solutions and terminates). Moreover, a strong listing algorithm running in time $O(f(I) + \min\set{t, t'} \cdot g(I))$ yields a listing algorithm running in time $O(f(I) + t \cdot g(I))$, by simply setting $t' := 2^d$. Our result is a converse statement.

\begin{definition}
Let $f_1, \dots, f_l$ be functions that map instances to $\N$. We say that $(f_1, \dots, f_l)$ are \EMPH{constructible} if each $f_i$ can be evaluated on any given instance $I$ in time $O\left(\sum_{j \in [l]} f_j(I)\right)$.
\end{definition}

\begin{theorem}[Enumeration-to-Listing Reduction]
\label{thm:enumtolist}
Let $\beta \in [0,1]$ be a constant. Let \algosplit{} be a splitting algorithm. Let $(f,g,s)$ be constructible, \algosplit{}-monotone functions.
If there is a listing algorithm in time $O(f(I) + t \cdot g(I))$ and \algosplit{} runs in time $O(s(I))$, then there is an enumeration algorithm with preprocessing time $O((f(I) + g(I) + s(I)) \cdot d^{2 - \beta})$ and delay $O(g(I) \cdot d^\beta)$.
\end{theorem}

We stress that \Cref{thm:enumtolist} assumes that the given listing algorithm is deterministic, and the resulting enumeration algorithm is also deterministic.

The choice of $\beta \in [0,1]$ describes a tradeoff. At one extreme $\beta = 1$, we obtain preprocessing time $O((f(I) + g(I) + s(I)) \cdot d)$ and delay $O(g(I) \cdot d)$; at the other extreme $\beta = 0$, we obtain preprocessing time $O((f(I) + g(I) + s(I)) \cdot d^2)$ and delay $O(g(I))$.

Let us discuss the assumptions of \Cref{thm:enumtolist}. First, the theorem assumes splittability (i.e., there is a splitting algorithm \algosplit{} and $f,g,s$ are \algosplit{}-monotone). This assumption is well-motivated because it also appears in the search-to-decision reduction (\Cref{thm:searchtodec}), and because Colored $(H,I)$-Subhypergraph satisfies it, as we will verify soon.
Second, the theorem assumes constructibility of $(f,g,s)$ in order to set certain parameters in the reduction. This assumption is also satisfied by typical listing algorithms, as most time bounds are simple functions of the input size $n$ (which under typical circumstances can be read from the instance in constant time) and thus are easy to compute.

\subsection{Application to $(H,I)$-Subhypergraphs}

We next apply our general enumeration-to-listing reduction to the Colored $(H,I)$-Subhypergraph problem, thus proving \Cref{thm:main2-simplified} as a corollary of \Cref{thm:enumtolist}.

\MainSimplifiedTheorem*

\begin{proof}
    We assume that the host hypergraph $G$ has vertex set $\set{0,\ldots,n-1}$; we can ensure this by a simple $O((n+m) \log n)$-time preprocessing.
    So each vertex corresponds to a $\ceil{\log_2 n}$-bit number, and a solution corresponds to a $d$-bit number for $d := |I| \cdot \ceil{\log_2 n}$. Thus, the search space is $\set{0,1}^d$, which fits our technical setup.

    With respect to this search space, the Colored $(H,I)$-Listing problem has a simple splitting algorithm \algosplit{}. Indeed, fixing a bit of the search space to $b \in \set{0,1}$ corresponds to removing from $V_i$ all vertices whose $j$-th bit is not equal to $b$ (for some easily computable $i$ and $j$). In particular, this is again a valid instance. To perform this removal, it suffices to iterate once over all vertices and hyperedges to filter out the surviving ones. This runs in time $O(n + m)$; for concreteness suppose that the running time is bounded by $s(G) := c \cdot (n(G) + m(G))$ for some constant $c > 0$. Note that \algosplit{} computes a subhypergraph $G_b$ of $G$ for each $b \in \set{0,1}$, and thus $n(G_b) \le n(G)$ and $m(G_b) \le m(G)$.

    By assumption, $f(n,m)$ and $g(n,m)$ can be computed in time $O(f(n,m) + g(n,m))$, and $s(n,m)$ can clearly be computed in time $O(s(n,m))$. Hence $(f,g,s)$ are constructible. The functions are \algosplit{}-monotone because they are monotone in $n,m$, which can only become smaller in the subinstances. Hence, the requirements of \Cref{thm:enumtolist} are satisfied. With $\beta := 0$, it yields the claimed enumeration algorithm since $d = O(\log n)$.
\end{proof}

Now we can prove the statements in \Cref{thm:main1,thm:algohypergraphs} on Colored $(H,I)$-Enumeration.

\begin{theorem}[\Cref{thm:main1} on Colored Enumeration] \label{thm:main1-colored-enumeration}
    For every tree $H$ and every $I$, Colored $(H,I)$-Enumeration can be solved in preprocessing time $\tOh(n^{3/\alfa})$ and delay $\tOh(1)$.
    The constant hidden in $\tOh$-notation is polynomial in $k$.
\end{theorem}
\begin{proof}
    Combine \Cref{thm:main2-simplified} with \Cref{thm:main1-colored-listing}. One can check that this blows up the constant hidden in $\tOh$-notation by a polynomial in $k$.
\end{proof}

\begin{theorem}[\Cref{thm:algohypergraphs} on Colored Enumeration] \label{thm:algohypergraph-colored-enumeration}
    For every hypergraph $H$ and every $I$, Colored $(H,I)$-Enumeration can be solved in preprocessing time $\tOh(m^{\subw(H) \cdot 3/\alfa})$ and delay $\tOh(1)$.
\end{theorem}
\begin{proof}
    Combine \Cref{thm:main2-simplified} with \Cref{thm:algohypergraphs-colored-listing}.
\end{proof}

\begin{remark}
    In \Cref{thm:main2-simplified}, we assumed that $f$ and $g$ are monotone functions in terms of $n$ and~$m$. Beyond the parameters $n$ and $m$, the same argument works for any subhypergraph-monotone parameters, e.g., the maximum degree $\Delta = \Delta(G)$. If $f$ and $g$ are monotone functions in terms of subhypergraph-monotone parameters such that $(f,g)$ are constructible, and Colored $(H,I)$-Listing is in time $O(f(G) + t \cdot g(G))$, then \Cref{thm:enumtolist} yields an enumeration algorithm with preprocessing time $O((f(G) + g(G) + n + m) \cdot d^{2 - \beta})$ and delay $O(g(G) \cdot d^\beta)$, for any constant $\beta \in [0,1]$. 
\end{remark}

\subsection{Proof Preparations}
The proof of \Cref{thm:enumtolist} makes use of the Cheater's Lemma~\cite{CK21}, which provides ways in which an improper enumeration algorithm may cheat so that it can still be converted into a proper enumeration algorithm. The original statement of this lemma uses very specific running time bounds, therefore here we rework it to allow more general time bounds. (Later in this paper we will need yet another variant, which we call variant II.)

\begin{lemma}[Cheater's Lemma, Variant I] \label{lem:cheaterI}
Let $(f,g)$ be constructible functions with $f(I) \ge g(I)$ for all~$I$. Suppose there is a listing algorithm that prints in time $f(I) + j \cdot g(I)$ at least $j$ solutions, for all $1 \le j \le t$. Then there is an enumeration algorithm with preprocessing time $O(f(I))$ and delay $O(g(I))$.
\end{lemma}

\begin{proof}
We augment the listing algorithm with a counter that is incremented after each step, so it counts the number of steps performed by the algorithm. We run this augmented algorithm, gathering its listed solutions. After the $(f(I) + j \cdot g(I))$-th time step, we print the $j$-th gathered solution. If at this point less than $j$ solutions have been gathered, then we terminate. 

Note that to run this algorithm we need to know $f(I)$ and $g(I)$, and we can indeed compute $f(I)$ and $g(I)$ in time $O(f(I) + g(I))$, since $(f,g)$ are constructible.

The guarantee of the listing algorithm ensures that $j$ solutions are gathered in time at most $f(I) + j \cdot g(I)$, so after the $(f(I) + j \cdot g(I))$-th time step we can print the $j$-th gathered solution. This holds for any $1 \le j \le t$. Thus, if at this point less than $j$ solutions have been gathered, then $j > t$, in which case we have already listed all $j-1 = t$ solutions and correctly terminate.

Clearly, the first solution is printed in time $O(f(I) + g(I)) = O(f(I))$ and the time between printing two solutions is $O(g(I))$. The time between printing the last solution and terminating is also $O(g(I))$. That is, we obtained an enumeration algorithm with preprocessing time $O(f(I))$ and delay $O(g(I))$.
\end{proof}

For the remainder of \Cref{sec:enumtolist} we fix a splitting algorithm \algosplit{}. We now introduce some notation by considering the following trivial listing algorithm.

\begin{pseudocode}
\Function{trivialList}{$J, u$}
    \LComment{lists all solutions to instance $J$, prepended by $u$}
    \If{$d(J) = 0$}
        \If{$J$ has a solution}
        	\State print $u$
        \EndIf
    \Else
        \State $(J_0, J_1) := \algosplit{J}$
        \State \algotrivlist{$J_0, u0$}
        \State \algotrivlist{$J_1, u1$}
    \EndIf
\EndFunction
\end{pseudocode}

Let $I$ be a given instance of depth $d$ and consider a call to \algotrivlist{$I, \epsilon$}, where $\epsilon$ is the empty string. The algorithm uses \algosplit{} to construct subproblems $I_0, I_1$ and recurse on both. In the base case, we have a subproblem of depth $0$, which is either satisfied or unsatisfied. If it is satisfied, then the path $u \in \set{0,1}^d$ we have taken to reach this subproblem is a solution to the original instance $I$, otherwise it is not. Accordingly, the algorithm lists all solutions to~$I$. 

Note that with the setup as described so far, it is not clear how to check whether an instance of depth 0 is satisfied. Later we will have access to a listing algorithm, which can perform this task for us.

We introduce some additional notation that helps to analyze subsequent variants of this algorithm. We denote by $\set{0,1}^{\le d}$ the bitstrings of length at most $d$, including the empty string $\epsilon$. We denote the subproblems constructed by the trivial listing algorithm by $I_u$, for $u \in \set{0,1}^{\le d}$; they can be defined by setting $I_\epsilon := I$ and $(I_{u0},I_{u1}) := \text{\algosplit{$I_u$}}$ for any $u \in \set{0,1}^{\le d-1}$. We denote the number of solutions to subproblem $I_u$ by $t(u) := t(I_u)$, and its depth by $d(u) := d(I_u) = d - |u|$. 

The bitstrings $u \in \set{0,1}^{\le d}$ are in one-to-one correspondence with the nodes of the perfect binary tree of height $d$ and $2^d$ leaves, where $u0$ (respectively $u1$) is the left (respectively right) child of $u$. For any $u \in \set{0,1}^{\le d}$, we let $A(u)$ be the set of ancestors of $u$ in this tree (including $u$ itself); in other words, $A(u)$ is the set of prefixes of the bitstring $u$.

\subsection[Warmup: \Cref{thm:enumtolist} for $\beta = 1$]{Warmup: \Cref{thm:enumtolist} for \boldmath$\beta = 1$}

To build intuition for the general argument, in this section we present a simplified proof for the special case of \Cref{thm:enumtolist} where $\beta = 1$. Readers primarily interested in the full generality of \Cref{thm:enumtolist} may skip directly to \Cref{sec:proofenumtolist}.
We start with a simple tool.

\begin{lemma}
\label{lem:budgetlistsimple}
Let $(f,g)$ be constructible functions. If there is a listing algorithm running in time at most $f(I) + t \cdot g(I)$, then there is an algorithm \algobudgetlist{$I, b$} that, given an instance $I$ and an integer $b \ge 1$, in time $O(f(I) + b \cdot g(I))$ either lists all solutions to $I$ or reports that it has more than $b$ solutions.
\end{lemma}

\begin{proof}
We augment the listing algorithm with a counter that counts the number of steps performed by the algorithm. After $f(I) + b \cdot g(I)$ steps we abort and report that there are more than $b$ solutions. If the algorithm terminates before we abort it, then we print all solutions that it has listed.

Note that to run this algorithm we need to know $f(I)$ and $g(I)$, and we can indeed compute them in time $O(f(I) + g(I))$ due to constructibility.

To see correctness, note that if we abort the algorithm then it took more than $f(I) + b \cdot g(I)$ steps, so by its running time guarantee there must be more than $b$ solutions, so we correctly report that there are more than $b$ solutions. If we don't abort the algorithm, then it finished and thus successfully listed all solutions.
\end{proof}

The following lemma essentially proves the extreme case $\beta = 1$ of \Cref{thm:enumtolist}.

\begin{lemma}
\label{lem:enumtolistsimple}
Let $(f,g)$ be constructible, \algosplit{}-monotone functions with $f(I) \ge g(I)$ for all $I$. If there is a listing algorithm running in time at most $f(I) + t \cdot g(I)$, and \algosplit{} runs in time at most $f(I)$, then there is an enumeration algorithm with preprocessing time $O(f(I) \cdot d)$ and delay $O(g(I) \cdot d)$.
\end{lemma}

\begin{proof}
We first present a ``cheating'' enumeration algorithm in the sense of \Cref{lem:cheaterI}.

\begin{pseudocode}
\Function{cheatEnum}{$J, u$}
    \LComment{lists all solutions to instance $J$, prepended by $u$}
    \If{\algobudgetlist{$J, b$} lists all solutions to $J$}
        \State prepend $u$ to each one and print them
    \Else
        \State $(J_0, J_1) := \protect\algosplit{J}$
        \State \algocheatenum{$J_0, u0$}
        \State \algocheatenum{$J_1, u1$}
    \EndIf
\EndFunction
\end{pseudocode}

To list solutions to an instance $I$, we compute $f(I),g(I)$, set the parameter $b := \ceil{f(I)/g(I)}$ globally, and then call \algocheatenum{$I, \epsilon$}. Observe that all solutions are listed because the algorithm proceeds similarly as the trivial listing algorithm, except that a subproblem $J$ returns immediately if \algobudgetlist{} has listed all of its solutions. Note that for each recursive call ($J, u$) we have $J = I_u$, and thus prepending $u$ to the solutions to $J$ yields solutions to the original instance $I$. 

Fix any $1 \le j \le t$. Our goal is to bound the time until the algorithm prints its $j$-th solution. So consider the moment when the $j$-th solution is printed. This happens during the recursive call \algocheatenum{$I_{u^*}, u^*$} for some $u^* \in \set{0,1}^{\le d}$. Note that the set of ancestors $A(u^*)$ is the set of recursive calls that have started but not yet finished at the moment.
Let $R \subseteq \set{0,1}^{\le d}$ be the set of all started recursive calls, i.e., $u \in \set{0,1}^{\le d}$ is in $R$ if \algocheatenum{$I_u, u$} has started before printing the $j$-th solution.
Let $L \subseteq R$ be the leaves of the recursion, i.e., $u \in R$ is in $L$ if neither $u0$ nor $u1$ is in~$R$.
Let $R' := R \setminus L$, and let $L'$ be the leaves of the resulting subtree, i.e., $u \in R'$ is in $L'$ if neither $u0$ nor $u1$ is in $R'$.
See \Cref{fig:recursion-tree} for an illustration.

\begin{figure}[hbt]
    \centering
    \includegraphics{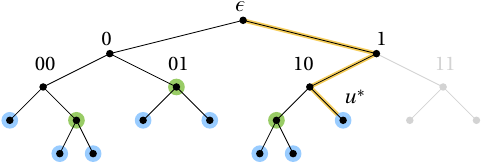}
    \caption{The recursion tree has root $\epsilon$ and nodes labeled by bitstrings. The tree $R$ (drawn in black) corresponds to the recursive calls that have started. The ancestors $A(u^*)$ (drawn as the yellow path) correspond to the recursive calls that have started but not yet finished. The leaves $L$ of $R$ are marked in blue. The leaves $L'$ of $R' = R \setminus L$ are marked in green.}
    \label{fig:recursion-tree}
\end{figure}

\begin{claim}
$|L'| \le j/b + 1$.
\end{claim}
\begin{proof}
Since $L'$ is the set of leaves of $R'$, different $u,v \in L'$ satisfy that $u$ is not an ancestor of $v$ and $v$ is not an ancestor of $u$, and thus the set of solutions to $I_u$ and $I_v$ are disjoint. Moreover, since any $u,v \in A(u^*)$ satisfy that $u$ is an ancestor of $v$ or $v$ is an ancestor of $u$, at most one $l \in L'$ is contained in $A(u^*)$ and thus corresponds to an unfinished recursive call. The remaining nodes $u \in L' \setminus \{l\}$ correspond to finished recursive calls, so all of their solutions have been listed at this moment. Since these solutions are disjoint and listed, we infer $\sum_{u \in L' \setminus \{l\}} t(I_u) < j$. Finally, each $u \in L'$ is not a leaf of the recursion tree, so \algobudgetlist{} reported that $I_u$ has more than $b$ solutions, i.e., $t(I_u) > b$. This implies $b \cdot |L' \setminus \{l\}| < \sum_{u \in L' \setminus \{l\}} t(I_u) < j$, and thus $|L'| = |L' \setminus \{l\}| + 1 \le j/b + 1$.
\end{proof}

\begin{claim}
$|R| \le 3d |L'| + 1$.
\end{claim}

\begin{proof}
We have $|R'| \le \sum_{u \in L'} |A(u)| \le d \, |L'|$, since each $v \in R'$ is an ancestor of some leaf $u \in L'$, and since $|A(u)|$ is at most the recursion depth $d$.
Moreover, we have $|L| \le 2 |R'| + 1$, since the parent of each $u \in L$ is in $R'$ unless $u$ is the root of the tree, and each $v \in R'$ has at most two children.
Together, we obtain $|R| = |R'| + |L| \le 3 |R'| + 1 \le 3d |L'| + 1$.
\end{proof}

Combining the above two claims yields that the number of recursive calls started before printing the $j$-th solution is at most $3d \cdot (j/b + 1) + 1$.

Now we bound the direct cost of each recursive call \algocheatenum{$I_u, u$}. The call of \algobudgetlist{} runs in time $O(f(I_u) + b \cdot g(I_u)) \le O(f(I) + b \cdot g(I)) \le O(f(I))$, by \algosplit{}-monotonicity of $f$ and $g$, our choice of $b = \ceil{f(I)/g(I)}$, and the assumption $f(I) \ge g(I)$. The time to construct the subinstances $J_0 = I_{u0}$ and $J_1 = I_{u1}$ is $O(f(I_u)) \le O(f(I))$ by \algosplit{}-monotonicity. Hence, the direct cost of each recursive call is $O(f(I))$.

Finally, since $(f,g)$ are constructible, the global parameter $b$ can be computed in time $O(f(I) + g(I)) = O(f(I))$. Putting all bounds together, the time until printing the $j$-th solution is
\[ O( f(I) + g(I) + f(I) \cdot d \cdot (j/b + 1) ) = O(f(I) \cdot d + j \cdot g(I) \cdot d ). \]
This finishes the analysis of \algocheatenum{}. Applying \Cref{lem:cheaterI} on \algocheatenum{} converts it to a proper enumeration algorithm and finishes the proof.
\end{proof}

\subsection{Proof of \Cref{thm:enumtolist}} \label{sec:proofenumtolist}

To prove \Cref{thm:enumtolist} in full generality, we first extend \Cref{lem:budgetlistsimple} to the following.

\begin{lemma}
Let $(f,g)$ be constructible, \algosplit{}-monotone functions with $f(I) \ge g(I)$ for all $I$. If there is a listing algorithm in time at most $f(I) + t \cdot g(I)$, then there is an algorithm \algoparbudgetlist{$I_0, I_1, b$} that, given $(I_0, I_1) = \protect\algosplit{I}$ and an integer $b \ge 1$, in time $O(f(I) + \min\{t(I_0), t(I_1), b\} \cdot g(I))$ lists all solutions to $I_0$ or lists all solutions to $I_1$ or reports that $I_0$ and $I_1$ have more than $b$ solutions.
\end{lemma}

\begin{proof}
We run the listing algorithm in parallel on $I_0$ and $I_1$ with a time budget of $2(f(I) + b \cdot g(I))$. That is, in each even time step we advance the computation on $I_0$, and in each odd time step we advance the computation on $I_1$. If any of the two parallel runs terminates within the time budget, then we print all its listed solutions as soon as it terminates. Otherwise we report that both $I_0$ and $I_1$ have more than $b$ solutions.

Note that this algorithm needs to know $f(I)$ and $g(I)$, and indeed they can be computed in time $O(f(I) + g(I)) = O(f(I))$ by constructibility and the assumption $f(I) \ge g(I)$.

Write $\tau := \min\set{t(I_0), t(I_1), b}$. We analyze the algorithm by a case distinction. In the case $t(I_0) = \tau$, the run on $I_0$ terminates within $f(I_0) + t(I_0) \cdot g(I_0) \leq f(I) + t(I_0) \cdot g(I) = f(I) + \tau \cdot g(I)$ steps since $f,g$ are \algosplit{}-monotone. This happens before we abort. Thus, the run that terminates first terminates within $f(I) + \tau \cdot g(I)$ steps, so we list all solutions to $I_0$ or $I_1$ within time $O(f(I) + \tau \cdot g(I))$.

Symmetrically, in the case $t(I_1) = \tau$, we list all solutions to $I_0$ or $I_1$ in time $O(f(I) + \tau \cdot g(I))$. 

Finally, in the case $\tau = b$, it could still happen that $I_0$ or $I_1$ terminates before we abort. Thus, we either list all solutions to $I_0$ or list all solutions to $I_1$ or abort. If we do abort, then both runs must have more than $b$ solutions, which we correctly report. Since we abort after $2(f(I) + b \cdot g(I))$ time steps, the running time is $O(f(I) + b \cdot g(I)) = O(f(I) + \tau \cdot g(I))$.
\end{proof}

We can then extend \Cref{lem:enumtolistsimple}.

\begin{lemma} \label{lem:enumtolistmain}
Let $\beta \in [0,1]$ be a constant. Let $(f,g)$ be constructible, \algosplit{}-monotone functions with $f(I) \ge g(I)$ for all $I$.
If there is a listing algorithm running in time at most $f(I) + t \cdot g(I)$, and \algosplit{} runs in time at most $f(I)$, then there is an enumeration algorithm with preprocessing time $O(f(I) \cdot d^{2 - \beta})$ and delay $O(g(I) \cdot d^\beta)$.
\end{lemma}

\begin{proof}
We present a ``cheating'' enumeration algorithm in the sense of \Cref{lem:cheaterI}.

\begin{pseudocode}
\Function{pCheatEnum}{$J, u$}
    \LComment{lists all solutions to instance $J$, prepended by $u$}
    \If{$d(J) = 0$}
        \State run the listing algorithm on $J$, prepend $u$ to each solution and print it
    \Else
        \State $(J_0, J_1) := \protect\algosplit{J}$
        \State run \algoparbudgetlist{$J_0, J_1, b$}
        \If{\algoparbudgetlist{} listed all solutions to $J_0$}
            \State prepend $u0$ to each one and print them
        \Else
            \State \algoparcheatenum{$J_0, u0$} 
        \EndIf
        \If{\algoparbudgetlist{} listed all solutions to $J_1$}
            \State prepend $u1$ to each one and print them
        \Else
            \State \algoparcheatenum{$J_1, u1$} 
        \EndIf
    \EndIf
\EndFunction
\end{pseudocode}

Given an instance $I$, we compute $f(I),g(I)$, set the parameter $b := \ceil{f(I)/g(I) \cdot d^{1-\beta}}$ globally, and then call \algoparcheatenum{$I, \epsilon$}. Observe that this lists all solutions to instance $I$. Indeed, the algorithm proceeds similarly as the trivial listing algorithm, except that a recursive call is skipped if \algoparbudgetlist{} succeeds in listing all of its solutions.

Fix any $1 \le j \le t$. Our goal is to bound the time until the algorithm prints its $j$-th solution. So consider the moment when the $j$-th solution is printed. This happens during a recursive call \algoparcheatenum{$I_{u^*}, u^*$} for some $u^* \in \set{0,1}^{\le d}$. Note that the set of ancestors $A(u^*)$ is the set of recursive calls that have started but not yet finished at this point in time.
Let $R \subseteq \set{0,1}^{\le d}$ be the set of all started recursive calls, i.e., $u \in \set{0,1}^{\le d}$ is in $R$ if \algoparcheatenum{$I_u, u$} has started before printing the $j$-th solution.
Let $L \subseteq R$ be the leaves of the recursion, i.e., $u \in R$ is in $L$ if neither $u0$ nor $u1$ is in~$R$. Let $B \subseteq R$ be the branching nodes, i.e., $u \in R$ is in $B$ if both $u0$ and $u1$ are in $R$.
So each node $u \in R \setminus (L \cup B)$ has exactly one child in $R$.

\begin{claim}
$|B| \le |L|$.
\end{claim}
\begin{proof}
This is essentially the fact that every full binary tree has at least as many leaves as inner nodes. For a quick proof, note that on the one hand the total number of $R$-children of nodes in $B$ is equal to $2|B|$, since each node in $B$ has two children in $R$. On the other hand, starting from an $R$-child of a node $u \in B$ and skipping nodes in $R \setminus (L \cup B)$ we arrive at a descendant $u'$ of $u$ such that $u' \in B$ or $u' \in L$, and the node $u'$ is unique. Therefore, the total number of $R$-children of nodes in $B$ is at most $|B| + |L|$. Hence, $2|B| \le |B| + |L|$, which shows $|B| \le |L|$.
\end{proof}

\begin{claim}
$|L| \le j/b + 1$.
\end{claim}
\begin{proof}
Let $u \in L$. If no ancestor of $u$ is in $B$, then each ancestor of $u$ is in $R \setminus B$ and thus has at most one child in $R$. In this case, the recursion tree is a path and has exactly one leaf, so we have $|L| \le 1$ and the claim holds.

So we can assume that each $u \in L$ has some ancestor in $B$. Let $a_u$ be the ancestor of $u$ of smallest depth $d(a_u)$ such that the parent of $a_u$ is in $B$ (recall that smallest depth means closest to the leaves). Then all nodes on the path from $a_u$ to $u$ are in $R \setminus B$, so they have at most one child in $R$. Also, $u$ has no child since $u \in L$. It follows that $a_u$ has a unique descendant in $L$, namely $u$. Thus, for different leaves $u,v \in L$ we have chosen ancestors $a_u, a_v$ such that $a_u$ is not an ancestor of $a_v$ and $a_v$ is not an ancestor of $a_u$ (indeed, if $a_u$ were an ancestor of $a_v$ then by transitivity it would be an ancestor of $v$, and thus it would be an ancestor of two leaves $u$ and $v$). 

Since for any two nodes $u,v \in A(u^*)$ it holds that $u$ is an ancestor of $v$ or $v$ is an ancestor of $u$, at most one leaf $l \in L$ has $a_l \in A(u^*)$. The remaining leaves $u \in L \setminus \{l\}$ have $a_u \notin A(u^*)$, so $a_u$ is a finished recursive call, so the $t(a_u)$ solutions to $a_u$ have already been listed when we print the $j$-th solution. It follows that $\sum_{u \in L \setminus \{l\}} t(a_u) < j$. Moreover, for each $u \in L$ we have $t(a_u) > b$, since $a_u$ is a child of a branching node, and a node can only be branching if \algoparbudgetlist{} reported that both children have more than $b$ solutions. Hence, $b \cdot |L \setminus \{l\}| \le \sum_{u \in L \setminus \{l\}} t(a_u) < j$, which implies $|L| = |L \setminus \set{l}| + 1 \le j/b + 1$.
\end{proof}

\begin{claim}
$|R| \le (d+1) \cdot (j/b + 1)$.
\end{claim}
\begin{proof}
We have $|R| \le \sum_{u \in L} |A(u)| \leq (d+1) |L| \le (d+1) \cdot (j/b + 1)$. Here, the first step used that each $v \in R$ is an ancestor of a leaf in $L$. The second step used that every leaf has at most $d+1$ ancestors (including itself). The last step used the previous claim.
\end{proof}

Now we bound the direct cost of each recursive call \algoparcheatenum{$I_u, u$}. In the base case $d(I_u) = 0$, running the listing algorithm on $I_u$ takes time $O(f(I_u) + g(I_u)) \le O(f(I) + g(I)) \le O(f(I))$ due to \algosplit{}-monotonicity of $f$ and $g$ and the assumption $f(I) \ge g(I)$.

In the case $d(I_u) \geq 1$, running \algoparbudgetlist{} takes time
\MATH{
    O \left(f(I_u) + \min\Set{t(I_{u0}), t(I_{u1}), b} \cdot g(I_u) \right) \NEWLINE
    \le O\left( f(I) + \min\set{t(I_{u0}), t(I_{u1}), b} \cdot g(I) \right)
}
by \algosplit{}-monotonicity. The time to construct the subinstances $J_0 = I_{u0}$ and $J_1 = I_{u1}$ is $O(f(I_u)) \le O(f(I))$ by \algosplit{}-monotonicity. Hence, the direct cost of the recursive call is $O(f(I) + \min\{t(I_{u0}), t(I_{u1}), b\} \cdot g(I))$.

Next we sum the direct cost over all $u \in R$.
\begin{itemize}
    \item For each $u \in L \setminus \set{u^*}$, \algoparcheatenum{$I_u,u$} was started but did not recurse any further. So the depth $d(I_u)$ must be zero. Hence, the direct cost is $O(f(I))$, and the total cost is $\tau_1 = O(|L| \cdot f(I))$.

    \item For each $u \in B \cup A(u^*)$, we can bound the direct cost by $O(f(I) + b g(I))$. Since $|B \cup A(u^*)| \leq |B| + |A(u^*)| \leq j/b+d+2$ by the claims above, the total cost is $\tau_2 = O\big((j/b+d) \cdot (f(I) + b g(I))\big)$.
    
    \item For each $u \in R \setminus (L \cup B \cup A(u^*))$, \algoparbudgetlist{} listed all solutions to exactly one of $I_{u0}$ and $I_{u1}$. Let $t_u \in \{ t(I_{u0}), t(I_{u1}) \}$ denote the number of listed solutions. Then the direct cost is $O(f(I) + t_u \, g(I))$. Since node $u$ is a finished recursive call (as $u \notin A(u^*)$), these $t_u$ solutions have already been listed when we print the $j$-th solution. These solutions are also disjoint: When \algoparbudgetlist{} lists all solutions to a subinstance, we do not recurse on that subinstance, and we cannot create the same solutions in a different way. It follows that $\sum_{u \in R \setminus (L \cup B \cup A(u^*))} t_u < j$. So the total cost is $\tau_3 = O(|R \setminus L| \cdot f(I) + j \cdot g(I))$.
\end{itemize}

Altogether, we can bound the total time until printing the $j$-th solution by
\begin{align*}
\tau_1 + \tau_2 + \tau_3
&= O\left( \left(\frac{j}{b}+d+|R|\right) f(I) + (j+bd) \, g(I) \right) \\
&\leq O\left( \left(\frac{dj}{b}+d \right) f(I) + (j+bd) \, g(I) \right) \\
&= O\left( d f(I) + bd \, g(I) + j \cdot \left(\frac{d}{b} f(I) + g(I)\right) \right),
\end{align*}
where we used the claim $|R| \leq (d+1) \cdot (j/b+1)$. Plugging in the parameter $b := \ceil{f(I)/g(I) \cdot d^{1-\beta}}$ into the bound, we conclude that the time until \algoparcheatenum{} prints the $j$-th solution is $O(f(I) d^{2 - \beta} + j \cdot g(I) d^\beta)$.

Finally, we apply \Cref{lem:cheaterI} to turn \algoparcheatenum{} into an enumeration algorithm with preprocessing time $O(f(I) d^{2-\beta})$ and delay $O(g(I) d^\beta)$.
\end{proof}

\Cref{thm:enumtolist} now follows easily from \Cref{lem:enumtolistmain}.

\begin{proof}[Proof of \Cref{thm:enumtolist}]
Let $(f,g,s)$ be constructible, \algosplit{}-monotone functions. Define $f'(I) := c \cdot (f(I) + g(I) + s(I))$ and $g'(I) := c \cdot g(I)$ for a constant $c \geq 1$. Note that $(f',g')$ are constructible, and $f'(I) \ge \max\set{g'(I),s(I)}$ for all~$I$. If there is a listing algorithm in time $O(f(I) + t \cdot g(I))$ and \algosplit{} runs in time $O(s(I))$, then we can choose a sufficiently large $c$ so that the times are bounded by $f'(I) + t \cdot g'(I)$ and $f'(I)$, respectively. Hence, we can apply \Cref{lem:enumtolistmain} on $f',g'$ and obtain an enumeration algorithm with preprocessing time $O(f'(I) \cdot d^{2 - \beta}) = O((f(I) + g(I) + s(I)) \cdot d^{2 - \beta})$ and delay $O(g'(I) \cdot d^\beta) = O(g(I) \cdot d^\beta)$.
\end{proof}

\subsection{Open Directions}
We leave the following open questions on the enumeration-to-listing reduction.

\paragraph{Improvements}
Can the overhead factors be further improved? Specifically, is it possible to improve the overheads in preprocessing time and delay from the current tradeoff $(O(d^{2-\beta}), O(d^\beta))$ to, e.g., $(O(d^{1.99}), O(1))$ or $(O(d), O(d^{0.99}))$ or even $(O(d), O(1))$? 
    
It is natural to expect that the overhead in preprocessing time must be at least $d$, at least in the general setting considered here, as this is the overhead of the classic search-to-decision reduction. Can one prove this?

\paragraph{Non-constructible time bounds}
Is there a similar reduction without the assumption that the time bounds are constructible? Note that, e.g., treewidth is a subgraph-monotone parameter, but it cannot be computed efficiently, so time bounds involving treewidth are not constructible, which shows a limitation of the current reduction.

\paragraph{Uncolored problems}
While we demonstrated our reduction for Colored $(H,I)$-Listing, it is not directly applicable to Uncolored $(H,I)$-Listing, since the uncolored problem does not have a self-reduction and is not splittable in the sense of \Cref{def:splitting}. Could there be a variant of \Cref{thm:enumtolist} that would work for uncolored problems?

\paragraph{More applications}
Does our reduction have applications beyond subgraph listing and conjunctive queries, such as enumerating spanning trees, matchings, etc.? Or is there some variant of the reduction that would admit more applications?

\paragraph{Randomized algorithms}
\Cref{thm:enumtolist} assumes that the given listing algorithm is deterministic, and yields a deterministic enumeration algorithm. In some settings of randomized algorithms it might be possible to avoid the factor $d$ overhead altogether. 

\section{From Colored to Uncolored Problems}
\label{sec:uncoloredalgo}

In this section, we show how to convert a listing or enumeration algorithm for Colored $(H,I)$-Subgraph into one for Uncolored $(H,I)$-Subgraph, with only a factor $\polylog(n)$ overhead in preprocessing time and delay.

We start with a lemma which phrases the deterministic color coding technique~\cite{AYZ95} as a reduction. The $O$-notation in the lemma hides a $k^k$ dependence.

\begin{lemma}
    \label{lem:color-coding}
    Let $H = ([k], \mathcal{E})$ be a pattern hypergraph. Given a host hypergraph $G = (V,E)$, in time $O(n \log n)$ we can compute $r = O(\log n)$ partitions $\mathcal{P}_1,\ldots,\mathcal{P}_r$ of $V$ with the following property. Let $\sol$ be the set of uncolored $H$-copies in $G$. Let $\sol_l$ be the set of colored $H$-copies in $G$ with respect to partition $\mathcal{P}_l$. We have $\sol = \bigcup_{l \in [r]} \sol_l$.
\end{lemma}

\begin{proof}
    Following \cite[Section 4]{AYZ95}, in time $O(n \log n)$ we can compute a family $\mathcal{F}$ of functions $V \to [k]$ with the following properties:
    \begin{enumerate}[(1)]
        \item $|\mathcal{F}| = O(\log n)$, and every function $f \in \mathcal{F}$ can be evaluated at a given point in constant time.
        \item For every tuple $(v_1, \dots, v_k) \in V^k$ of distinct vertices, there exists $f \in \mathcal{F}$ such that $f(v_1), \dots, f(v_k)$ are distinct, or equivalently, there exist $f \in \mathcal{F}$ and a permutation $\pi: [k] \to [k]$ with $\pi(f(v_i)) = i$ for all $i \in [k]$.
    \end{enumerate}
    For each function $f \in \mathcal{F}$ and each permutation $\pi : [k] \to [k]$, we define a partition $\mathcal{P}_{f,\pi}$ of $V$ by putting vertex $v$ in the $\pi(f(v))$-th part. Computing all these partitions takes time $O(n \log n)$ by property (1).
    
    Consider any uncolored $H$-copy $(v_1,\dots,v_k) \in \sol$. By property (2), there exists $f,\pi$ such that $\pi(f(v_i)) = i$ for all $i \in [k]$. In particular, the vertex $v_i$ is put into the $i$-th part of $\mathcal{P}_{f,\pi}$. Hence $(v_1, \dots, v_k)$ is a colored $H$-copy in $G$ with respect to the partition $\mathcal{P}_{f,\pi}$. Therefore, $\sol \subseteq \bigcup_{f,\pi} \sol_{f,\pi}$. For the converse direction~$\supseteq$, note that every colored $H$-copy in $G$ with respect to any partition is also an uncolored $H$-copy in~$G$.
\end{proof}

A direct application of \Cref{lem:color-coding} transforms colored listing algorithms to uncolored listing algorithms. 

\begin{lemma} \label{lem:listing_colored_to_uncolored}
    If Colored $(H,I)$-Listing can be solved in time $f(n,m,t)$ then Uncolored $(H,I)$-Listing can be solved in time $\tOh(f(n,m,t))$.
\end{lemma}
\begin{proof}
    Use \Cref{lem:color-coding} to construct partitions $\mathcal{P}_1,\ldots,\mathcal{P}_r$. For $l \in [r]$, run the colored algorithm on $G$ with respect to partition $\mathcal{P}_l$ to compute the set of solutions $\sol_l$. Finally, output $\sol = \bigcup_{l \in [r]} \sol_l$ after sorting and deduplicating.
\end{proof}

In order to obtain uncolored \emph{enumeration} algorithms, we additionally need the help of the Cheater's Lemma~\cite{CK21}. Compared to the variant we used in \Cref{sec:enumtolist} (\Cref{lem:cheaterI}), we now need a variant where the initial algorithm may print duplicate solutions.

\begin{lemma}[Cheater's Lemma, Variant II]
    \label{lem:cheaterII}
    Let $f,g,h$ be functions such that $f(n,m), g(n,m), h(n,m)$ can be computed in time $O(f(n,m) + g(n,m) h(n,m))$. Suppose there is an algorithm that, given an instance $G$ of Colored $(H,I)$-Subgraph, (1)~prints every solution at least once and at most $h(n,m)$ times, and (2)~prints at least $j$ solutions (counting duplicates) in time $f(n,m) + j \cdot g(n,m)$ for all $j \in [s]$, where $s$ is the total number of printed solutions. Then there exists an enumeration algorithm with preprocessing time $\tOh(f(n,m) + g(n,m)h(n,m))$ and delay $\tOh(g(n,m)h(n,m))$.
\end{lemma}

\begin{proof}
    Suppose an algorithm satisfies (1) and (2). We first run the algorithm for $f(n,m)$ steps and push the printed solutions to a queue $Q$. Duplicates are discarded during the process: an incoming solution shall not enter $Q$ if it appeared before. Then we run the following loop:
    \begin{itemize}
        \item Run the algorithm for another $g(n,m)h(n,m)$ steps, again pushing the printed solutions to $Q$.
        \item If $Q$ is empty then break the loop. Else, remove the head of $Q$, output it and continue.
    \end{itemize}

    Consider the $i$-th iteration of the loop, $i \in [t]$. When we reach the if-else line, we have already run the algorithm for $f(n,m) + i \cdot g(n,m)h(n,m)$ steps. Property (2) implies that at least $j = \min\set{s, i h(n,m)}$ solutions have entered $Q$ (counting duplicates). If $j=s$, then the algorithm has nothing more to print, so all $t$ distinct solutions have entered $Q$ by property (1). If $j=ih(n,m)$, then at least $i$ distinct solutions have entered $Q$ by property (1). In any case, $Q$ is non-empty because we have output only $i-1$ distinct solutions (one from each previous iteration). In other words, the $i$-th iteration outputs a new solution.
    
    At the $(t+1)$-st iteration, we have already output all $t$ distinct solutions, so $Q$ must be empty and we terminate correctly.

    Checking whether the incoming solution appeared before requires a data structure such as a binary search tree, which takes time $O(\log t) \le O(\log n^{|V(H)|}) = O(\log n)$ per operation. Therefore, we obtain an enumeration algorithm with preprocessing time $\tOh(f(n,m) + g(n,m)h(n,m))$ and delay $\tOh(g(n,m)h(n,m))$.
\end{proof}

By combining \Cref{lem:color-coding} and \Cref{lem:cheaterII} we can transfer enumeration algorithms from colored problems to uncolored problems.

\begin{lemma}
    \label{lem:enum_colored_to_uncolored}
    Let $f,g$ be functions such that $f(n,m),g(n,m)$ can be computed in time $O(f(n,m) + g(n,m))$. 
    If Colored $(H,I)$-Enumeration can be solved in preprocessing time $f(n,m)$ and delay $g(n,m)$, then Uncolored $(H,I)$-Enumeration can be solved in preprocessing time $\tOh(f(n,m) + g(n,m))$ and delay $\tOh(g(n,m))$.
\end{lemma}

\begin{proof}
    We construct a ``cheating'' algorithm in the sense of \Cref{lem:cheaterII}. First, compute partitions $\mathcal{P}_1,\ldots,\mathcal{P}_r$ from \Cref{lem:color-coding}. Second, for $l = 1, \dots, r$ sequentially, call the colored enumeration algorithm on $G$ with respect to the partition $\mathcal{P}_l$.
    
    Let $\sol$ denote the set of uncolored solutions in $G$, and $\sol_l$ the set of colored solutions in $G$ with respect to $\mathcal{P}_l$. \Cref{lem:color-coding} shows that $\sol = \bigcup_{l \in [r]} \sol_l$, so the algorithm above prints every solution at least once and at most $r$ times. Moreover, writing $s$ for the number of printed solutions (counting duplicates), we claim that the algorithm prints at least $j$ solutions (counting duplicates) in time $r f(n,m) + (j + r) \cdot g(n,m)$ for all $j \in [s]$. Indeed, each of the $r$ calls to the enumeration algorithm takes time $f(n,m)$ for preprocessing, time $g(n,m)$ to terminate after printing the last solution, and otherwise time $g(n,m)$ per printed solution.

    Applying \Cref{lem:cheaterII} with $h(n,m) := r$ converts this algorithm into an algorithm for Uncolored $(H,I)$-Enumeration with preprocessing time $O(f(n,m) \cdot r + g(n,m) \cdot r) = \tOh(f(n,m) + g(n,m))$ and delay $\tOh(g(n,m) \cdot r) = \tOh(g(n,m))$.
    
    We remark that the $\tOh$-notation hides a factor $k^k$, where $k = |V(H)|$.
\end{proof}

We are finally ready to prove \Cref{thm:main1} and \Cref{thm:algohypergraphs} in full generality.

\MainOneTheorem*
\begin{proof}
    Colored $(H,I)$-Listing is handled in \Cref{thm:main1-colored-listing}.
    For Uncolored $(H,I)$-Listing, combine \Cref{thm:main1-colored-listing} with \Cref{lem:listing_colored_to_uncolored}.
    Colored $(H,I)$-Enumeration is handled in \Cref{thm:main1-colored-enumeration}.
    For Uncolored $(H,I)$-Enumeration, combine \Cref{thm:main1-colored-enumeration} with \Cref{lem:enum_colored_to_uncolored}.

    In these theorems, we have expressed the exponent as $3/\alfa$. Using the bound $3/\alfa \le 17.42$ from \Cref{sec:mm} yields the stated time. If $n \times n$ matrix multiplication were in time $\tOh(n^2)$, then we could take $\alfa = 1$, so the exponent would improve to $3$.
\end{proof}

\AlgoHypergraphsTheorem*
\begin{proof}
    Colored $(H,I)$-Listing is handled in \Cref{thm:algohypergraphs-colored-listing}.
    For Uncolored $(H,I)$-Listing, combine \Cref{thm:algohypergraphs-colored-listing} with \Cref{lem:listing_colored_to_uncolored}.
    Colored $(H,I)$-Enumeration is handled in \Cref{thm:algohypergraph-colored-enumeration}.
    For Uncolored $(H,I)$-Enumeration, combine \Cref{thm:algohypergraph-colored-enumeration} with \Cref{lem:enum_colored_to_uncolored}.
    Using the bound $3/\alfa \le 17.42$ from \Cref{sec:mm} yields the stated time.
\end{proof}

\section{Faster Algorithm for Projected Stars}
\label{sec:star-m2}
This section aims to prove \Cref{thm:star-listing}. We do so by solving a more general counting problem:
\begin{theorem}
    \label{thm:star-counting}
    There is a deterministic algorithm that, given functions $\mu_i: [N_0] \times [N_i] \to \N$ for $i \in [d]$, outputs all tuples $\tup{v} = (v_1, \dots, v_d) \in \prod_{i \in [d]} [N_i]$ with non-zero value $\mu(\tup{v})$, where
    \[ \mu(\tup{v}) := \sum_{v_0 \in [N_0]} \prod_{i \in [d]} \mu_i(v_0, v_i). \]
    These tuples are called \EMPH{solutions}. In addition, the algorithm computes $\mu(\tup{v})$ for each solution $\tup{v}$. The algorithm runs in time $\tOh(m^2 t^{1-\alfa} + t) \leq \tOh(m^{2/\alfa} + t)$, where $t$ is the number of solutions and $m$ is the number of non-zeros in the input functions.
\end{theorem}

\begin{remark}
    As usual, we assume that the functions are sparsely represented: each $\mu_i$ is given as a list of key-value pairs, and only non-zero values appear in the list.
    If $\alfa = 1$ then our algorithm runs in time $\tOh(m^2 + t)$, while Hu's algorithm~\cite{Hu25} runs in time $\tOh(m t^{2/3 - \Theta(1/d)} + t)$. Our running time is optimal for large $d$ under fine-grained complexity assumptions; see Section~\ref{sec:lower-bounds}.
\end{remark}

The counting problem specializes to Colored projected-$d$-star-Listing, by encoding the host graph in the functions $\mu_i$. This allows us to derive \Cref{thm:star-listing} from \Cref{thm:star-counting}:

\begin{proof}[Proof of \Cref{thm:star-listing}]
    In Colored projected-$d$-star-Listing, the pattern graph is $H$ with $V(H) = \set{0,1,\dots,d}$, $E(H) = \Set{\set{0,i} : i \in [d]}$ and $I = [d]$. We are given a host graph $G = \left(\bigcup_{0 \leq i \leq d} V_i, E\right)$ as input. Without loss of generality we assume $V_i = [N_i]$ for $0 \leq i \leq d$. For each $i \in [d]$, we define a function $\mu_i : [N_0] \times [N_i] \to \set{0,1}$ that encodes the adjacency between $V_0$ and $V_i$: assign $\mu_i(v_0, v_i) := 1$ if $\set{v_0,v_i} \in E$, and 0 otherwise. We solve the counting problem on $\mu_1, \dots, \mu_d$ and output the set of solutions, i.e., the tuples $\tup{v}$ with non-zero value $\mu(\tup{v})$.
    
    Observe that $\mu(\tup{v})$ counts the number of $H$-copies in $G$ that extend $\tup{v}$. In particular, $\tup{v}$ is a solution in the counting problem if and only if it is a solution to $(H,I)$-Listing on $G$, and thus the algorithm is correct. Regarding time complexity, note that the parameter $m$ (number of edges versus number of non-zeros in the functions) agrees in both problems, so does the parameter $t$ (number of solutions). Hence the counting problem can be solved in time $\tOh(t + m^{2/\alfa})$ by \Cref{thm:star-counting}. The construction of the functions takes time $O(m)$, and printing the solutions takes time $O(t)$. Overall the running time is $\tOh(t + m^{2/\alfa})$.

    To handle Uncolored projected-$d$-star-Listing, we apply \Cref{lem:listing_colored_to_uncolored}. To handle Colored projected-$d$-star-Enumeration, we apply \Cref{thm:main2-simplified}. To handle Uncolored projected-$d$-star-Enumeration, we further apply \Cref{lem:enum_colored_to_uncolored}.
\end{proof}

At a high level, the proof works as follows. We reduce the instance to $\tOh(1)$ ``regular'' subinstances, where every $v_0$ in the middle part has roughly the same number of incident non-zeros. Then we design an algorithm for  ``dense'' regular instances via matrix multiplication, and further reduce general subinstances to dense ones via hashing and iterative recovery. The pieces are discussed in the next three subsections, and the full algorithm is assembled in the subsection that follows.

\subsection{Degree Uniformization}
\begin{definition}
    Let $\mu_1, \dots, \mu_d$ be an instance where $\mu_i : [N_0] \times [N_i] \to \N$ for $i \in [d]$. The \EMPH{degree} of $v_0 \in [N_0]$ is defined as $\sum_{i \in [d]} |\set{v_i \in [N_i] : \mu_i(v_0,v_i) > 0}|$. The instance is \EMPH{$\Delta$-near-regular} if the degree of every $v_0 \in [N_0]$ is in the range $(\frac{\Delta}{2},\Delta]$.
\end{definition}

\begin{lemma}
    \label{lem:degree-uniform}
    If for every $\Delta \in \N$ the counting problem on $\Delta$-near-regular instances can be solved in time $f(m,t)$, then the counting problem in general can be solved in time $O(f(m,t) \log m)$.
\end{lemma}

\begin{proof}
    Let $\mu_1, \dots, \mu_d$ be a given general instance. We compute the degree of every $v_0 \in [N_0]$ in time $O(m)$. Note that degrees are integers in $[m]$. Then we partition $[N_0] = V_0^0 \cup \cdots \cup V_0^{\ceil{\log m}}$, where $V_0^l$ consists of all $v_0$ with degree in $(2^{l-1}, 2^l]$. For each $l$, we define a subinstance $\mu_1^l, \dots, \mu_d^l$ where $\mu_i^l$ is the function $\mu_i$ restricted to $V_0^l$. Note that $\mu(\tup{v}) = \sum_{l} \mu^l(\tup{v})$ for all tuples $\tup{v} \in \prod_{i \in [d]} [N_i]$. Hence, the set of solutions to the original instance is the union of the set of solutions to the subinstances. Since each subinstance $l$ is $2^l$-near-regular by definition, it can be solved in time $f(m,t)$. The overall running time for solving the original instance is thus $O(m + f(m,t) \log m + t) = O(f(m,t) \log m)$.
\end{proof}

\subsection{Counting in Dense Instances}
We consider ``dense'' instances where $N_1, \dots, N_d$ are relatively small in the following sense:
\begin{definition}
    We say that $N_1, \dots, N_d$ are \EMPH{$(N_0, D)$-bounded} if $\prod_{i \in I} N_i \leq N_0 D^{|I|}$ for all $I \subseteq [d]$.
\end{definition}

A useful combinatorial property is that every set of $(N_0,D)$-bounded numbers admits a balanced partition. Write $N_I = \prod_{i \in I} N_i$.
\begin{lemma}
    \label{lem:dense}
    Let $c, N_0, D \geq 1$. If $N_1, \dots, N_d$ are $(N_0,D)$-bounded, then there exists a partition $[d] = I \cup \overline I$ such that
    \begin{equation*}
        \max\Set{ N_I, N_{\overline I} } \;\leq\; \frac{N_{[d]}}{N_0^c} + N_0^c D^{c+1} \;\leq\; \frac{N_{[d]}}{N_0^c} + N_0^c D^{2c}.
    \end{equation*}
\end{lemma}

\begin{proof}
    The second inequality follows just from the assumption $c \geq 1$, so we focus on showing the first inequality. If $d < 2c$ then the partition is trivial. Indeed, split $[d]$ arbitrarily into two parts~$I, \overline I$ of size at most $\ceil{\frac{d}{2}} \leq \ceil{c}$. Under the premise of the lemma, this partition satisfies
    \begin{equation*}
        \max\Set{N_I, N_{\overline I}} \leq N_0 D^{\ceil{c}} \leq N_0^c D^{c+1}.
    \end{equation*}
    
    So suppose that $d \geq 2c$. In this case we construct the partition by induction. Reorder the numbers so that $N_d = \min\set{N_1, \dots, N_d}$, and inductively obtain a partition $[d-1] = J \cup \overline J$ such that
    \begin{equation*}
        \max\Set{ N_J, N_{\overline J} } \leq \frac{N_{[d-1]}}{N_0^c} + N_0^c D^{c+1}.
    \end{equation*}
    Without loss of generality assume that $N_J \leq N_{\overline J}$. Then let $I = J \cup \set{d}$. Note that $\overline{I} = \overline{J}$ and thus $N_{\overline I} = N_{\overline J} \leq \frac{N_{[d]}}{N_0^c} + N_0^c D^{c+1}$ as required. It remains to analyze $N_I$. We have
    \[
        N_I = N_J \cdot N_d \leq \sqrt{N_J N_{\overline J}} \cdot N_d = \sqrt{N_{[d-1]}} \cdot N_d,
    \]
    so in the following we will bound $\sqrt{N_{[d-1]}} \cdot N_d$. We distinguish two cases:
    \begin{itemize}
        \item If $N_{[d-1]} > N_0^{2c}$, then
        \begin{equation*}
            \sqrt{N_{[d-1]}} \cdot N_d = \frac{N_{[d]}}{\sqrt{N_{[d-1]}}} < \frac{N_{[d]}}{N_0^c}.
        \end{equation*}
        \item If $N_{[d-1]} \leq N_0^{2c}$, then
        \[
            N_d \leq (N_{[d]})^{\frac{1}{d}} \leq \left(N_0 D^d\right)^{\frac{1}{d}} = N_0^{\frac{1}{d}} D
        \]
        by the assumption $N_d = \min\set{N_1, \dots, N_d}$ and the premise of the lemma. Besides this bound, we have $N_{[d-1]} \leq N_0 D^d$ (again by the premise of the lemma) and $N_{[d-1]} \leq N_0^{2c}$ (by the case assumption). Let $\gamma = \frac{2c}{d}$, and note that $0 < \gamma \leq 1$. Taking a geometric mean of both bounds with weights $\gamma$ and $1 - \gamma$,
        \begin{align*}
            N_{[d-1]} \ALIGN\leq (N_0 D^d)^\gamma \cdot N_0^{2c(1-\gamma)} \NEWLINE
            \ALIGN= N_0^{2c+\gamma(1-2c)} D^{2c} \leq N_0^{2c-\frac{2}{d}} D^{2c},
        \end{align*}
        where the last step used $c \geq 1$. Therefore,
        \[
            \sqrt{N_{[d-1]}} \cdot N_d \leq N_0^{c-\frac{1}{d}} D^{c} \cdot N_0^{\frac{1}{d}} D = N_0^c D^{c+1}.
        \]
    \end{itemize}
    This completes the proof.
\end{proof}

\begin{lemma}
\label{lem:dense-star-counting}
There is a deterministic algorithm \algodcount{$D, \allowbreak \mu_1, \dots, \mu_d$} that, given $\mu_1, \dots, \mu_d$ where the domain sizes $N_1, \dots, N_d$ are $(N_0,D)$-bounded, solves the counting problem in time $\tOh\big(\prod_{i \in [d]} N_i + (N_0 D)^2 (\prod_{i \in [d]} N_i)^{1-\alfa} \big)$.
\end{lemma}

\begin{proof}
    Let $[d] = I \cup \overline I$ be a partition from \Cref{lem:dense}; we compute it in constant time by enumerating and testing all partitions of $[d]$. Construct a matrix $A \in \N^{(\prod_{i \in I} [N_i]) \times [N_0]}$ where $A_{\tup{v}, v_0} = \prod_{i \in I} \mu_i(v_0, v_i)$. Similarly, construct a matrix $B \in \N^{[N_0] \times (\prod_{i \in \overline{I}} [N_i])}$ where $B_{v_0, \tup{w}} = \prod_{i \in \overline I} \mu_i(v_0, w_i)$. Compute their product $C = A B$ via \Cref{prp:mm}. By definition, each entry $C_{\tup{v},\tup{w}}$ in the product is exactly $\mu((\tup{v},\tup{w}))$.
    
    The running time is dominated by the matrix multiplication, which by \Cref{prp:mm} takes time
    \begin{align*}
        \mm(N_I, N_0, N_{\overline I})
        \ALIGN\;=\; \tOh\left( N_I N_{\overline I} + N_0 \left(N_I N_{\overline I}\right)^{1-\alfa} \left(N_I^\alfa + N_{\overline I}^\alfa \right) \right) \NEWLINE
        \ALIGN\;=\; \tOh\left( N_{[d]} + N_0 N_{[d]}^{1-\alfa} \max\Set{N_I^\alfa, N_{\overline I}^\alfa} \right).
    \end{align*}
    Applying \Cref{lem:dense} with $c = 1/\alfa \geq 1$, we have $\max \Set{N_I^\alfa, N_{\overline I}^\alfa} \leq N_{[d]}^\alfa / N_0 + N_0 D^2$. So we can continue bounding the time by
    \begin{align*}
        \tOh\left( N_{[d]} + N_0 N_{[d]}^{1-\alfa} \left(N_{[d]}^\alfa / N_0 + N_0 D^2 \right) \right) \NEWLINE
        = \tOh\left( N_{[d]} + (N_0 D)^2 N_{[d]}^{1-\alfa} \right),
    \end{align*}
    as claimed by the lemma.
\end{proof}

\subsection{Densification via Contractions}
To densify an instance, the key definition is that of a contraction, which compresses the universe $[N_1] \times \cdots \times [N_d]$ independently along each dimension into a smaller universe $[n_1] \times \cdots \times [n_d]$.

\begin{definition}
    A \EMPH{contraction} is a tuple of functions $h = (h_1, \dots, h_d)$ where $h_i : [N_i] \to [n_i]$. The \EMPH{volume} of the contraction is $\prod_{i \in [d]} n_i$. The \EMPH{image} of a point $\tup{v} = (v_1, \dots, v_d) \in \prod_{i \in [d]} [N_i]$ under $h$ is defined as
    \[ h(\tup{v}) = (h_1(v_1), \dots, h_d(v_d)) \;\in\; \prod_{i \in [d]} [n_i]. \]
    Two points are said to \EMPH{collide} under $h$ if they have the same image under $h$.
\end{definition}

\begin{definition}
    Let $P \subseteq \prod_{i \in [d]} [N_i]$ and let $h$ be a contraction. A point $q \in \prod_{i \in [d]} [N_i]$ is \EMPH{isolated} from $P$ under $h$ if it does not collide with any $p \in P \setminus \set{q}$ under $h$. For a family of contractions~$\mathcal H$, we say that a point is \EMPH{isolated} from $P$ under $\mathcal H$ if it is isolated from $P$ under some $h \in \mathcal H$.
\end{definition}

\begin{theorem}
\label{thm:contract}
    There is a deterministic algorithm \algocontract{$P$} that, given a set $P \subseteq \prod_{i \in [d]} [N_i]$, computes a family $\mathcal H$ of $\tOh(1)$ contractions, each of volume $\tOh(|P|)$, such that all but $|P|/4$ points in $P$ are isolated from $P$ under $\mathcal H$. The algorithm runs in time $\tOh\left(|P| + \sum_{i \in [d]} N_i\right)$.
\end{theorem}

Let us emphasize again that $d$ is considered a constant. The $\tOh$ notation hides factors $(\log |P|)^{d 2^d}$.

Our algorithm \algocontract{} employs incremental \emph{orthogonal range searching}. This is a deterministic data structure that maintains a point set $P \subset \R^{d'}$ under insertions. Upon an \emph{orthogonal range query} of the form $B = [a_1, b_1] \times \cdots \times [a_{d'}, b_{d'}]$, the data structure reports the count $|P \cap B|$. It is known that such a data structure with update and query time $\tOh(1)$ exists.

Towards a proof of \Cref{thm:contract}, we start with a lemma that generalizes~\cite[Lemma~3.7]{ABFK24}. It deterministically constructs a family of $O(\log |P|)$ contractions that isolates every point in a given point set. Note, however, that the volume bound in this lemma depends on the structure of the point set. So to close the proof of \Cref{thm:contract} we still need an argument about the point set structure, which will come later.

\begin{lemma}
    \label{lem:isolation}
    Let $Q \subseteq P \subseteq \prod_{i \in [d]} [N_i]$. Let $n_1, \dots, n_d$ be integers such that $\card{\set{p \in P : p_I = q_I}} \leq \prod_{i \notin I} n_i$ for all $q \in Q$ and $I \subseteq [d]$. There is a deterministic algorithm that computes a family $\mathcal H$ of $O(\log |Q|)$ contractions, each with volume \smash{$O(\prod_{i \in [d]} n_i)$}, such that every $q \in Q$ is isolated from $P$ under $\mathcal H$. The algorithm runs in time $\tOh\left(|P| + \sum_{i \in [d]} N_i\right)$.
\end{lemma}

\begin{proof}
For each $i \in [d]$, let $M_i$ be the smallest power of two such that $M_i \geq 5dn_i$. We start with a simple randomized algorithm and later derandomize it.
\begin{pseudocode}
\State $\mathcal{H} := \emptyset$
\While{$Q \neq \emptyset$}
    \State sample $h = (h_1, \dots, h_d)$ where each $h_i : [N_i] \to [M_i]$ is drawn uniformly at random
    \State add $h$ to $\mathcal{H}$
    \State remove all points in $Q$ that are isolated from $P$ under $h$
\EndWhile
\State \Return $\mathcal{H}$
\end{pseudocode}

Consider one round of the loop. Let $Z$ count the number of pairs $q \in Q$, $p \in P \setminus \set{q}$ that collide under~$h$. We claim that $\Pr(Z \geq |Q|/2) < 1/2$. To see this, write $\eq(p,q) := \set{i \in [d] : p_i = q_i}$ and note that two points $q \in Q$ and $p \in P \setminus \set{q}$ collide with probability $\prod_{i \notin \eq(p,q)} \Pr(h_i(p_i) = h_i(q_i)) = \prod_{i \notin \eq(p,q)} \frac{1}{M_i}$. Hence the expected number of points in $P \setminus \set{q}$ colliding with $q$ is at most
\begin{align*}
    \sum_{I \subsetneq [d]} \frac{|\set{p \in P : p_I = q_I}|}{\prod_{i \notin I} M_i}
    \ALIGN\leq \sum_{I \subsetneq [d]} \frac{1}{(5d)^{d-|I|}} \NEWLINE
    \ALIGN= \left(1+\frac{1}{5d} \right)^d - 1
    < \frac{1}{4}.
\end{align*}
Here, we used the premise of the lemma and $M_i \geq 5dn_i$. This implies $\E Z < |Q|/4$, so $\Pr(Z \geq |Q|/2) < 1/2$ by Markov's inequality. In particular, with probability $1/2$ at least half of the points in $Q$ are isolated from $P$ under $h$ (and thus are removed from $Q$ at the end of this round). Hence with high probability, the algorithm terminates in $O(\log |Q|)$ rounds and returns a family $\mathcal H$ of $O(\log |Q|)$ contractions, each of volume $\prod_{i \in [d]} M_i = O(\prod_{i \in [d]} n_i)$, such that every $q \in Q$ is isolated from $P$.

Next we derandomize the algorithm. The only randomness comes from sampling $h$. It is convenient to refine this sampling process as follows:
\begin{pseudocode}
\For{$i = 1, \dots, d$}
    \For{$x = 1, \dots, N_i$}
        \State initialize an interval $Y := [M_i]$
        \For{$l = 1, \dots, \log M_i$}
            \State partition the interval $Y$ into two halves $Y_0 \cup Y_1$
            \State sample a random bit $b$ and let $Y := Y_b$
        \EndFor
        \State assign $h_i(x)$ to be the unique element in $Y$
    \EndFor
\EndFor
\end{pseudocode}
\noindent Observe that the refined process generates uniformly random functions $h_i : [N_i] \to [M_i]$ indeed.

To derandomize the process, we follow the method of conditional expectations. Consider the time when iteration $(i,x,l)$ starts. Note that $h_1, \dots, h_{i-1}$ were determined, $h_{i+1}, \dots, h_d$ are fully random, and $h_i$ is just partially random: $h_i(1), \dots, h_i(x-1)$ were fixed, $h_i(x+1), \dots, h_i(N_i)$ are fully random, and $h_i(x)$ is uniformly distributed over the current interval $Y$. Let random variable $Z^{i,x,l}$ count the number of pairs $q \in Q$, $p \in P \setminus \set{q}$ that collide under the (partially random) contraction $h = (h_1, \dots, h_d)$ at this time. By the law of total expectation,
\[ \E\left(Z^{i,x,l} \right) = \frac{1}{2} \E\left(Z^{i,x,l+1} \mid b=0 \right) + \frac{1}{2} \E\left(Z^{i,x,l+1} \mid b=1 \right). \]
Hence, there exists a choice of $b \in \set{0,1}$ that makes the conditional expectation of $Z^{i,x,l+1}$ at most $\E(Z^{i,x,l})$. With this insight, we can derandomize the process by letting $b := 0$ if $\E\left(Z^{i,x,l+1} \mid b=0 \right) \leq \E\left(Z^{i,x,l+1} \mid b=1 \right)$, and $b := 1$ otherwise. This ensures that the expected number of colliding pairs never increases over time, so in the end we deterministically obtain a contraction $h$ under which at most $|Q|/4$ pairs collide. In particular, at most $|Q|/4$ points are not isolated. Replacing the sampling process in the original algorithm with this derandomized process, the number of rounds is bounded by $O(\log |Q|)$ deterministically.

It remains to show how to compare $\E\left(Z^{i,x,l+1} \mid b = 0 \right)$ with $\E\left(Z^{i,x,l+1} \mid b = 1 \right)$ efficiently. Let us first analyze the probability that a point $q \in Q$ collides with another point $p \in P \setminus \set{q}$ under $h$.
\begin{itemize}
    \item If $h_j(p_j) \neq h_j(q_j)$ for some $j < i$, then the probability is clearly zero.
    \item Otherwise, the probability can be expressed as $\gamma(p,q) \cdot \theta(p,q)$, where
    \begin{align*}
        \gamma(p,q) \ALIGN= \prod_{\substack{j>i \\ j \notin \eq(p,q)}} \frac{1}{M_j} \AND
        \theta(p,q) \ALIGN= \Pr\left[h_i(p_i) = h_i(q_i)\right].
    \end{align*}
    We further analyze $\theta(p,q)$ by a case distinction. If $p_i = q_i$ then $\theta(p,q)=1$. Otherwise:
    \begin{itemize}
        \item If $p_i, q_i < x$ then both $h_i(p_i)$ and $h_i(q_i)$ have been fixed, so $\theta(p,q) = \indicate[h_i(p_i) = h_i(q_i)]$.
        \item If $p_i > x$ or $q_i > x$, then $h_i(p_i)$ or $h_i(q_i)$ is uniform over $[M_i]$, so $\theta(p,q) = 1/M_i$.
        \item If $p_i < q_i = x$ then $h_i(q_i)$ is uniform over $Y_b$. Since $|Y_b| = M_i/2^l$, we have $\theta(p,q) = \indicate[h_i(p_i) \in Y_b] \cdot 2^l / M_i$. Symmetrically, if $q_i < p_i = x$ then $\theta(p,q) = \indicate[h_i(q_i) \in Y_b] \cdot 2^l / M_i$.
    \end{itemize}
\end{itemize}
Observe that the bit $b$ only matters in the last case of $\theta(p,q)$, so for the sake of comparing the conditional expectations under different choices of $b$, it suffices to compute $\sum_{(p,q) \in S \cup S'} \gamma(p,q)$ where
\begin{align*}
    S \ALIGN:= \left\{ (p,q) \;\;\middle|\;\; \begin{aligned}
        & q \in Q,\;
        p \in P \setminus \set{q}\\
        & \forall j < i,\; h_j(p_j) = h_j(q_j)\\
        & p_i < q_i = x,\;
        h_i(p_i) \in Y_b
    \end{aligned}\right\}
    \AND
    S' \ALIGN:= \left\{ (p,q) \;\;\middle|\;\; \begin{aligned}
        & q \in Q,\;
        p \in P \setminus \set{q}\\
        & \forall j < i,\; h_j(p_j) = h_j(q_j)\\
        & q_i < p_i = x,\;
        h_i(q_i) \in Y_b
    \end{aligned}\right\}.
 \end{align*}

We will focus on computing the sum over $S$; the sum over $S'$ can be computed analogously. To this end, for every $J \subseteq [d]$ we define a subset $S_J := \Set{(p,q) \in S : \eq(p,q) = J}$. Then
\[
    \sum_{(p,q) \in S} \gamma(p,q)
    = \sum_{J \subseteq [d]} \sum_{(p,q) \in S_J} \gamma(p,q)
    = \sum_{J \subseteq [d]} |S_J| \cdot \prod_{\substack{j > i \\ j \notin J}} \frac{1}{M_j}.
\]
Hence the problem reduces to computing $|S_J|$ for each $J \subseteq [d]$. This can be done as follows. At the start of iteration $i$, we initialize a $(d+i)$-dimensional range searching data structure, and then proceed with the inner iterations. As soon as we determine $h_i(x)$ at the end of iteration $x$, we insert the point $(p_1, \dots, p_d, h_1(p_1), \dots, h_i(p_i))$ into the data structure, for each $p \in P$ with $p_i = x$. To compute $|S_J|$ at iteration $(i,x,l)$, we enumerate all $q \in Q$ with $q_i = x$ and aim to count the number of $p \in P$ such that $(p,q) \in S_J$. Laying out all the constraints explicitly, we require
\begin{itemize}
    \item $h_j(p_j) = h_j(q_j)$ for all $j < i$;
    \item $h_i(p_i) \in Y_b$;
    \item $p_i < x$;
    \item $p_j = q_j$ for all $j \in J$;
    \item $p_j \neq q_j$ for all $j \notin J$.
\end{itemize}
Recall that $q$ is fixed and $p \in P$ is a free variable. All but the last constraint are orthogonal range constraints on $(p_1, \dots, p_d, h_1(p_1), \dots, h_i(p_i))$. The last constraint, on the other hand, can be compiled to the disjunction of $2^{d-|J|}$ orthogonal range constraints by trying all possible comparisons between $p_j$ and $q_j$. Therefore, we can extract the count by querying the orthogonal range searching data structure $2^{d-|J|} = O(1)$ times.

Let us bound the running time. In each iteration $i$, we insert and query the data structure $O\left(\sum_{x \in [N_i]} |P[x]|\right) = O(|P|)$ times, where $P[x] = \set{p \in P : p_i = x}$. Each insertion and query takes $\tOh(1)$ time. Therefore, the overall running time is $\tOh\left(|P| + \sum_{i \in [d]} N_i \right)$.
\end{proof}

Now we turn to study the structure of a point set $P$. Our plan is to classify points $q \in P$ into a small number of types based on the measures $|\set{p \in P : p_I = q_I}|$ for $I \subseteq [d]$. We say that a type is ``captured'' if it results in a volume bound of $\tOh(|P|)$ in \Cref{lem:isolation}. The main challenge is in showing that most points have captured types, which we argue via projection inequalities and LP duality. 

\begin{definition}
    Let $P \subseteq \prod_{i \in [d]} [N_i]$ be a point set and denote $\lambda := d 2^d \log \left( 2 (\log |P| + 2) \right)$. A \EMPH{type} is a tuple $l = (l_I : I \subseteq [d])$ where $l_{[d]} = 0$ and $l_I \in \set{0, \dots, \ceil{\log |P|}}$ for $I \subset [d]$. The type is \EMPH{captured} if there exist $x_1, \dots, x_d \geq 0$ such that
    \begin{enumerate}[(i)]
        \item $\sum_{i \in [d]} x_i \leq \log |P| + \lambda$;
        \item $\sum_{i \notin I} x_i \geq l_I$ for all $I \subseteq [d]$.
    \end{enumerate}
    
    For a point $q \in P$, we let $\EMPH{C_I(q)} := \Card{\set{p \in P : p_I = q_I}}$ and associate $\EMPH{\type(q)} := \left( \ceil{\log C_I(q)}: I \subseteq [d] \right)$.
\end{definition}

\begin{lemma} \label{lem:contraction-captured}
There is a deterministic algorithm that, given a point set $P$ and a captured type $l$ along with the witnesses $x_1, \dots, x_d$, computes a family $\mathcal H$ of $O(\log |P|)$ contractions, each with volume $\tOh(|P|)$, such that all points in $P$ of type $l$ are isolated from $P$ under $\mathcal{H}$. The algorithm runs in time $\tOh\left( |P| + \sum_{i \in [d]} N_i \right)$.
\end{lemma}
\begin{proof}
    First we compute $Q := \set{ q \in P : \type(q) = l}$. To this end, we build an orthogonal range searching data structure on $P$. For each $q \in P$ and $I \subseteq [d]$ we can compute $C_I(q)$ by querying the data structure, and as a result we also obtain $\type(q)$. Therefore, $Q$ can be computed in time $\tOh(|P|)$.
    
    Next we let $n_i = 2^{\ceil{x_i}}$ for $i \in [d]$. By condition (ii), for all $q \in Q$ and $I \subseteq [d]$ we have
    \[
        C_I(q) \leq 2^{l_I}
        \leq 2^{\sum_{i \notin I} x_i}
        \leq \prod_{i \notin I} n_i.
    \]
    We may therefore apply \Cref{lem:isolation} to compute a family $\mathcal H$ of $O(\log |Q|)$ contractions that isolates $Q$ from $P$. Each contraction has volume $O\left(\prod_{i \in [d]} n_i\right) \leq O\left(2^{\sum_{i \in [d]} x_i}\right) \leq O(2^{\log |P| + \lambda}) = \tOh(|P|)$ by property~(i). This takes time $\tOh\left(|P|+\sum_{i \in [d]} N_i\right)$.
\end{proof}

The remaining task is to argue that most points have captured types. We need the following lemma known as Shearer's inequality (see \cite{CGFS86} for the original formulation, and e.g., \cite{Shearer} for the distributional version).

\begin{lemma}[Shearer]
    \label{lem:shearer}
    Let $q = (q_1, \dots, q_d)$ be a discrete random variable. If $\mathcal{I}$ is a probability distribution over subsets of $[d]$ such that $\Pr_{I \sim \mathcal{I}}(i \in I) \geq r$ for all $i \in [d]$, then $\E_{I \sim \mathcal{I}}[\entropy(q_I)] \geq r \entropy(q)$.
\end{lemma}

\begin{corollary}
    \label{cor:shearer}
    Let $Q \subseteq \Z^d$. If $\mathcal{I}$ is a probability distribution over subsets of $[d]$ so that $\Pr_{I \sim \mathcal{I}}(i \in I) \geq r$ for all $i \in [d]$, then $\E_{I \sim \mathcal{I}}[\log |Q_I|] \geq r \log |Q|$, where $Q_I = \set{q_I : q \in Q}$ is the projection of $Q$ to the coordinates in $I$.
\end{corollary}

\begin{proof}
    Let $q$ be a uniformly random point from $Q$. Its entropy is $\entropy(q) = \log |Q|$. Its projection $q_I$ is a random variable taking values in $Q_I$, and we can trivially bound $\entropy(q_I) \leq \log |Q_I|$. Therefore,
    \begin{equation*}
        \E_{I \sim \mathcal{I}}[\log |Q_I|] \geq \E_{I \sim \mathcal{I}}[\entropy(q_I)] \geq r \entropy(q) = r \log |Q|,
    \end{equation*}
    where we applied \Cref{lem:shearer} in the second step.
\end{proof}

We also recall Farkas' Lemma of linear programming (see e.g. \cite[Proposition 6.4.3]{GM07}): If a linear system $A\tup{x} \leq \tup{b}$ does not have a solution $\tup{x} \geq \tup{0}$, then there exists $\tup{y} \geq \tup{0}$ such that $\tup{y}^\mathrm{T} \tup{b} < 0$ and $\tup{y}^{\mathrm{T}} A \geq \tup{0}$. Now we are ready to show the following:

\begin{lemma}
    \label{lem:capture}
    There are at most $|P|/4$ points $p \in P$ for which $\type(p)$ is uncaptured.
\end{lemma}
\begin{proof}
    Let $l = (l_I : I \subseteq [d])$ be an uncaptured type, and let $Q := \Set{q \in P: \type(q) = l}$. Our goal is to upper bound $|Q|$ in terms of $|P|$.
    
    Consider an arbitrary $I \subseteq [d]$. By definition, each point $q_I \in Q_I$ extends to at least $C_I(q) \geq 2^{l_I - 1}$ points in $P$. Therefore, $|P| \geq |Q_I| \cdot 2^{l_I - 1}$. Taking logarithm on both sides, we obtain
    \begin{equation}
        \label{eq:type-projection}
        \forall I \subseteq [d], \;\; \log |Q_I| \leq \log |P| + 1 - l_I.
    \end{equation}
    
    Next we relate $\log |Q|$ with $\log |Q_I|$. Recall that $l$ is uncaptured, so the linear system (i)(ii) does not have a solution $x_1, \dots, x_d \geq 0$. By Farkas' Lemma, there exist $(y_I)_{I \subseteq [d]} \geq \mathbf{0}, z \geq 0$ such that
    \begin{align*}
        (\log |P| + \lambda) \cdot z -\sum_{I \subseteq [d]} l_I y_I \ALIGN< 0
        \AND
        \forall i \in [d], \;\; z - \sum_{I \subseteq [d] \setminus \set{i}} y_I \ALIGN\geq 0.
    \end{align*}
    We may assume $y_{[d]} = 0$ since $l_{[d]} = 0$. By moving the terms and scaling the variables, we can further ensure $(y_I)_{I \subseteq [d]} \geq \mathbf{0}, z = 1$ and
    \[
        \sum_{I \subseteq [d]} l_I y_I > \log |P| + \lambda
        \quad\text{and}\quad
        \forall i \in [d],\;\; \sum_{I \subseteq [d] \setminus \set{i}} y_I \leq 1.
    \]
    Let us define a probability distribution $\mathcal{I} (I) := y_I / M$ with normalizing constant $M := \sum_{I \subset [d]} y_I$. Observe that $1 < M \leq d$. Indeed, since $l_I \leq \log |P| + 1$ for all $I \subseteq [d]$, the first constraint implies $M > 1$. On the other hand, $M \leq \sum_{i \in [d]} \sum_{I \subseteq [d] \setminus \set{i}} y_I \leq \sum_{i \in [d]} 1 = d$ by the second constraint.
    
    Rephrasing the constraints in terms of this distribution, we have
    \[
        \E_{I \sim \mathcal{I}}(l_I) > \frac{\log |P| + \lambda}{M}
        \quad\text{and}\quad
        \forall i \in [d],\;\; \Pr_{I \sim \mathcal{I}}(i \notin I) \leq \frac{1}{M},
    \]
    In particular, $\Pr_{I \sim \mathcal{I}}(i \in I) \geq 1 - \frac{1}{M}$, so Corollary \ref{cor:shearer} implies that $\left( 1 - \frac{1}{M} \right) \log |Q| \leq \E_{I \sim \mathcal{I}}[\log |Q_I|]$. Plugging (\ref{eq:type-projection}) into this bound, we get
    \begin{align*}
        \left( 1 - \frac{1}{M} \right) \log |Q|
        &\leq \log |P| + 1 - \E_{I \sim \mathcal{I}}(l_I)\\
        &< \log |P| + 1 - \frac{\log |P| + \lambda}{M}\\
        &= \left( 1 - \frac{1}{M} \right) \log |P| + \frac{M-\lambda}{M}.
    \end{align*}
    Rearranging, we obtain $\log \frac{|Q|}{|P|} \leq \frac{M-\lambda}{M-1}$. Recall that $1 < M \leq d$ and $\lambda := d 2^d \log \left( 2 (\log |P| + 2) \right)$, so the numerator $M-\lambda$ is negative and the denominator $M-1$ is positive. Hence
    \[
        \log \frac{|Q|}{|P|}
        \leq \frac{d-\lambda}{d-1}
        \leq -2^d \log \left( 2 (\log |P| + 2) \right).
    \]
    In other words, $|Q| \leq |P| / \left( 2 (\log |P| + 2) \right)^{2^d}$.
    
    The argument works for any uncaptured type $l$. There are at most $(\log |P| + 2)^{2^d}$ types in total. So by a union bound, we conclude that at most $|P|/2^{2^d} \leq |P|/4$ points have uncaptured types.
\end{proof}

\begin{proof}[Proof of \Cref{thm:contract}]
    For each type $l$, we solve the linear system (i)(ii) by the simplex method (which takes constant time as the number of variables and constraints are both constant). If the system has a solution $x_1, \dots, x_d$, then $l$ is captured and we call \Cref{lem:contraction-captured} to compute a family of $\tOh(1)$ contractions $\mathcal H_l$ with volume $\tOh(|P|)$, isolating all points of type $l$ from~$P$. In the end, output $\mathcal H := \bigcup_l \mathcal H_l$; note that $|\mathcal H| = \tOh(1)$. By \Cref{lem:capture}, all but $|P|/4$ points in~$P$ have a captured type and thus are isolated from $P$ under $\mathcal H$. The running time is clearly as claimed.
\end{proof}

\subsection{Final Assembly}
\begin{theorem}
    \label{thm:star-counting-bounded-degree}
    Assume that $N_1, \dots, N_d$ are powers of two, and that every $v_0 \in [N_0]$ has degree at most~$\Delta$. There is an algorithm \algocount{$\Delta, \mu_1, \dots, \mu_d$} that solves the counting problem on such instances in time $\tOh\left(m + (N_0 \Delta)^2 t^{1-\alfa} + t \right)$.
\end{theorem}

\begin{proof}
The algorithm \algocount{$\Delta, \mu_1, \dots, \mu_d$} works recursively. Throughout we write $D := 4\Delta \cdot (\log m)^c$ for some constant $c > 0$ determined later.

In the base case $N_1 = \cdots = N_d = 1$, we simply compute the number $\sum_{v_0 \in [N_0]} \prod_{i \in [d]} \mu_i(v_0, 1)$, and return it if it is positive.

Otherwise, we fix a dimension $j \in [d]$ with $N_j \geq 2$. The idea is similar to that of sparse matrix multiplication (\Cref{lem:sparse-mm}): Fold the space in half and solve the subinstance recursively; unfold the solutions from the subinstance into a superset of solutions in the original instance; and finally, isolate and recover the actual solutions.

In detail, let $\phi : \N \to \N$ be the permutation that swaps each odd number $2x-1$ with the even number $2x$; for example, $\phi(5) = 6$ and $\phi(6) = 5$. For $\tup{v} = (v_1, \dots, v_d) \in \prod_{i \in [d]} [N_i]$, we denote $\underline{\tup{v}} := (v_1, \dots, v_{j-1}, \ceil{v_j/2}, v_{j+1}, \dots, v_d)$ and $\partial \tup{v} := (v_1, \dots, v_{j-1}, \phi(v_j), v_{j+1}, \dots, v_d)$. Note that $\underline{\tup{v}} = \underline{\partial \tup{v}}$. We proceed as follows:
\begin{itemize}
    \item Construct a function $\eta_j : [N_0] \times [N_j/2] \to \N$ by $\eta_j(v_0, u_j) := \mu_j(v_0, 2u_j-1) + \mu_j(v_0, 2u_j)$. For $i \neq j$, let $\eta_i := \mu_i$.
    \item Compute $(\mathcal{R},\eta) := \text{\algocount{$\Delta, \eta_1, \dots, \eta_d$}}$ recursively.
    \item Let $S := \Set{\tup{v} \in \prod_{i \in [d]} [N_i] \;:\; \underline{\tup{v}} \in \mathcal{R} }$.
    \item Initialize a set $P := S$ and an all-zero function $f : S \to \N$. Repeat until $P = \emptyset$:
    \begin{itemize}
        \item Compute $\mathcal{H} := \text{\algocontract{$P$}}$.
        \item For each $h = (h_1, \dots, h_d) \in \mathcal H$ whose codomain sizes $n_1, \dots, n_d$ are $(N_0,D)$-bounded:
        \begin{enumerate}[(1)]
            \item Compute the set $Q \subseteq P$ of tuples that are isolated from $P$ under $h$.
            \item For each $i \in [d]$, construct $\mu_i' : [N_0] \times [n_i] \to \N$ by
            $\mu_i'(v_0, a_i) := \sum_{v_i : h_i(v_i) = a_i} \mu_i(v_0, v_i)$.
            \item Compute $\mu' := \text{\algodcount{$D, \mu_1', \dots, \mu_d'$}}$.
            \item For each $\tup{v} \in Q$, assign
            $f(\tup{v}) := \mu'(h(\tup{v})) - \sum_{\tup{w} \in S \setminus P \,:\, h(\tup{w}) = h(\tup{v})} f(\tup{w})$
            and
            $f(\partial \tup{v}) := \eta(\underline{\tup{v}}) - f(\tup{v})$.
            \item For each $\tup{v} \in Q$, remove both $\tup{v}$ and $\partial \tup{v}$ from $P$.
        \end{enumerate}
    \end{itemize}
    \item Return $\set{\tup{v} \in S : f(\tup{v}) > 0}$, along with $f$ restricted to this set.
\end{itemize}

Let us establish some useful facts. Let $\sol := \{\tup{v} \in \prod_{i \in [d]} [N_i] : \mu(\tup{v}) > 0\}$ denote the set of solutions. Observe that if $\tup{v} \in \sol$, then $\eta(\underline{\tup{v}}) = \mu(\tup{v}) + \mu(\partial \tup{v}) > 0$, thus $\underline{\tup{v}} \in \mathcal{R}$. Conversely, if $\underline{\tup{v}} \in \mathcal{R}$, then at least one of $\tup{v}$ and $\partial \tup{v}$ is in $\sol$. Therefore, $\sol \subseteq S$ and $|S| = 2|\mathcal{R}| \leq 2 |\sol|$.

\paragraph{Correctness}
Now we are ready to analyze the repeat-until loop. We inductively show that after each round, $f(\tup{w}) = \mu(\tup{w})$ for all $\tup{w} \in S \setminus P$. To this end, consider step (4). We have
\begin{align*}
    \mu'(h(\tup{v}))
    &= \sum_{v_0 \in [N_0]} \prod_{i \in [d]} \mu_i'(v_0, h_i(v_i)) \\
    &= \sum_{v_0 \in [N_0]} \prod_{i \in [d]} \sum_{w_i: h_i(w_i) = h_i(v_i)} \mu_i(v_0, w_i) \\
    &= \sum_{v_0 \in [N_0]} \sum_{\tup{w} : h(\tup{w}) = h(\tup{v})} \prod_{i \in [d]} \mu_i(v_0, w_i) \\
    &= \sum_{\tup{w} : h(\tup{w}) = h(\tup{v})} \sum_{v_0 \in [N_0]} \prod_{i \in [d]} \mu_i(v_0, w_i) \\
    &= \mu(\tup{v}) + \sum_{\substack{\tup{w} \neq \tup{v}\\h(\tup{w}) = h(\tup{v})}} \mu(\tup{w}).
\end{align*}
Recall that $\mu(\tup{w}) = 0$ for $\tup{w} \notin \sol \subseteq S$. Also recall that the tuple $\tup{v} \in Q$ is isolated from $P$, so it can only collide with points in $S \setminus P$. Hence,
\[
    \sum_{\substack{\tup{w} \neq \tup{v}\\h(\tup{w}) = h(\tup{v})}} \mu(\tup{w})
    = \sum_{\substack{\tup{w} \in S \setminus P\\h(\tup{w}) = h(\tup{v})}} \mu(\tup{w})
    = \sum_{\substack{\tup{w} \in S \setminus P\\h(\tup{w}) = h(\tup{v})}} f(\tup{w})
\]
where the last equality used the induction hypothesis. Therefore, the algorithm  assigns $f(\tup{v}) = \mu(\tup{v})$. Moreover, recall that $\eta(\underline{\tup{v}}) = \mu(\tup{v}) + \mu(\tup{\partial \tup{v}})$ by the definition of $\eta$, so the algorithm assigns $f(\partial \tup{v}) = \mu(\partial \tup{v})$. The induction is complete.

When the algorithm terminates, we have $f(\tup{w}) = \mu(\tup{w})$ for all $\tup{w} \in S$. Since $\sol \subseteq S$, the algorithm is correct (if it terminates).

\paragraph{Parameters in the recursion}
The recursion has depth $O(\sum_{i \in [d]} \log N_i) = \tOh(1)$ because each level halves $N_j$ for some $j$. Also, note that the parameters $t,m,\Delta$ do not increase as we descend in the recursion. Now we focus on analyzing a fixed level of the recursion.

\paragraph{Time outside the loop}
Outside the repeat-until loop, constructing the functions $\eta_i$ takes time $O(m)$, computing $S$ takes time $O(|\mathcal{R}|) = O(t)$, and the final return takes time $O(t)$.

\paragraph{Time of each round}
Consider one round of the loop. Computing $\mathcal{H}$ takes time $\tOh\left(|P| + \sum_{i \in [d]} N_i \right) = \tOh(|P| + m)$ by \Cref{thm:contract}. By the same theorem, the family $\mathcal H$ consists of $\tOh(1)$ contractions, each of volume $\tOh(|P|)$. For each iteration $h \in \mathcal H$, step (1) takes time $\tOh(|P|)$ by querying the image of each tuple in a binary search tree containing all images; step (2) takes time $O(m)$; step (3) takes time $\tOh\left(\prod_{i \in [d]} n_i + (N_0 D)^2 (\prod_{i \in [d]} n_i)^{1-\alfa}\right)$ by \Cref{lem:dense-star-counting}; step (4) takes time $\tOh(|Q| + |S \setminus P|)$ by precomputing $\sum_{\tup{w} \in S \setminus P:h(\tup{w})=\tup{a}} f(\tup{w})$ for all possible images $\tup{a}$; step (5) takes time $\tOh(|P|)$. Recall that \smash{$\prod_{i \in [d]} n_i$} is the volume of $h$, so it is at most $\tOh(|P|)$. Gathering all these bounds and recalling $D = \tOh(\Delta)$ and $|P|,|Q| \leq |S| \leq 2t$, a single round runs in time $\tOh(m + (N_0 \Delta)^2 t^{1-\alfa} + t)$.

\paragraph{Number of rounds}
Initially we have $P = S = \{\tup{v} \in \prod_{i \in [d]} [N_i] \;:\; \underline{\tup{v}} \in \mathcal{R} \}$, and step (5) always removes the couple $\tup{v},\partial\tup{v}$ together from $P$. Therefore, at any point in time we have $P = \{\tup{v} \in \prod_{i \in [d]} [N_i] \;:\; \underline{\tup{v}} \in R \}$ for some subset $R \subseteq \mathcal{R}$. Since at least one tuple among each couple is a solution, we have the invariant that $|P \cap \sol| \geq |R| = |P|/2$.

By \Cref{thm:contract}, all but $|P|/4$ tuples in $P$ are isolated from $P$ under $\mathcal H$. In particular, at least $|P \cap \sol| - |P|/4 \geq |P| / 4$ tuples in $P \cap \sol$ are isolated from $P$ under $\mathcal H$. Hence by the pigeonhole principle, there is a function $h \in \mathcal H$ that isolates at least $|P| / (4 |\mathcal H|)$ tuples in $P \cap \sol$ from $P$. We claim that the codomain sizes $n_1, \dots, n_d$ of this function are $(N_0,D)$-bounded.

Let $\xi, \zeta > 0$ be sufficiently large constants such that $|\mathcal{H}| \leq (\log |P|)^\xi$ and $\prod_{i \in [d]} n_i \leq |P| \cdot (\log |P|)^\zeta$. Consider $I \subseteq [d]$. Note that every $v_0 \in [N_0]$ still has degree at most $\Delta$ in the subinstance $(\mu_1', \dots, \mu_d')$. In particular, the number of solutions $t'$ in the subinstance can be bounded by $t' \leq N_0 \Delta^{|I|} \prod_{i \not\in I} n_i$. On the other hand, $t' \geq |P| / (4 |\mathcal H|)$ because $h$ isolates at least this many solutions in the original instance, and by the definition of isolation each original solution leads to a distinct solution in the subinstance. Putting both bounds together,
\begin{equation*}
    \prod_{i \not\in I} n_i
    \geq \frac{|P|}{4 |\mathcal H| \cdot N_0 \Delta^{|I|}}
    \geq \frac{|P|}{4 (\log |P|)^\xi \cdot N_0 \Delta^{|I|}}.
\end{equation*}
As a result,
\begin{align*}
    \prod_{i \in I} n_i = \frac{\prod_{i \in [d]} n_i}{\prod_{i \notin I} n_i}
    \ALIGN\leq \frac{|P| \cdot (\log |P|)^\zeta}{|P| / \left(4 (\log |P|)^\xi \cdot N_0 \Delta^{|I|} \right)} \NEWLINE
    \ALIGN= 4 N_0 \Delta^{|I|} \cdot (\log |P|)^{\xi+\zeta}.
\end{align*}
Choosing $D := 4 \Delta \cdot (\log |P|)^{\xi + \zeta}$, we obtain $\prod_{i \in I} n_i \leq N_0 D^{|I|}$ for all $I \subseteq [d]$. (For $I = \emptyset$ the bound is trivial.) This establishes the claim that $n_1, \dots, n_d$ are $(N_0,D)$-bounded.

The claim implies that each round removes at least $|P|/(4 |\mathcal H|)$ tuples from $P$, thus $|P|$ decreases by a factor of $1 - \frac{1}{4 |\mathcal{H}|} \leq 1 - \frac{1}{4 (\log |P|)^\xi}$. Hence the number of rounds is $O\left( (\log |P|)^{\xi + 1} \right) = \tOh(1)$.

\medskip
Putting all the bounds together, the proof of the theorem is now complete.
\end{proof}

\begin{proof}[Proof of \Cref{thm:star-counting}]
    By \Cref{lem:degree-uniform}, it suffices to design an algorithm for $\Delta$-near-regular instances. Given a $\Delta$-near-regular instance $\mu_1, \dots, \mu_d$, we pad $N_1, \dots, N_d$ to the next powers of two, which blows up the running time by at most a constant factor. By $\Delta$-near-regularity, every $v_0 \in [N_0]$ has degree at most $\Delta$. Therefore, the assumptions of \Cref{thm:star-counting-bounded-degree} are satisfied, and we can apply \algocount{$\Delta, \mu_1, \dots, \mu_d$} to solve the problem in time $\tOh(m + (N_0\Delta)^2 t^{1-\alfa} + t)$. Note that $m \geq N_0\Delta/2$ since every $v_0 \in [N_0]$ has degree at least $\Delta/2$. Hence the time can be further bounded by $\tOh(m^2 t^{1-\alfa} + t)$. Finally, we apply the AM-GM inequality $x^\alfa y^{1-\alfa} \leq \alfa x + (1-\alfa)y$ to bound the time by $\tOh(m^{2/\alfa} + t)$.
\end{proof}

\section{Conditional Lower Bounds}
\label{sec:lower-bounds}

The starting point of our lower bounds is the $\clique{k, n_1, \dots, n_k}$ problem. Here, the input is a $k$-partite graph $G = (V_1 \cup \cdots \cup V_k, E)$ where $|V_i| = n_i$ for all $i \in [k]$, and the goal is to decide if $G$ contains a $k$-clique. Clearly, any $k$-clique in $G$ (if one exists) must pick exactly one vertex from each part.

For $n_1 = \ldots = n_k = n$, the fastest known algorithm for the problem works by grouping the $k$ parts into 3 roughly balanced groups, which induce a tripartite graph on $n^{k/3+O(1)}$ vertices, then running a triangle detection algorithm on this tripartite graph. Since triangle detection can be solved in matrix multiplication time, this gives an $n^{\omega k/3 + O(1)}$-time algorithm for $\clique{k, n, \dots, n}$. If only ``combinatorial'' algorithms are allowed (i.e., fast matrix multiplication is disallowed), then the state of the art is essentially brute-force. This gives rise to the following frequently used hypotheses in fine-grained complexity; see e.g.~\cite{AbboudBW15a,AbboudBBK17,BringmannGL17}.

\begin{unnumbered}[Combinatorial $k$-Clique Hypothesis]
    Let $k \geq 3$ and $\delta > 0$. No combinatorial algorithm can solve $\clique{k,n,\dots,n}$ in time $O(n^{k-\delta})$.
\end{unnumbered}

\begin{unnumbered}[$k$-Clique Hypothesis]
    Let $k \geq 3$ and $\delta > 0$. No algorithm can solve $\clique{k,n,\dots,n}$ in time $O(n^{\omega k/3-\delta})$.
\end{unnumbered}

Next we pose a new hypothesis that has not been previously used. 

\begin{unnumbered}[Lopsided 4-Clique Hypothesis]
    Let $\beta \geq 1$ and $\delta > 0$. No algorithm can solve $\clique{4,n,n,n,n^\beta}$ in time $O(n^{2+\beta-\delta})$.
\end{unnumbered}

The known algorithms solve $\clique{4,n,n,n,n^\beta}$ by building a (possibly dense) $n^2 \times n^\beta$ matrix, and thus take time $\Omega(n^{2+\beta})$. The hypothesis formalizes this running time barrier, postulating that no algorithm can be faster. Naturally, the level of trust in a new hypothesis is significantly lower. Nevertheless, we use it as a barrier towards faster listing algorithms. 

\medskip
After discussing the relevant hypotheses, we now present our main reduction. This is a generic reduction from $\clique{k,n_1,\dots,n_k}$ to listing projected-$(k-1)$-stars.

\begin{lemma}
    \label{lem:lower-bound-generic}
    Let $k \ge 3$. 
    \begin{itemize}
        \item If there is an algorithm for Colored projected-$(k-1)$-star-Listing in time $f(n,m,t)$, then there is an algorithm for $\clique{k,n,\dots,n,N}$ in time $O(n^{k-1} + nN) + f(n+N,nN,n^{k-1})$.
        \item If there is an algorithm for Uncolored projected-$(k-1)$-star-Listing in time $f(n,m,t)$, then there is an algorithm for $\clique{k,n,\dots,n,N}$ in time $O((n+N)^{k-1}) + f(n+N,nN,(n+N)^{k-1})$.
    \end{itemize}
\end{lemma}

\begin{proof}
    We denote the projected-$(k-1)$-star by $H = ([k], \set{(i,k) : i \in I})$, $I = [k-1]$.
    Consider the following generic algorithm for $\clique{k,n,\dots,n,N}$.

    Given a graph $G = (V, E)$ with partition $V = V_1 \cup \cdots \cup V_k$, obtain a subgraph $G'$ by keeping only the edges adjacent to $V_k$. Solve Colored or Uncolored projected-$(k-1)$-star-Listing on $G'$ and let $\sol$ be the set of solutions. For each $\tup{v} \in \sol$, if $\tup{v}$ is a $(k-1)$-clique in $G$ then output true and terminate. At the end of the loop, output false.

    In the colored setting, we have $\sol \subseteq V_1 \times \cdots \times V_{k-1}$. A tuple $\tup{v} \in V_1 \times \cdots \times V_{k-1}$ can be completed to a $k$-clique if and only if $\tup{v}$ is a $(k-1)$-clique and $\tup{v} \in \sol$. Hence, the algorithm outputs true if and only if there is a $k$-clique in $G$.

    In the uncolored setting, we have $\sol \subseteq V^{k-1}$, and it might happen that a solution has more than one vertex in the same part. We say that a tuple $(v_1, \dots, v_{k-1}) \in V^{k-1}$ is \emph{aligned} if there exists a permutation $\pi : [k-1] \to [k-1]$ such that $v_i \in V_{\pi(i)}$ for all $i \in [k-1]$. We claim that, for every non-aligned solution $(v_1, \dots, v_{k-1}) \in \sol$, there exists $i \in [k]$ such that $|V_i \cap \set{v_1, \dots, v_{k-1}}| \geq 2$.

    To this end, let $u$ be a common neighbor of $v_1, \dots, v_{k-1}$. We distinguish two cases:
    \begin{itemize}
        \item If $u \in V_{k}$, then $u$ does not have any neighbor in $V_{k}$, thus $v_1, \dots, v_{k-1} \notin V_{k}$. In this case, there is a mapping $f: [k-1] \to [k-1]$ such that $v_i \in V_{f(i)}$ for all $i \in [k-1]$. Since the solution is not aligned, $f$ is not a permutation, so it is not injective and the claim holds.

        \item Otherwise, $u \in \bigcup_{i \in [k-1]} V_i$, so its neighbors $v_1, \dots, v_{k-1}$ have to be in $V_{k}$. Again, the claim holds.
    \end{itemize}
    The claim is established. Now observe that, for every tuple $\tup{v} \in V^{k-1}$, the following statements are equivalent:
    \begin{enumerate}[(1)]
        \item $\tup{v} \in \sol$ and it is a $(k-1)$-clique in $G$.
        \item $\tup{v} \in \sol$ is aligned and a $(k-1)$-clique in $G$.
        \item $(\tup{v},v_k)$ is a $k$-clique in $G$ for some $v_{k} \in V_{k}$.
    \end{enumerate}
    The implications $(1) \Leftarrow (2) \Leftrightarrow (3)$ are clear. To see $(1) \Rightarrow (2)$, suppose for contradiction that $\tup{v} \in \sol$ is not aligned. Then the claim above states that there is a part $V_i$ containing at least two of $v_1, \dots, v_{k-1}$. These two vertices are not adjacent in $G$, so $\tup{v}$ is not a $(k-1)$-clique after all.

    This equivalence implies that the algorithm outputs true if and only if there is a $k$-clique in $G$.

    Regarding time complexity, constructing $G'$ takes time $O(nN)$, solving projected-$(k-1)$-star-Listing on $G'$ takes time $f(n+N,nN,|\sol|)$, and the for-each loop takes time $O(|\sol|)$. In the colored setting, we can bound $|\sol| \leq n^{k-1}$; in the uncolored setting, we can bound $|\sol| \leq ((k-1)n+N)^{k-1} = O((n+N)^{k-1})$. The total running time is thus as stated.
\end{proof}

\LowerBoundCombinatorialTheorem*

\begin{proof}
    First, suppose for contradiction that an algorithm with running time $f(n,t) = \tOh(n^{k-\eps} + t)$ exists. \Cref{lem:lower-bound-generic} gives a combinatorial algorithm for $\clique{k,n,\dots,n}$ that runs in time $O(n^{k-1} + n^2) + f(2n,n^{k-1}) = \tOh(n^{k-1} + n^{k-\eps})$. This violates the Combinatorial $k$-Clique Hypothesis.

    \medskip
    Next, suppose for contradiction that an algorithm with running time $f(m,t) = \tOh(m^{k-1-\eps} + t)$ exists. We consider the following combinatorial algorithm for $\clique{k,n,\dots,n}$:
    \begin{itemize}
        \item Evenly partition $V_k = V_k^1 \cup \cdots \cup V_k^r$ where $r := n^{1-\eps/k}$.
        \item For each $l \in [r]$, solve $\clique{k,n,\dots,n, \frac{n}{r}}$ on $G^l := G[V_1, \dots, V_{k-1}, V_k^l]$ via \Cref{lem:lower-bound-generic}. If it outputs true, then output true and terminate.
        \item In the end, output false.
    \end{itemize}

    Observe that $G$ contains a $k$-clique if and only if $G^l$ contains a $k$-clique for some $l \in [r]$. Hence the algorithm is correct.

    Each invocation to \Cref{lem:lower-bound-generic} takes time $O(n^{k-1}) + f(\frac{n^2}{r}, (n+\frac{n}{r})^{k-1}) = \tOh(n^{k-1} + (\frac{n^2}{r})^{k-1-\eps})$. Over all $r = n^{1-\eps/k}$ iterations the algorithm runs in time
    \[ \tOh\left( n^{k-\eps/k} + n^{1-\eps/k} \cdot \left(n^{1+\eps/k}\right)^{k-1-\eps} \right) \leq \tOh(n^{k - \eps/k}). \]
    This violates the Combinatorial $k$-Clique Hypothesis.
\end{proof}

\LowerBoundLopsidedTheorem*

\begin{proof}
   Applying \Cref{lem:lower-bound-generic} with $N = n^{1+\eps/3}$, we get an algorithm for $\clique{4,n,n,n,N}$ in time $\tOh\left(n^3 + n^{(1+\eps/3)(3-\eps)}\right) \leq \tOh(n^3)$.
\end{proof}

Observe that the implication of the theorem above would violate the Lopsided 4-Clique Hypothesis. Hence, assuming the hypothesis, no algorithm can solve Colored projected-3-star-Listing in time $\tOh(n^{3-\eps} + t)$ for any $\eps > 0$.

\LowerBoundCliqueTheorem*

\begin{proof}
    Fix $\eps > 0$. Pick $k' := \Ceil{\frac{4\omega^2}{\eps(3-\omega)^2}}$ and $k := \Ceil{\frac{3k'}{\omega}} + 1$. Suppose for contradiction that there exists an algorithm that solves Colored projected-$(k'-1)$-star-Listing in time $f(m,t) = \tOh(m^{\frac{\omega}{3-\omega} - \eps} + t)$. 
    
    We describe an algorithm for $\clique{k, n, \dots, n}$:
    Given $G = (V_1 \cup \cdots \cup V_{k}, E)$, construct a graph $G' = (V_1 \cup \cdots \cup V_{k'-1} \cup U, E')$, where $U$ is the set of all cliques in the graph $G[V_{k'}, \dots, V_{k}]$, and for $i \in [k'-1]$ put an edge between $v_i \in V_i$ and $(v_{k'}, \dots, v_{k}) \in U$ if and only if $(v_i, v_{k'}, \dots, v_{k})$ is a clique. Note that $|U| \leq n^{k-k'+1}$. Apply \Cref{lem:lower-bound-generic} to solve $\clique{k', n, \dots, n, n^{k-k'+1}}$ on $G'$, and output its answer.

    Observe that $G$ has a $k$-clique if and only if $G'$ has a $k'$-clique. Hence the algorithm is correct.

    The invocation to \Cref{lem:lower-bound-generic} takes time $O(n^{k'-1}) + f(n^{k-k'+2}, n^{k'-1}) = \tOh\left(n^{k'-1} + n^{(k-k'+2) (\frac{\omega}{3-\omega} - \eps)}\right)$. Recalling the definitions of $k$ and $k'$, we have
    \begin{align*}
        (k-k'+2) \left( \frac{\omega}{3-\omega} - \eps \right)
        \ALIGN\leq \left( \frac{3-\omega}{\omega}k' + 4 \right)\left( \frac{\omega}{3-\omega} - \eps \right) \NEWLINE
        \ALIGN< k' + \frac{4\omega}{3-\omega} - \frac{3-\omega}{\omega} k'\eps
        \leq k'
    \end{align*}
    where the last step follows from $k' \geq \frac{4\omega^2}{\eps (3-\omega)^2}$. Hence, the running time is $\tOh(n^{k'}) \leq \tOh(n^{\omega (k-1)/3})$. This violates the $k$-Clique Hypothesis.
\end{proof}

\LowerBoundCliqueNTheorem*

\begin{proof}
    Follows from \Cref{thm:lower-bound-clique} since $n \le m$.
\end{proof}

\paragraph*{Acknowledgements}
We thank Nofar Carmeli, Mahmoud Abo Khamis, and Florent Capelli for helpful comments on drafts of this paper. We thank Xiao Hu for pointing out an inaccuracy in an earlier version.

\bibliography{references}

\ifbool{break}{\appendices}{\appendix}

\section{Further Comparison with Prior Work}
\label{sec:subw}
In this appendix, we discuss prior works from \cite{BDG07} and \cite{KNS17,KNS25,KNS26,KC26}, which solve conjunctive queries with a running time that is bounded in terms of the so-called free-connex treewidth and free-connex submodular width, respectively. We show that both algorithms only achieve the trivial running times $O(n^{k})$ and $O(m^{k-1})$ for the projected-$(k-1)$-star.

\subsection{Free-Connex Treewidth}
\begin{definition}
    A \EMPH{tree decomposition} of a hypergraph $H = (V,E)$ is a tree such that
    \begin{itemize}
        \item every tree node $x$ is associated with a subset $\bag(x) \subseteq V$;
        \item for every $e \in E$, there is a node $x$ such that $e \subseteq \bag(x)$;
        \item for every $v \in V$, the nodes $\set{x : v \in \bag(x)}$ form a connected subtree.
    \end{itemize}
\end{definition}

\begin{definition}
    Let $H = (V,E)$ be a hypergraph and $I \subseteq V$. A \EMPH{free-connex tree decomposition} of $(H,I)$ is a tree decomposition of $H$ which contains a connected subtree $T'$ with $\bigcup_{x \in T'} \bag(x) = I$.
\end{definition}

Denote by $\mathcal{T}(H)$ the set of tree decompositions of $H$. The \EMPH{treewidth} of $H$ is defined as
\[ \tw(H) = \min_{T \in \mathcal{T}(H)} \max_{x \in V(T)} |\bag(x)| - 1. \]
Similarly, denote by $T(H,I)$ the set of free-connex tree decompositions of $(H,I)$. The \EMPH{free-connex treewidth} of $(H,I)$ is defined as
\[ \fctw(H,I) = \min_{T \in T(H,I)} \max_{x \in V(T)} |\bag(x)| - 1. \]

It is well known that every tree, in particular the $(k-1)$-star, has treewidth $1$. In contrast, the free-connex treewidth can be much larger:
\begin{proposition}
    \label{prp:fctw}
    Consider the projected-$(k-1)$-star $H = ([k],\, \set{\set{i,k} : i \in [k-1]})$, $I = [k-1]$. Any free-connex tree decomposition of $(H,I)$ has a node $y$ with $[k] \subseteq \bag(y)$. Hence, $\fctw(H,I) = k-1$.
\end{proposition}

\begin{proof}
    Consider a free-connex tree decomposition $T \in T(H,I)$. By definition, it contains a connected subtree $T'$ such that $\bigcup_{x \in V(T')} \bag(x) = I$; in particular, $k \notin \bag(x)$ for all nodes $x \in V(T')$. 
    On the other hand, the third property of tree decomposition implies a connected subtree $T''$ of $T$ where $V(T'') = \set{x : k \in \bag(x)}$. Since $T'$ and $T''$ are disjoint subtrees of $T$, there is a node $y \in V(T'')$ that separates the two subtrees.
    
    For each $i \in I$, the second property of tree decompositions implies the existence of some node $x_i$ with $\set{i,k} \subseteq \bag(x_i)$. Note that $x_i \in V(T'')$, so node $y$ separates it from $T'$. Since $i \in \bigcup_{x \in V(T')} \bag(x)$ and $i \in \bag(x_i)$, we also have $i \in \bag(y)$ by the third property of tree decompositions. Moreover, since $y \in V(T'')$ we have $k \in \bag(y)$. Hence, $[k] \subseteq \bag(y)$.
\end{proof}

The $(H,I)$-enumeration algorithm from \cite{BDG07} has preprocessing time $O(n^{\fctw(H,I) + 1} + m)$ and constant delay. By Proposition~\ref{prp:fctw}, on the projected-$(k-1)$-star this algorithm requires preprocessing time $O(n^{k})$, which is the same as the trivial algorithm.

\subsection{Free-Connex Submodular Width}
Let $H = (V,E)$ be a hypergraph. Let $F(H)$ be the set of functions $h : 2^V \to \R^+$ that satisfy the following properties:
\begin{enumerate}[(i)]
    \item $h(\emptyset) = 0$.
    \item $h(X) \leq h(Y)$ for all $X \subseteq Y \subseteq V$.
    \item $h(X) + h(Y) \geq h(X \cap Y) + h(X \cup Y)$ for all $X,Y \subseteq V$.
    \item $h(e) \leq 1$ for all $e \in E$.
\end{enumerate}

The \EMPH{submodular width} of $H$ is defined as
\[ \subw(H) = \sup_{h \in F(H)} \min_{T \in \mathcal{T}(H)} \max_{x \in T} h(\bag(x)). \]
The \EMPH{free-connex submodular width} of $(H,I)$ is the same restricted to free-connex tree decompositions:
\[ \fcsubw(H,I) = \sup_{h \in F(H)} \min_{T \in T(H,I)} \max_{x \in T} h(\bag(x)). \]

It is known that submodular width is upper bounded by treewidth; see for example \cite[Section 3]{Marx13}. Since the $(k-1)$-star has treewidth $1$, its submodular width is also $1$. In contrast, the free-connex submodular width can be much larger.

\begin{proposition}
    \label{prp:fcsubw}
    For the projected-$(k-1)$-star $H = ([k],\,\set{\set{i,k} : i \in [k-1]})$, $I = [k-1]$, we have $\fcsubw(H,I) = k-1$.
\end{proposition}

\begin{proof}
    We define a function $h : 2^{[k]} \to \R^+$ by $h(X) := |X \setminus \set{k}|$. Clearly $h$ satisfies properties (i) and (ii). It satisfies (iii) since $h(X) + h(Y) = |X \setminus \set{k}| + |Y \setminus \set{k}| = |(X \cup Y) \setminus \set{k}| + |(X \cap Y) \setminus \set{k}| = h(X \cup Y) + h(X \cap Y)$. It satisfies (iv) since for every edge $e = \set{i,k}$, we have $h(e) = 1$. Therefore, $h \in F(H)$.

    For any $T \in T(H,I)$, there exists a node $y \in T$ such that $[k] \subseteq \bag(y)$ by Proposition~\ref{prp:fctw}. In particular, we have $h(\bag(y)) \geq k-1$, which witnesses that $\fcsubw(H,I) \geq k-1$.
\end{proof}

The PANDA algorithm~\cite{KNS17,KNS25} for $(H,I)$-listing runs in time $\tOh(t + m^{\fcsubw(H,I)})$. (The authors called the parameter submodular width and denoted it by $\subw$, but in fact they refer to the free-connex submodular width as clarified in \cite{KNS25}.) By Proposition~\ref{prp:fcsubw}, on the projected-$(k-1)$-star this algorithm takes time $\tOh(t + m^{k-1}) = \tOh(m^{k-1})$, which is the same as the trivial algorithm.

\section{Recovering Matrix Product}
\label{sec:omitted-proofs}
We include a proof of \Cref{lem:recover} for completeness. We start by restating Definition 3.1 and Lemma 3.6 from \cite{ABFK24}. 

\begin{definition}
    Let $S$ be a set of two-dimensional points. We say that $(i,j)$ is \EMPH{isolated} from $S$ under a function $h$ if there is no $j' \neq j$ such that $(i,j') \in S$ and $h(j) = h(j')$.
\end{definition}

\begin{lemma}
    \label{lem:isolation-2d}
    There is a deterministic algorithm that, given $S \subseteq [x] \times [z]$ where $\Card{\set{j : (i,j) \in S}} \leq s$ for all $i \in [x]$, computes in time $\tOh(z + |S|)$ a family $\mathcal H$ of functions $[z] \to [2s]$ with the following properties:
    \begin{enumerate}[1.]
        \item Each point $(i,j) \in S$ is isolated from $S$ under some $h \in \mathcal{H}$.
        \item $|\mathcal{H}| \leq \log |S| + 1$.
    \end{enumerate}
\end{lemma}

Now we are ready to prove \Cref{lem:recover}.

\RecoverProduct*

\begin{proof}    
    The algorithm \algorecover{$A,B,S,\kappa$} works as follows:
    \begin{pseudocode}
        \State remove columns of $B$ that contain zeros only, and let $z$ denote the new dimension
        \For{$i \in [x]$}
            \State compute $z_i := \Card{\set{j \in [z] : (i,j) \in S}}$
        \EndFor
        \For{$l = 0, \dots, \log \kappa$}
            \State let $I_l := \set{i \in [x] : 2^l \leq z_i < 2^{l+1}}$
        \EndFor
        \State initialize $C \in \N^{x \times z}$ as a zero matrix
        \For{$l = 0, \dots, \log \kappa$}
            \State let $A'$ be the matrix $A$ restricted to rows $I_l$
            \State let $S'$ be the set $S$ restricted to rows $I_l$
            \State let $s := 2^{l+1}$; note that $z_i < s$ for all $i \in I_l$
            \State compute a family $\mathcal H$ of functions $[z] \to [2s]$ that isolates points from $S'$ by \Cref{lem:isolation-2d}
            \For{$h \in \mathcal H$}
                \State construct $B' \in \N^{y \times 2s}$ by $B'_{ib} := \sum_{j \in h^{-1}(b)} B_{ij}$
                \State compute $C' := A'B'$
                \For{$(i,j) \in S'$ isolated under $h$}
                    \State $C_{ij} := C'_{i,h(j)}$
                \EndFor
            \EndFor
        \EndFor
        \State \Return $C$ after adding back the removed zero columns
    \end{pseudocode}
    
    We claim that iteration $l$ correctly assigns $C_{ij} = (AB)_{ij}$ for all $(i,j) \in I_l \times [z]$. To this end, consider a point $(i,j) \in I_l \times [z]$. If $(i,j) \notin S'$, then the algorithm never assigns $C_{ij}$, so the entry remains zero till the end. This is correct due to the assumption $\supp(AB) \subseteq S$. If $(i,j) \in S'$, then we recall that $(i,j)$ is isolated from $S'$ under some $h \in \mathcal H$ by \Cref{lem:isolation-2d}. In particular, it is isolated from $\supp(AB)$ under some $h \in \mathcal H$. In the inner loop corresponding to such an $h$, the algorithm assigns
    \begin{align*}
        C_{ij} = C'_{i,h(j)} = (A'B')_{i,h(j)}
        \ALIGN= \sum_{k \in [y]} A'_{ik} B'_{k,h(j)} \NEWLINE
        \ALIGN= \sum_{k \in [y]} A_{ik} \sum_{j' : h(j) = h(j')} B_{kj'} \NEWLINE
        \ALIGN= \sum_{k \in [y]} A_{ik} B_{kj}
        = (AB)_{ij}.
    \end{align*}
    Here, the penultimate step used that $(i,j)$ is isolated from $\supp(AB)$ under $h$. Therefore, the assignment is also correct. This finishes the proof of the claim.
    
    Recall from the premise of the lemma that $z_i \leq \kappa$ for all $i$. Hence, the $I_l$'s form a partition of $[x]$ (except for the zero rows), thus the claim implies that the algorithm returns the correct product $C$.
    
    It remains to analyze the running time. The removal of zero columns takes time $O(|B|)$. Computing $z_i$'s and $I_l$'s takes time $O(x + |S|)$. Now consider iteration $l$ in the main loop, and write $x' := |I_l|$, $z' := 2s = 2^{l+2}$. Notice that $x' \leq x$, $z' \leq 4\kappa$ and $x'z' \leq 2^{l+2} |I_l| = 4 \cdot 2^l |I_l| \leq 4|S|$. Building the matrix $A'$ takes time $O(|A|)$. Computing the family $\mathcal H$ takes time $\tOh(z + |S|) \leq \tOh(|B| + |S|)$. Each iteration of the inner loop takes time $O(|B|) + \mm(x', y, z') + \tOh(|S|)$, so over $|\mathcal{H}| = \tOh(1)$ inner iterations we spend time $\tOh(\mm(x',y,z') + |B| + |S|)$. Summing up these terms, every outer iteration takes time
    \[
        \tOh\Bigg( |A| + |B| + |S| +
        \max_{\substack{x' \leq x, z' \leq 4\kappa,\\ x'z' \leq 4|S|}}
        \mm(x', y, z') \Bigg).
    \]
    As there are $1 + \log \kappa = \tOh(1)$ outer iterations, the total running time is of the same form.
\end{proof}

\end{document}